\documentclass[11pt]{article}
\usepackage{arxiv}
\usepackage{amsmath,amssymb,amsthm,bm,booktabs,graphicx,microtype,xcolor}
\usepackage{float,placeins}
\usepackage[numbers,sort&compress]{natbib}
\usepackage{subcaption}

\usepackage{mathtools}
\graphicspath{{figures/}}
\theoremstyle{plain}
\newtheorem{proposition}{Proposition}[section]
\newtheorem{lemma}[proposition]{Lemma}
\newtheorem{corollary}[proposition]{Corollary}
\theoremstyle{remark}
\newtheorem{remark}[proposition]{Remark}
\definecolor{darkred}{rgb}{0.6, 0.0, 0.0}

\newcommand{\rev}[1]{{\color{red}#1}}

{\par\medskip\begingroup\color{blue}\footnotesize
 \noindent\textit{--- original text (replaced above; kept for comparison) ---}\par}%
{\par\noindent\textit{--- end of original ---}\endgroup\medskip\par}
\definecolor{v8teal}{rgb}{0.0, 0.45, 0.45}

\usepackage{cancel}
\usepackage[normalem]{ulem}
\definecolor{v9blue}{rgb}{0.05,0.20,0.75}
\definecolor{v9green}{rgb}{0.00,0.40,0.15}
\newcommand{\nrev}[1]{{\color{v9blue}#1}}

\definecolor{v10purple}{rgb}{0.50,0.00,0.55}
\newcommand{\xrev}[1]{{\color{v10purple}#1}}

\definecolor{v12magenta}{rgb}{0.80,0.00,0.45}
\DeclareRobustCommand{\yrev}[1]{{\color{v12magenta}#1}}
\DeclareRobustCommand{\delete}[1]{\ifmmode{\color{gray}\cancel{#1}}\else{\color{gray}\sout{#1}}\fi}
\usepackage{hyperref}
\hypersetup{
    colorlinks=true,
    linkcolor=darkred,
    citecolor=darkred,
    urlcolor=blue,
    filecolor=magenta
}
\usepackage{cleveref}
\newcommand{\spopinf}{SPOpInf}
\newcommand{\opinf}{OpInf}
\newcommand{\R}{\mathbb{R}}
\title{Stochastic Physics-Constrained Operator Inference for Complex Systems}
\author{Changhong Mou\\[2pt]
  Department of Mathematics and Statistics\\
  Utah State University, Logan, UT 84322, USA\\
  \texttt{changhong.mou@usu.edu}}
  \date{\today}

\definecolor{red}{rgb}{0,0,0}
\definecolor{blue}{rgb}{0,0,0}
\definecolor{gray}{rgb}{0,0,0}
\definecolor{magenta}{rgb}{0,0,0}
\definecolor{darkred}{rgb}{0,0,0}
\definecolor{v8teal}{rgb}{0,0,0}
\definecolor{v9blue}{rgb}{0,0,0}
\definecolor{v9green}{rgb}{0,0,0}
\definecolor{v10purple}{rgb}{0,0,0}
\definecolor{v12magenta}{rgb}{0,0,0}
\begin{document}
\maketitle

\begin{abstract}
Quadratically nonlinear systems driven by additive noise arise throughout science and engineering, both as stochastic dynamical systems in their own right and as reduced-order models (ROMs) of turbulent flows.  
We propose stochastic physics-constrained operator inference (\spopinf{}), a non-intrusive method that learns a stochastic differential equation, i.e., a quadratic--linear drift \emph{and} a diffusion amplitude, from data.  
Unlike operator inference (OpInf), which fits a deterministic drift by least squares, \spopinf{} maximizes an Euler--Maruyama increment likelihood in which the residual covariance is itself an unknown, so the deterministic and stochastic components are calibrated jointly; it enforces the energy-preserving quadratic structure exactly, through linear equality constraints imposed inside the optimization; and it works in an upper-triangular parametrization that removes the structural non-identifiability of the full quadratic tensor.  
The resulting optimization is solved by an expectation--maximization \yrev{(EM)-based} iterative learning algorithm that updates the drift coefficients and the noise levels in turn: each iteration requires only an inexpensive weighted least-squares solve for the drift and a residual-variance update for the noise.
With five stochastic benchmark systems \yrev{that show intermittency and chaos}, \spopinf{} recovers the deterministic operators more accurately than OpInf, satisfies the imposed constraints to machine precision, accurately estimates the noise amplitudes, and reproduces long-time statistics.  
Applied to a stochastic Burgers equation and a quasi-geostrophic double-gyre flow, the same estimator yields non-intrusive stochastic ROMs that remain stable \yrev{over long times where unconstrained alternatives blow up.}
\end{abstract}
\raggedbottom
\section{Introduction}
Quadratically nonlinear systems driven by noise are among the most common model classes in science and engineering \citep{majda2016turbulent,majdaharlim2012,chen2018conditional}.  The quadratic nonlinearity originates in advection, vortex stretching, and triad interactions among wave modes, and stochastic forcing enters either as physical randomness or as the effective action of unresolved scales on resolved ones \citep{hasselmann1976,majda2001framework,majdawang2006}.  Systems of this type serve two complementary purposes.  
As low-dimensional test problems, from the Lorenz models to the Charney--DeVore and truncated Burgers--Hopf systems, they are laboratories for understanding chaos, intermittency, regime switching, and extreme events in complex systems \citep{lorenz1963,crommelin2004,majda2006slow,farazmand2019}.  
As discretizations and \yrev{reduced-order} models (ROMs) of turbulent flows, they are the computational models themselves: spectral and finite-element discretizations of the Burgers, Navier--Stokes, and quasi-geostrophic equations yield systems of exactly this quadratic form, though with a very large number of degrees of freedom (DOFs), and projection-based ROMs, such as the POD--Galerkin ROM, inherit the same quadratic form in low-dimensional subspaces \citep{holmes1996,noack2003,benner2015,wang2020reduced}.
Learning such systems from trajectory data is therefore a central task in two senses: as system identification, when the goal is to recover the dynamics of an observed stochastic system, and as \yrev{reduced-order} modeling, when the learned equations themselves serve as a non-intrusive \yrev{low-dimensional} surrogate of a complex PDE.

Over the past decade, each of these two tasks has acquired an effective regression-based method.  For system identification, the sparse identification of nonlinear dynamics (SINDy) discovers interpretable governing equations for observed states by sparse regression on a library of candidate terms \citep{brunton2016sindy,rudy2017pdefind}.  
For \yrev{reduced-order} modeling, operator inference (OpInf) learns the operators of a projection-based ROM non-intrusively, by fitting a polynomial model of the form suggested by the governing equations to reduced trajectory data \citep{peherstorfer2016opinf}; OpInf has been extended to general nonlinear PDEs through lifting transformations \citep{qian2020lift,qian2022pdeopinf}, applied to challenging engineering problems such as combustion \citep{mcquarrie2021regularized}, and analyzed for noisy training data \citep{uy2023active}. 
The two methods are closely related: when the SINDy library is restricted to linear and quadratic monomials and the sparsity threshold is switched off, the SINDy regression coincides with the OpInf least-squares problem, up to each method's customary regularization.  
The difference between the two lies in the setting and the goal rather than in the algorithm: SINDy works on observed states with a large candidate library and uses sparsity for equation discovery, whereas OpInf learns the reduced operators non-intrusively from projected trajectory data of a known high-dimensional model.

For the systems above, however, two essential components are often desirable that the standard OpInf formulation was not designed to provide.  
The first is \emph{stochasticity}.  In turbulent flows, the discarded scales act back on the resolved ones as an effectively random forcing, so a deterministic reduced model cannot reproduce the correct long-time statistics; this observation underlies stochastic climate modeling and stochastic parameterization \citep{hasselmann1976,mtv2003,franzke2015,berner2017} and modern stochastic closures for turbulence \citep{tian2021mz,hoekstra2024closure}.  
Stochasticity is equally essential when the system itself is a stochastic dynamical system, and in data assimilation, where every forecast model is imperfect: representing the inevitable model error by a calibrated stochastic term supplies the forecast uncertainty that the filter requires, and is central to the skillful filtering of turbulent signals \citep{majdaharlim2012,chen2018conditional,mou2023multiscale}.
Stochastic identification approaches exist \citep{boninsegna2018stochastic,callaham2021langevin}, learning both a nonlinear drift, that is, the deterministic part of the dynamics, and the noise amplitude from libraries of candidate terms.  
However, because they estimate the drift and the noise by averaging the observed dynamics over small regions of state space, they have in practice been limited to one or a few variables \citep{callaham2021langevin}, and physics enters only through the choice of candidate terms, with no constraint on the identified coefficients. 
The second missing component is \emph{energy conservation}.  In the physical systems above, the quadratic nonlinearity redistributes energy among components without creating or destroying it; a fitted quadratic operator that violates this property, even slightly, can inject energy without bound and blow up in finite time \citep{schlegel2015,kaptanoglu2021trapping}.  
Structure-aware variants of OpInf enforce energy preservation or stability constraints \citep{sawant2023structure,koike2024epopinf,goyal2025stable, gkimisis2026representation}, but they identify deterministic models only.

In this paper, we propose \emph{stochastic physics-constrained operator inference} (\spopinf{}) to close this gap.  \spopinf{} differs from classical OpInf in three ways.  First, OpInf estimates a deterministic drift, whereas \spopinf{} identifies a stochastic differential equation, that is, a drift \emph{and} a diffusion amplitude.  
Second, OpInf solves an ordinary least-squares regression, whereas \spopinf{} maximizes an Euler--Maruyama increment likelihood in which the residual covariance is itself an unknown, so the deterministic and stochastic components are calibrated jointly rather than sequentially. 
Third, \spopinf{} enforces
linear equality constraints to retain energy-preserving quadratic structure
exactly during the likelihood optimization, not as a post hoc correction,
and works in an upper-triangular parametrization that removes the structural
non-identifiability of the full quadratic tensor.  
Computationally, \spopinf{} is an extension of the OpInf regression.  
From the same snapshot data, it estimates time derivatives by finite differences and fits the quadratic--linear model by iteratively alternating two simple steps until convergence, each as cheap as an ordinary least-squares solve: one step estimates the noise level of each equation from the current fitting residuals, and the other refits the coefficients, weighting each equation by its estimated noise level and enforcing the energy constraints exactly.  
In this way, the drift and the noise are updated iteratively, each improving the estimate of the other, until the change in the estimates falls below a prescribed tolerance; the noise amplitude is then read off from the final residuals.  
Every step is an inexpensive least-squares solve, so \spopinf{} keeps nearly the cost of OpInf.
The proposed method applies in two settings.  
In the first, it identifies stochastic dynamical systems from observed trajectories.  
In the second, it builds ROMs for complex PDEs: snapshot data of the PDE solution are projected onto a reduced subspace spanned by POD or Fourier modes, the model is inferred from the resulting time coefficients, and then serves as a non-intrusive stochastic ROM.

{Our contributions are: (i) an equality-constrained joint-likelihood calibration of the deterministic and stochastic components of quadratic SDEs, in which the residual covariance is an unknown of the same optimization; (ii) a compact, provably complete construction of the energy-preserving
constraints in the identifiable upper-triangular parametrization
(\cref{prop:energy-equivalence}); (iii) a systematic formulation whose instances cover both directly observed stochastic systems and ROMs for complex PDEs\yrev{,} so one
estimator serves system identification and stochastic reduced-order modeling; (iv) an investigation across five stochastic systems which show intermittency, chaos, and metastable regime switching, together with two ROMs for PDE applications; and (v) \yrev{sensitivity tests} (\cref{sec:robustness}) quantifying the dependence of the estimates on the sampling interval, the training data, the realization of the driving
noise, and observation noise.}

The paper is organized as follows.  \Cref{sec:problem} states the problem setting for stochastic quadratic systems, from low-dimensional dynamical systems to discretized complex PDEs; \cref{sec:method} develops the \spopinf{} framework, from the operator-inference model through the upper-triangular formulation and the energy-preserving constraints to the stochastic optimization; \cref{sec:numerical} presents the numerical tests in two categories, identifying stochastic dynamical systems (\cref{sec:ident-sds}) and \yrev{reduced-order} modeling for complex PDEs (\cref{sec:pde-applications}); and \cref{sec:conclusion} concludes.

\section{Stochastic quadratic systems: dynamical systems and complex PDEs}
\label{sec:problem}

\subsection{Stochastic quadratic model}
\label{sec:model-class}
We begin by defining the model considered throughout this paper: a stochastic differential equation with quadratic drift, written in a notation used consistently for all later sections.
Throughout, $\bm a(t)\in\R^r$ denotes the resolved state: for a stochastic dynamical system, $\R^r$ is the state space itself and $\bm a$ its observed state; for a complex PDE, the solution is approximated by an expansion $\sum_{i=1}^{r} a_i(t)\,\phi_i(\bm x)$ in basis functions $\phi_i$, \yrev{such as Fourier or POD modes}, and $\bm a=(a_1,\ldots,a_r)^\top$ collects the time-dependent coefficients (\cref{sec:two-settings}).  
In either case, the object to be identified is the following stochastic quadratic model.
We use componentwise tensor notation: let $\bm B_i\in\R^{r\times r}$ be the $i$th slice of the third-order quadratic tensor $\mathcal B$, let $\bm C\in\R^{r\times r}$ be the linear operator, let $\bm D\in\R^r$ be a
constant forcing, and let $\bm\Sigma=\operatorname{diag}(\sigma_1,\ldots,\sigma_r)$ be the amplitude of an $r$-dimensional standard Wiener process $\bm W_t$.  
{The} $i$th component of $\bm a$ evolves as
\begin{equation}
  \mathrm d a_i
  = \left[\bm a^\top\bm B_i\bm a
  +(\bm C\bm a)_i+D_i\right]\mathrm dt
  +\sigma_i\,\mathrm dW_i,
  \qquad i=1,\ldots,r,
  \label{eq:cpopinf-component}
\end{equation}
where the bracketed term is the deterministic component (the drift) and
$\sigma_i\,\mathrm dW_i$ is the stochastic component.  To simplify the
notation, stacking the quadratic terms of all $r$
components defines the vector-valued quadratic operator:
\begin{equation}
  \mathcal Q_{\mathcal B}(\bm a):=\bm a^\top\mathcal B\,\bm a
  =\begin{bmatrix}
  \bm a^\top\bm B_1\bm a & \cdots &
  \bm a^\top\bm B_r\bm a
  \end{bmatrix}^{\!\top}
  .
  \label{eq:quadratic-map}
\end{equation}
The stochastic
quadratic model then takes the vector form
\begin{equation}
  \mathrm d\bm a
  =\left[\bm a^\top\mathcal B\,\bm a
  +\bm C\bm a+\bm D\right]\mathrm dt
  +\bm\Sigma\,\mathrm d\bm W_t.
  \label{eq:cpopinf-vector}
\end{equation}
The stochastic quadratic model \cref{eq:cpopinf-vector} is the finite-dimensional form of the canonical structure of turbulent dynamical systems \citep{majda2016turbulent,majdaharlim2012,chen2018conditional,%
mou2023multiscale}.  
In that setting, the skew-symmetric part of the linear operator $\bm C$ represents dispersion and its negative-definite symmetric part represents dissipation; the constant vector $\bm D$ represents the deterministic forcing; the additive noise $\bm\Sigma\,\mathrm d\bm W_t$ represents either physical stochastic forcing or the effect of unresolved scales on the resolved ones.
Finally, the quadratic interactions transfer energy across scales without creating or destroying it: this energy conservation is also the physics constraint that \spopinf{} enforces exactly during identification (\cref{sec:constraints}).

The inference problem addressed in this paper is then stated as follows: \emph{given a sampled trajectory of $\bm a$, infer the deterministic operators {$(\mathcal B,\bm C,\bm D)$} and the stochastic amplitude $\bm\Sigma$ in \cref{eq:cpopinf-vector}, subject to linear equality constraints on the quadratic tensor.}  Depending on the origin of $\bm a$, this is a system identification problem, when $\bm a$ is the observed state of a stochastic dynamical system, or an operator inference problem, when $\bm a$ contains the ROM coefficients of a complex PDE and the inferred model serves as a stochastic ROM (\cref{sec:two-settings}).  
{In general no part of the drift is assumed known: $\bm D$ is inferred together with $(\mathcal B,\bm C)$ as part of \spopinf{}. Any component \yrev{that \emph{is} known a priori, such as the projected forcing in Setting II or a known linear damping,} can instead be supplied to the method and removed from the inference, as detailed in \cref{sec:data-model}.}

This inference problem is difficult for three reasons.  
First, the energy-conserving quadratic structure is easily violated by regression, and even a slight violation of the fitted quadratic {tensor can} inject energy without bound and cause finite-time blow-up
\citep{schlegel2015,kaptanoglu2021trapping,moore2026learning}.
Second, the systems of interest are chaotic, so individual trajectories cannot be matched for long; an inferred model must instead be evaluated by its long-time statistics, such as the distribution of states and their temporal correlations \citep{majda2016turbulent}, and these statistics are sensitive to small errors in the inferred operators.
Third, phenomena such as intermittency and regime transitions depend quantitatively on the noise, so the noise amplitude requires the same attention in inference as the drift; however, each observed state change over a time step combines a deterministic contribution and a stochastic one, which the data do not separate \citep{majdaharlim2012}.  
{These three difficulties motivate, respectively, the exact energy
constraints of \cref{sec:constraints}, the statistics-based evaluation metrics
used throughout \cref{sec:numerical}, and the joint estimation of the
deterministic and stochastic components in \cref{sec:joint-likelihood}.}
{\color{red}
\begin{table}[H]
\centering
\caption{Notations and acronyms.}
\nrev{\label{tab:notation}}
\small
\begin{tabular}{ll@{\qquad}ll}
\toprule
$\bm a\in\R^r$ & resolved state / reduced coefficients &
$\bm\Sigma,\ \sigma_i$ & \yrev{noise amplitude} \\
$\mathcal B,\ \bm B_i$ & quadratic tensor and its $i$th slice &
$\bm G$ & energy-constraint matrix \\
$\bm C$ & linear operator &
$\bm\theta$ & stacked unknown coefficients \\
$\bm D$ & constant forcing ({inferred, or known and moved to $\bm f_{\rm known}$}) &
$\bm\Phi_n,\ \bm\ell$ & feature matrix / feature vector \\
$\bm\tau$ & ROM closure term &
$\bm R,\ \bm Q_\Delta$ & residual / increment covariance \\
$\bm\phi_i,\ \varphi_i$ & basis functions &
$\varepsilon_{\rm tol}$ & tolerance \\
\bottomrule
\end{tabular}
\end{table}
}
\subsection{Two settings: stochastic dynamical systems and complex PDEs}
\label{sec:two-settings}
The stochastic quadratic model \cref{eq:cpopinf-vector} is used in two settings that differ not in the data, which in both cases is a sampled trajectory of $\bm a$, but in what the model represents: in the first, $\bm a$ is the state of a stochastic dynamical system and \cref{eq:cpopinf-vector} is a model of the system itself; in the second, $\bm a$ collects the ROM coefficients of a complex PDE and \cref{eq:cpopinf-vector} is a stochastic reduced-order model of that PDE.

\paragraph{Setting I: stochastic dynamical systems.}
Here $\bm a$ is the observed state of a finite-dimensional stochastic dynamical system, and \cref{eq:cpopinf-vector} is a model for the system itself.  
This is the classical stochastic system identification setting; the benchmark systems of \cref{sec:ident-sds} are of this type.

\paragraph{Setting II: complex PDEs.}
Here $\bm a$ contains the basis coefficients of a nonlinear PDE solution.
For illustration, consider the incompressible Navier--Stokes equations,
\begin{equation}
  \partial_t\bm u+(\bm u\cdot\nabla)\bm u
  =-\nabla p+\mathrm{Re}^{-1}\Delta\bm u,
  \qquad \nabla\cdot\bm u=0,
  \label{eq:nse}
\end{equation}
where $\bm u$ is the velocity, $p$ is the pressure, and $\mathrm{Re}$ is the Reynolds number. 
The velocity field is approximated by the expansion
\begin{equation}
  \bm u_r(\bm x,t)=\overline{\bm u}(\bm x)
  +\sum_{i=1}^r a_i(t)\bm\phi_i(\bm x),
  \label{eq:pde-pod-expansion}
\end{equation}
where $\overline{\bm u}$ is the mean field and the basis functions $\bm\phi_i$ are, for example, Fourier or POD modes.  
The Galerkin projection substitutes \cref{eq:pde-pod-expansion} into \cref{eq:nse} and takes the $L^2$ inner product $(\cdot,\cdot)$ with each basis function; for divergence-free bases the pressure term vanishes, and the result is a system of $r$ ordinary differential equations for the coefficients: for $i=1,\ldots,r$,
\begin{equation}
  \dot a_i
  =D_{G,i}+(\bm C_G\bm a)_i+\bm a^\top\bm B_{G,i}\,\bm a,
  \label{eq:grom-component}
\end{equation}
with
\begin{subequations}
\label{eq:grom-operators}
\begin{align}
  D_{G,i}&=\big(\mathrm{Re}^{-1}\Delta\overline{\bm u}
  -(\overline{\bm u}\cdot\nabla)\overline{\bm u},\,\bm\phi_i\big),
  \label{eq:grom-D}\\
  (\bm C_G)_{im}&=\big(\mathrm{Re}^{-1}\Delta\bm\phi_m
  -(\overline{\bm u}\cdot\nabla)\bm\phi_m
  -(\bm\phi_m\cdot\nabla)\overline{\bm u},\,\bm\phi_i\big),
  \label{eq:grom-C}\\
  (\bm B_{G,i})_{mn}&=-\big((\bm\phi_m\cdot\nabla)\bm\phi_n,\,
  \bm\phi_i\big).
  \label{eq:grom-B}
\end{align}
\end{subequations}
Because the nonlinearity of \cref{eq:nse} is quadratic, no higher-order
terms arise, and collecting the components gives the Galerkin
reduced-order model (G-ROM) in exactly the form of \cref{sec:model-class},
\begin{equation}
  \dot{\bm a}=\bm f_G(\bm a)
  =\mathcal Q_{\mathcal B_G}(\bm a)+\bm C_G\bm a+\bm D_G.
  \label{eq:pde-grom}
\end{equation}
The same construction applies to any PDE with quadratic nonlinearity,
including the Burgers and quasi-geostrophic equations of
\cref{sec:pde-applications}
\citep{holmes1996,noack2003,benner2015,mou2020reduced,besabe2025linear}.

The G-ROM \cref{eq:pde-grom} is exact only if the solution remains in
the span of the $r$ retained modes.  In turbulent and
convection-dominated flows it does not: for a small $r$, a substantial
fraction of the kinetic energy resides in the discarded modes, and those
modes interact with the retained ones through the quadratic
nonlinearity.  To make this precise, decompose the full-order solution
as $\bm u=\overline{\bm u}+\bm u_r+\bm u'$, where $\bm u_r$ is the part
resolved by the $r$ modes in \cref{eq:pde-pod-expansion} and $\bm u'$ is
the truncated remainder.  Projecting \cref{eq:nse} onto $\bm\phi_i$
without dropping $\bm u'$ gives the exact evolution of the resolved
coefficients,
\begin{equation}
  \dot{\bm a}=\bm f_G(\bm a)+\bm\tau(\bm a,\bm u'),
  \qquad
  \tau_i(\bm a,\bm u')
  =\big(\mathrm{Re}^{-1}\Delta\bm u'
  -(\bm u_r\cdot\nabla)\bm u'
  -(\bm u'\cdot\nabla)\bm u_r
  -(\bm u'\cdot\nabla)\bm u',\,\bm\phi_i\big),
  \label{eq:pde-closure}
\end{equation}
where the closure term $\bm\tau$ collects every interaction that
involves the truncated modes.  The G-ROM is \cref{eq:pde-closure} with
$\bm\tau$ set to zero, and the discrepancy between the G-ROM and the
projected full-order solution is exactly the neglected $\bm\tau$.
Because $\bm\tau$ depends on $\bm u'$, which the ROM does not carry, it
cannot be evaluated within the ROM and must be modeled in terms of the
resolved coefficients $\bm a$ alone; this is the classical closure
problem of \yrev{ROMs} \citep{ahmed2021closures,snyder2022reduced}.

Two properties of $\bm\tau$ in the turbulent regime guide how it should
be modeled.  First, on average it is dissipative: the quadratic term
transfers energy from the large, retained scales to the small, discarded
scales, where viscosity removes it, so that the net effect of $\bm\tau$
on the retained modes is an additional energy dissipation.  The classical remedies model $\bm\tau$ as an eddy
viscosity or a variational-multiscale dissipation acting on $\bm a$
\citep{aubry1988,wang2012,sanmaulik2018,mou2021ddvms}; in the notation of
\cref{sec:model-class}, such a correction modifies the symmetric,
negative-definite part of the linear operator $\bm C$.  Second,
$\bm\tau$ fluctuates: the discarded modes evolve on time scales much
shorter than those of the retained ones, and the Mori--Zwanzig
formalism shows that the exact closure consists of a Markovian
correction to the drift, a memory integral, and an orthogonal-dynamics
term driven by the unresolved initial data that is unpredictable from
$\bm a$ \citep{chorin2000mz,tian2021mz}.  When the discarded modes
decorrelate quickly, the memory integral collapses onto the drift and
the orthogonal term is well approximated by white noise; this is the
central argument of stochastic climate modeling \citep{hasselmann1976}
and of stochastic mode reduction \citep{majda2001framework,mtv2003},
where averaging over the fast modes produces a closed quadratic model
for the slow ones with additional dissipation and additive noise, and it
is the basis of data-driven stochastic closures for  Galerkin ROMs \citep{kondrashov2015,mou2023multiscale,%
hoekstra2024closure}.  In the turbulent setting, therefore, the closure
term is well described by a combination of energy dissipation and white
noise, which is precisely the structure of the additive noise and the
dissipative part of $\bm C$ in the stochastic quadratic model
\cref{eq:cpopinf-vector}.

\subsection{Non-intrusive discrete-time formulation}
\label{sec:data-model}
In both settings of \cref{sec:two-settings}, the only input to the
inference is a sampled trajectory of the resolved state: snapshots
$\bm a_1,\ldots,\bm a_N\in\R^r$ at the uniform time step
$\Delta t$, so that $\bm a_n\approx\bm a(t_n)$ with $t_n=n\,\Delta t$.
In Setting I the $\bm a_n$ are the observed states themselves; in
Setting II they are the projections of the full-order snapshots onto the
basis of \cref{eq:pde-pod-expansion},
$a_{n,i}=(\bm u(\cdot,t_n)-\overline{\bm u},\,\bm\phi_i)$, and this
projection is the only place where the full-order solution enters.  In
particular, neither the Galerkin operators \cref{eq:grom-operators} nor
the closure term $\bm\tau$ of \cref{eq:pde-closure} is ever assembled or
evaluated.  The formulation is therefore non-intrusive in the sense of
operator inference \citep{peherstorfer2016opinf}: it requires no access
to the discretized PDE operators or to the solver that produced the
snapshots, and the same procedure applies verbatim to both settings.
This subsection converts the continuous-time model
\cref{eq:cpopinf-vector} into a regression model that links consecutive
snapshots; this discrete model is what the estimator of
\cref{sec:method} works with.

We first separate the drift of \cref{eq:cpopinf-vector} into a part that
is known a priori and a part to be inferred:
\begin{equation}
  \bm f_{\bm\theta}(\bm a)
  =\bm f_{\rm known}(\bm a)+\bm\Phi(\bm a)\,\bm\theta,
  \qquad {\bm\theta=\operatorname{vec}(\mathcal B,\bm C,\bm D)\in\R^p}.
  \label{eq:fixed-unknown-split}
\end{equation}
{Here $\bm\theta$ stacks all unknown coefficients of the quadratic tensor
$\mathcal B$, the linear operator $\bm C$, and the constant forcing $\bm D$
into one vector, and $\bm\Phi(\bm a)\in\R^{r\times p}$ is the feature matrix
whose columns evaluate the corresponding candidate terms, the unique
quadratic monomials $a_ja_k$ with $j\le k$, the linear terms $a_j$, and the
constant $1$, at the state $\bm a$.  The known part $\bm f_{\rm known}$
collects whatever portion of the drift is available a priori and may be
zero.  When a term is known---for example the projected forcing $\bm D_G$ of
\cref{eq:grom-D} in Setting II, or a known linear damping---it is carried by
$\bm f_{\rm known}$ and the corresponding columns are removed from the
library; in the fully known-forcing case, $\bm\theta$ reduces to
$\operatorname{vec}(\mathcal B,\bm C)$ and every inferred coefficient belongs
to an interaction term.}
In Setting II this split
also makes precise what is being learned: since the data are the
projected full-order trajectory, which obeys the closed equation
\cref{eq:pde-closure} rather than the G-ROM \cref{eq:pde-grom}, the
inferred $(\mathcal B,\bm C)$ are the Galerkin operators corrected by
the deterministic, dissipative part of the closure, and $\bm\Sigma$
represents its fluctuating part; neither correction requires knowledge
of $\bm\tau$ itself.

Discretizing \cref{eq:cpopinf-vector} over one time step by the
Euler--Maruyama scheme, the standard forward scheme for SDEs, gives the
increment model
\begin{equation}
  \bm a_{n+1}=\bm a_n
  +\Delta t\,\bm f_{\rm known}(\bm a_n)
  +\Delta t\,\bm\Phi_n\bm\theta
  +\bm\epsilon_n,
  \qquad
  \bm\epsilon_n\sim\mathcal N(\bm0,\bm Q_\Delta),
  \quad
  \bm Q_\Delta=\Delta t\,\bm\Sigma\bm\Sigma^\top,
  \label{eq:increment-model}
\end{equation}
where $\bm\Phi_n=\bm\Phi(\bm a_n)$: the next snapshot equals the current
one, advanced by the known and unknown parts of the drift over $\Delta t$,
plus a Gaussian random increment $\bm\epsilon_n$ whose covariance
$\bm Q_\Delta$ grows linearly in $\Delta t$, as is characteristic of
Brownian forcing.  Dividing \cref{eq:increment-model} by $\Delta t$ and
moving the known part to the left-hand side yields the equivalent
regression form
\begin{equation}
  \bm y_n\coloneqq\frac{\bm a_{n+1}-\bm a_n}{\Delta t}
  -\bm f_{\rm known}(\bm a_n)
  =\bm\Phi_n\bm\theta+\bm\eta_n,
  \qquad \bm\eta_n\sim\mathcal N(\bm0,\bm R),
  \qquad \bm R=\frac{\bm\Sigma\bm\Sigma^\top}{\Delta t},
  \label{eq:forward-response}
\end{equation}
in which the response $\bm y_n$ is a finite-difference approximation of
the unknown part of the drift, and the noise $\bm\eta_n$ is the scaled
stochastic increment.  Its covariance $\bm R$ is the covariance of this
regression residual, not of the state, and it determines the stochastic
amplitude through $\bm\Sigma\bm\Sigma^\top=\Delta t\,\bm R$; note that
$\bm R$ grows as $\Delta t$ decreases, reflecting the roughness of
stochastic trajectories.  \Cref{eq:forward-response} has the form of the
operator-inference regression \citep{peherstorfer2016opinf}, with one
essential difference: the residual $\bm\eta_n$ is not a fitting error to
be minimized away but carries the stochastic amplitude $\bm\Sigma$ that
we wish to identify.  For this reason the operators and the noise
cannot be estimated separately, and the inference of $(\bm\theta,\bm R)$,
and hence of $(\mathcal B,\bm C,\bm\Sigma)$, from
\cref{eq:forward-response} under the energy constraints of
\cref{sec:constraints} is the subject of \cref{sec:method}.
{Throughout, the snapshots are assumed to be exact, noise-free
observations of the state: measurement error would enter the finite
difference in \cref{eq:forward-response} at order
$\mathcal O(\sigma_{\rm obs}/\Delta t)$, \yrev{increasing $\bm R$ and biasing $\widehat{\bm\Sigma}$; this effect is examined in \cref{sec:robustness}.}}

\section{The stochastic physics-constrained operator inference framework}
\label{sec:method}

\subsection{Operator inference}
\label{sec:opinf-model}
The proposed framework is built on top of the classical operator-inference
(OpInf) regression \citep{peherstorfer2016opinf}, which we recall first in
the notation of \cref{sec:problem}.  For the resolved state
$\bm a(t)\in\R^r$, the OpInf model is
\begin{equation}
  \dot{\bm a}_{\mathrm{OpInf}}
  = \widehat{\bm c}+\widehat{\bm A}\bm a
  +\widehat{\bm H}\,(\bm a\,\widetilde\otimes\,\bm a),
  \label{eq:opinf-standard}
\end{equation}
where $\bm a\,\widetilde\otimes\,\bm a\in\R^{r(r+1)/2}$ collects the
unique quadratic monomials $a_ja_k$, $j\le k$
\citep{peherstorfer2016opinf,qian2020lift,mcquarrie2021regularized}.
This is exactly the drift of the stochastic quadratic model
\cref{eq:cpopinf-vector} with the stochastic component removed:
$\widehat{\bm A}$ corresponds to the linear operator $\bm C$,
$\widehat{\bm H}$ is the matricized form of the quadratic tensor
$\mathcal B$ in its upper-triangular representation, and
$\widehat{\bm c}$ corresponds to the constant forcing $\bm D$, \yrev{which is inferred unless it is known a priori and carried by $\bm f_{\rm known}$.}
OpInf estimates the operators by ordinary least squares on the regression
model \cref{eq:forward-response},
\begin{equation}
  \bm\theta_{\rm LS}=\arg\min_{\bm\theta}
  \sum_{n}\|\bm y_n-\bm\Phi_n\bm\theta\|_2^2,
  \label{eq:ls}
\end{equation}
with the same responses $\bm y_n$ and feature matrices $\bm\Phi_n$ as in
\cref{sec:data-model}, so OpInf and the proposed method are trained on
identical data and libraries.  \yrev{Throughout, OpInf serves as the deterministic reference model: it does not model the stochastic component, so in Setting I it is compared on operator-recovery errors only and is excluded from the statistical comparisons, and in Setting II it is integrated as a deterministic ROM.}
Its
differences from \spopinf{} are summarized in \cref{sec:opinf-contrast}.

{\paragraph{The four \yrev{reference models}.}
Throughout \cref{sec:numerical}, \spopinf{} is compared against the same four
reference models, defined once here.
\begin{itemize}\itemsep1pt
\item \emph{G-ROM}: the intrusive Galerkin ROM \cref{eq:pde-grom}, whose
  operators $(\mathcal B_G,\bm C_G,\bm D_G)$ are assembled from the PDE as in
  \cref{eq:grom-operators}, integrated with $\bm\Sigma=\bm0$.  It exists only
  in Setting II, where a full-order model is available.
\item \emph{G-ROM\,$(\sigma)$}: the same Galerkin operators, integrated with a
  nonzero diagonal amplitude obtained by projecting the physical forcing noise
  onto the reduced basis, $\sigma_i=\|\bm\Sigma_{\rm FOM}^\top\bm\phi_i\|_2$.
\item \emph{OpInf}: the ordinary least-squares fit \cref{eq:ls} on the same
  library and the same data, integrated deterministically
  ($\bm\Sigma=\bm0$).
\item \emph{OpInf\,$(\sigma)$}: the same least-squares drift, integrated with
  the diagonal amplitude read off \emph{afterwards} from its own regression
  residual, $\sigma_i=\sqrt{\Delta t\,\yrev{(N-1)^{-1}}\sum_n(y_{n,i}-(\bm\Phi_n
  \bm\theta_{\rm LS})_i)^2}$.  This is the sequential counterpart of
  \spopinf{}: same library, same noise model, but neither the energy
  constraints nor the joint calibration.
\end{itemize}
In Setting I (\cref{sec:ident-sds}) only OpInf and \spopinf{} are compared on
operator recovery, because there is no full-order model from which to assemble
a Galerkin \yrev{model}; and the statistical comparisons there use \spopinf{}
alone, because the reference is the true system itself rather than a
projection of it.  
The constrained-versus-unconstrained
\yrev{comparison} in Setting I is instead reported in \cref{sec:robustness}, where the
constraint set itself is varied on the CDV system.
}

\subsection{The stochastic physics-constrained operator inference (\spopinf{}) formulation}
\label{sec:joint-likelihood}
Stochastic physics-constrained operator inference (\spopinf{}) infers both
components of the stochastic quadratic model \cref{eq:cpopinf-vector} at
once.  The drift is parametrized through the candidate-term library of
\cref{sec:data-model},
\begin{equation}
  \bm f_{\bm\theta}(\bm a)
  =\bm f_{\rm known}(\bm a)+\bm\Phi(\bm a)\,\bm\theta,
  \qquad
  {\bm\ell(\bm a)
  =\big[(a_ja_k)_{j\le k},\;a_1,\ldots,a_r,\;1\big]^{\!\top}},
  \label{eq:spopinf-library}
\end{equation}
where {$\bm\ell(\bm a)$} is the feature vector of candidate terms, each row
of $\bm\Phi(\bm a)$ applies {$\bm\ell(\bm a)$} to one equation, and
{$\bm\theta=\operatorname{vec}(\mathcal B,\bm C,\bm D)$ stacks the
retained quadratic, linear, and constant coefficients, so the constant
forcing is inferred as part of \spopinf{}.  When $\bm D$, or any other part
of the drift, is known a priori, it is carried by $\bm f_{\rm known}$ and the
corresponding library columns are dropped (\cref{sec:data-model}).  The
energy constraints act only on the quadratic term, i.e., $\mathcal B$, so the rows of $\bm G$
carry zeros on the linear and constant columns.}
The second
unknown is the diagonal stochastic amplitude
$\bm\Sigma=\operatorname{diag}(\sigma_1,\ldots,\sigma_r)$, and physical
structure is imposed through linear equality constraints
{$\bm G\bm\theta=\bm0$, where $\bm G$ is a known constraint
matrix\yrev{.}}

Because the noise level enters the estimation criterion, the \spopinf{}
problem is \emph{not} the constrained version of the least-squares
problem \cref{eq:ls}.  Under the increment model
\cref{eq:increment-model}, the negative log-likelihood of the data, in
the response form of \cref{eq:forward-response} with
$\bm R=\bm\Sigma\bm\Sigma^\top/\Delta t$ and a ridge term
$\bm\Gamma\succeq\bm0$ for numerical stability, gives the \spopinf{}
formulation:
\begin{equation}
  \min_{\substack{\bm\theta,\;\bm R=\operatorname{diag}(R_1,\ldots,R_r)\\
                  R_i>0}}
  \frac{\yrev{N-1}}{2}\log|\bm R|
  +\frac12\sum_{n=\yrev{1}}^{N-1}
  (\bm y_n-\bm\Phi_n\bm\theta)^\top\bm R^{-1}
  (\bm y_n-\bm\Phi_n\bm\theta)
  +\frac12\bm\theta^\top\bm\Gamma\bm\theta,
  \qquad {\bm G\bm\theta=\bm0}.
  \label{eq:joint-deterministic-likelihood}
\end{equation}
The first term penalizes large noise levels and the second weights the
misfit by the inverse noise level, so the deterministic and stochastic
components are coupled in one criterion.  The following subsections
develop each ingredient of \cref{eq:joint-deterministic-likelihood} in
turn: \cref{sec:constraints-triangular} specifies the admissible
parametrization of $\bm\theta$ and constructs the rows of {$\bm G$}, and
{\cref{sec:alternating}} solves the optimization by
{an alternating calibration of the coefficients and the noise levels,
each step an inexpensive weighted least-squares solve}.

\subsection{Constraints and the upper-triangular formulation}
\label{sec:constraints-triangular}
\spopinf{} restricts which coefficient vectors $\bm\theta$ are admissible
in {\cref{eq:joint-deterministic-likelihood}}.
\Cref{sec:well-posedness} keeps one coefficient per quadratic monomial, the
upper-triangular formulation, and \cref{sec:constraints} imposes energy
conservation as linear equality constraints on $\bm\theta$, enforced exactly.

\subsubsection{Well-posedness and the upper-triangular representation}
\label{sec:well-posedness}
If every entry of each matrix $\bm B_i$ is treated as an independent
coefficient, the cross term in $\bm a^\top\bm B_i\bm a$ is
$[(\bm B_i)_{jk}+(\bm B_i)_{kj}]\,a_ja_k$, so the data can identify only
the sum.  The following proposition states the resulting failure of
well-posedness precisely; it sharpens the discussion in
\citet[Appendix~B.3.2--B.3.3]{mou2023multiscale}.

\begin{proposition}[Structural non-identifiability of the full
parametrization]
\label{prop:nonidentifiability}
Fix $i$ and $j\ne k$, let $\delta\ne0$, and let $\bm\nu$ be the parameter
perturbation that adds $+\delta$ to $(\bm B_i)_{jk}$ and $-\delta$ to
$(\bm B_i)_{kj}$ and leaves every other coefficient unchanged.  Then
$\bm f_{\bm\theta+\bm\nu}=\bm f_{\bm\theta}$ identically on $\R^r$.
Consequently, in the regression model \cref{eq:forward-response} the
feature matrices satisfy $\bm\Phi_n\bm\nu=\bm0$ for every snapshot, so
any estimation criterion that depends on $\bm\theta$ only through the
model, including the least-squares problem \cref{eq:ls} and the
likelihood introduced in \cref{sec:joint-likelihood}, is constant along
the affine lines $\bm\theta+\R\bm\nu$; the Gram matrix
$\sum_n\bm\Phi_n^\top\bm\Phi_n$, weighted or not, is singular; and no
amount of data determines $(\bm B_i)_{jk}$ and $(\bm B_i)_{kj}$
separately.  The estimation problem is therefore not well posed.
\end{proposition}

\begin{proof}
Since $a_ja_k=a_ka_j$, the form $\bm a^\top\bm B_i\bm a$ depends on the
pair $\big((\bm B_i)_{jk},(\bm B_i)_{kj}\big)$ only through its sum,
which $\bm\nu$ preserves; hence the drift is unchanged, and in the
regression \cref{eq:forward-response} the features multiplying the two
coefficients are the identical signal $a_j(t)a_k(t)$, so the
corresponding columns of every $\bm\Phi_n$ coincide and
$\bm\Phi_n\bm\nu=\bm0$.  Any criterion built from the residuals
$\bm y_n-\bm\Phi_n\bm\theta$ is then invariant under
$\bm\theta\mapsto\bm\theta+\bm\nu$, and $\bm\nu$ lies in the null space
of $\sum_n\bm\Phi_n^\top\bm W\bm\Phi_n$ for any weighting $\bm W$; in
particular, for Gaussian noise the Fisher information is proportional to
such a weighted Gram matrix and inherits the null vector.
\end{proof}
\begin{corollary}[Upper-triangular reparametrization]
\label{cor:triangular}
The map that sends each full slice to its upper-triangular representative,
$(\bm B_i)_{jk}\mapsto(\bm B_i)_{jk}+(\bm B_i)_{kj}$ for $j<k$ with
diagonal entries unchanged and lower entries set to zero, is a bijection
between the model-equivalence classes of
\cref{prop:nonidentifiability} and the retained coordinates.  It leaves the
vector field, and therefore the likelihood, unchanged, and it removes every
duplicated feature column, so the flat directions of
\cref{prop:nonidentifiability} disappear.  Identifiability of the restricted
problem then reduces to the standard requirement that \yrev{the design matrix have full column rank}, which is a property of
the data rather than of the parametrization.
\end{corollary}

Accordingly, \spopinf{} retains one coefficient per unique monomial
$a_ja_k$, $j\leq k$---an upper-triangular representation of each
$\bm B_i$.  The energy-preserving constraints derived next in
\cref{sec:constraints} are imposed in these coordinates, where they take
the explicit reduced form \cref{eq:reduced-relations}.

\Cref{fig:identifiability} illustrates \cref{prop:nonidentifiability} from
the inference viewpoint.  In the forward direction, the two coefficients act
as separate gains on the identical signal $a_j(t)a_k(t)$ feeding the same
trajectory component, so the observed contribution to the time derivative is
proportional to their sum.  In the inverse direction, the likelihood is
constant along the flat direction of \cref{prop:nonidentifiability}: its
minimizing set is an entire line, whereas after the reparametrization of
\cref{cor:triangular} the restricted problem is strictly convex in the
retained coefficient, with the unique minimizer equal to the sum.

\begin{figure}[H]
  \centering
  \includegraphics[width=\textwidth]{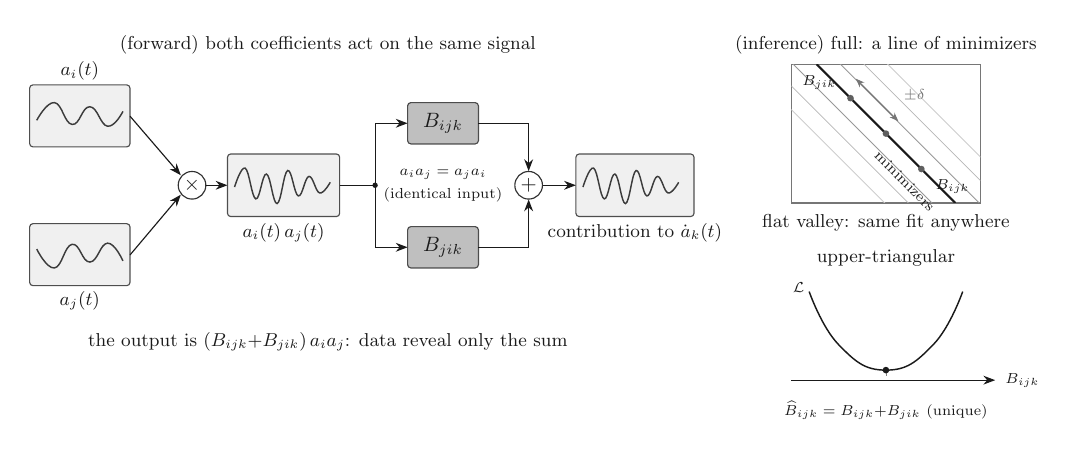}
  \caption{Why the upper-triangular representation is necessary.  Left
  (forward): {with $B_{ijk}=(\bm B_i)_{jk}$ denoting the coefficient of
  $a_ja_k$ in the $i$th equation, the coefficients $B_{ijk}$ and $B_{ikj}$ act
  as two gains applied to the identical signal $a_j(t)\,a_k(t)$ and are merged
  into the same output component, so the data reveal only the sum
  $B_{ijk}+B_{ikj}$}
  (\cref{prop:nonidentifiability}).
  \yrev{Right (inverse):} contours of the likelihood over the coefficient pair are parallel lines
  and the minimizers form an entire flat valley (top), whereas the
  upper-triangular restriction eliminates the flat direction and the
  restricted problem is strictly convex with the unique minimizer equal to
  the sum (bottom, \cref{cor:triangular}).}
  \label{fig:identifiability}
\end{figure}

\Cref{fig:three-models} places \spopinf{} between the two models it is compared with.  The closure ROM above it is intrusive: its operators are assembled from the PDE and its closure term cannot be evaluated in the reduced space.  
\spopinf{} is non-intrusive: every operator is inferred from the projected trajectory, and the stochastic component takes over
the role of the closure term.  Operator inference below it is likewise non-intrusive and shares the quadratic parametrization, but has neither the stochastic term nor the energy constraint.
\begin{figure}[H]
  \centering
  \includegraphics[width=.9\textwidth]{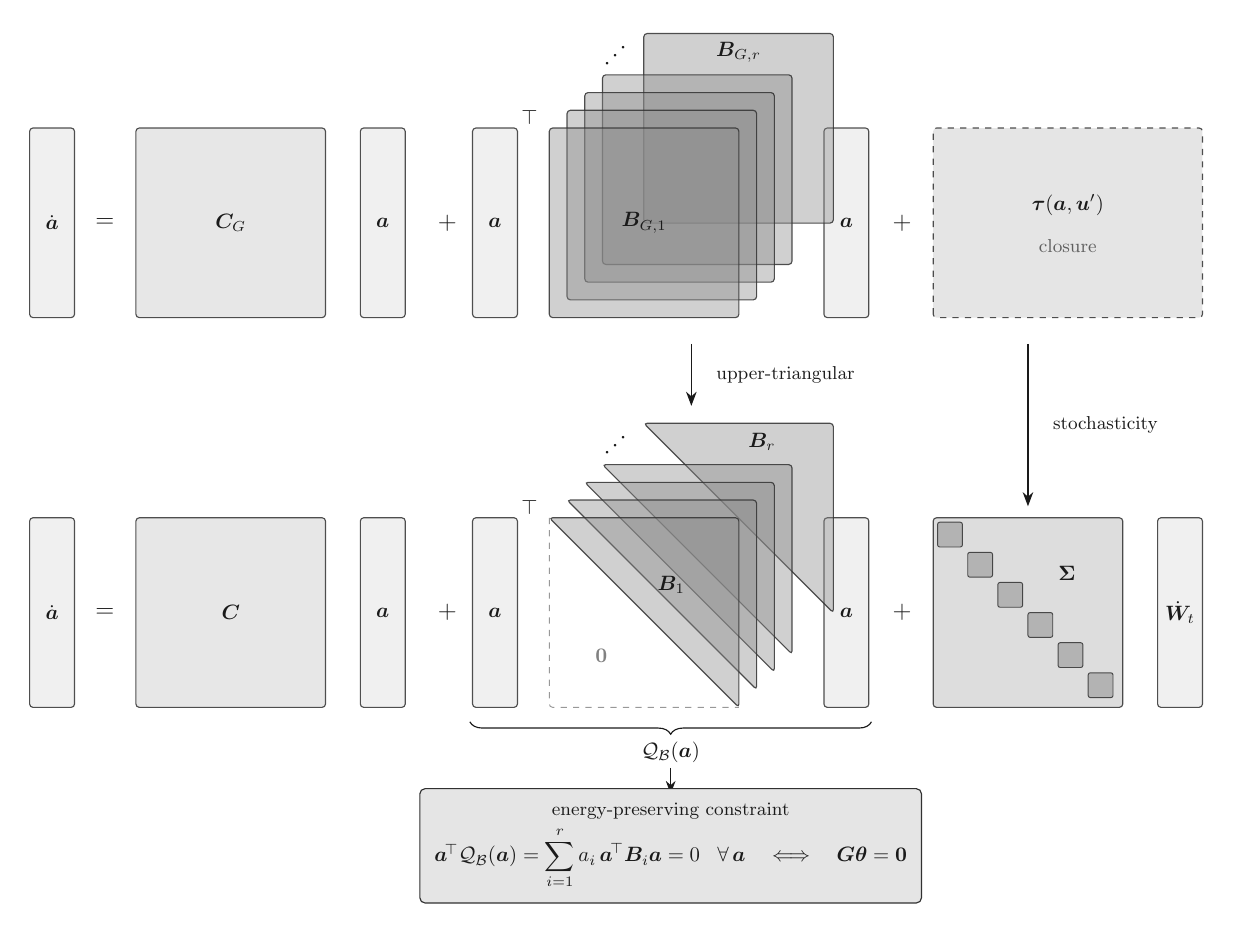}
  \caption{The three quadratic models compared in this paper.  (a) Closure
  ROM: the intrusive Galerkin projection \cref{eq:pde-grom}, with the full
  quadratic tensor $\mathcal B_G=(\bm B_{G,1},\ldots,\bm B_{G,r})$
  assembled from the PDE operators and the closure term
  $\bm\tau(\bm a,\bm u')$ (dashed) of \cref{eq:pde-closure}, which depends
  on the truncated modes and cannot be evaluated within the ROM.
  (b) \spopinf{}: the stochastic model \cref{eq:cpopinf-vector}, learned
  non-intrusively from the projected trajectory.  Its operators
  $(\mathcal B,\bm C)$ are inferred, not assembled, and differ from the
  Galerkin ones; the stochastic component $\bm\Sigma\dot{\bm W}_t$, the
  white-noise form of $\bm\Sigma\,\mathrm d\bm W_t$, mimics the effect of
  the closure term (arrow ``stochasticity''; the $\mathrm dt$ on the drift
  is implicit).  The quadratic tensor is in the upper-triangular
  representation of \cref{cor:triangular}, and the energy-preserving
  constraint \cref{eq:energy,eq:constraint} is enforced exactly (box).
  (c) Operator inference in its matricized form \cref{eq:opinf-standard},
  $\widehat{\bm H}\in\R^{r\times r(r+1)/2}$ acting on the unique monomials
  $\bm a\,\widetilde\otimes\,\bm a$: also non-intrusive and with the same
  parametrization as (b), but with neither stochastic term nor constraint,
  fitted by least squares.}
  \label{fig:three-models}
\end{figure}

\subsubsection{Energy-preserving quadratic constraints}
\label{sec:constraints}
Let $E(\bm a)=\tfrac12\|\bm a\|_2^2$ denote the resolved energy, so that
$\nabla E(\bm a)=\bm a$ and the level sets of $E$ are the spheres centered at
the origin.
{We call $E$ the resolved energy \yrev{because it coincides with the physical energy in}
\cref{sec:sbe-pde}, where the sine basis is orthonormal in the variable whose
$L^2$ norm is the energy.  In general $\tfrac12\|\bm a\|_2^2$ is the quadratic
invariant \emph{of the variable in which the basis is orthonormal}; for the
quasi-geostrophic problem of \cref{sec:qg-pde}, where POD is applied to the
vorticity, it is the resolved \emph{enstrophy}
$\tfrac12\|\omega_r\|_{L^2}^2$ rather than the energy
$\tfrac12(\psi_r,\omega_r)$, and \cref{rem:weighted-energy} treats the general
weighted case.  Nothing below depends on which of the two it is:
\cref{prop:energy-equivalence} is a statement about the cubic form alone.}
Along any trajectory of the deterministic part of
\cref{eq:cpopinf-vector}, the chain rule gives the energy rate
\begin{equation}
  \frac{\mathrm d}{\mathrm dt}E(\bm a)
  =\bm a^\top\dot{\bm a}
  =\underbrace{\bm a^\top\mathcal Q_{\mathcal B}(\bm a)}_{\text{quadratic}}
  +\;\bm a^\top\bm C\bm a+\bm a^\top\bm D.
  \label{eq:energy-rate}
\end{equation}
The quadratic term is called \emph{energy preserving} if its contribution to
\cref{eq:energy-rate} vanishes identically:
\begin{equation}
  \bm a^\top\mathcal Q_{\mathcal B}(\bm a)
  =\sum_{i=1}^r a_i\,\bm a^\top\bm B_i\bm a=0
  \quad\text{for every }\bm a\in\R^r.
  \label{eq:energy}
\end{equation}
Writing $B_{ijk}=(\bm B_i)_{jk}$ for the tensor entries, the left-hand side
of \cref{eq:energy} is the homogeneous cubic form
$\mathcal C(\bm a)=\sum_{i,j,k=1}^{r}B_{ijk}\,a_ia_ja_k$.  The following
proposition collects the equivalent characterizations of \cref{eq:energy}
that the estimator uses.

\begin{proposition}[Characterizations of energy preservation]
\label{prop:energy-equivalence}
The following statements are equivalent.
\begin{enumerate}
\item[\textup{(i)}] The quadratic term contributes no instantaneous energy
  change along any trajectory: the first term in \cref{eq:energy-rate}
  vanishes for every solution.
\item[\textup{(ii)}] $\mathcal Q_{\mathcal B}(\bm a)\perp\bm a$ for every
  $\bm a\in\R^r$; equivalently, the quadratic vector field is tangent to
  every level sphere of $E$.
\item[\textup{(iii)}] The full-tensor entries satisfy
\begin{subequations}
\label{eq:cubic-relations}
\begin{align}
  B_{iii}&=0,\\
  B_{iij}+B_{iji}+B_{jii}&=0, &&i\ne j,\\
  B_{ijk}+B_{ikj}+B_{jik}+B_{jki}+B_{kij}+B_{kji}&=0,
  &&i,j,k\ \text{distinct}.
\end{align}
\end{subequations}
\item[\textup{(iv)}] In the upper-triangular representation of
  \cref{sec:well-posedness}, where $B_{ijk}=0$ for $j>k$, the retained
  entries satisfy
\begin{subequations}
\label{eq:reduced-relations}
\begin{align}
  B_{iii}&=0, &&1\le i\le r,\\
  B_{iij}+B_{jii}&=0, &&i<j,\\
  B_{ijj}+B_{jij}&=0, &&i<j,\\
  B_{ijk}+B_{jik}+B_{kij}&=0, &&i<j<k,
\end{align}
\end{subequations}
which form $r+r(r-1)+\binom{r}{3}$ linear equations with at most three
nonzero entries per row.
\end{enumerate}
\end{proposition}

\begin{proof}
(i)$\Leftrightarrow$(ii): by \cref{eq:energy-rate}, the quadratic
contribution to the energy rate at state $\bm a$ equals
$\langle\nabla E(\bm a),\mathcal Q_{\mathcal B}(\bm a)\rangle
=\bm a^\top\mathcal Q_{\mathcal B}(\bm a)$, and $\nabla E(\bm a)=\bm a$ is
the outward normal of the sphere through $\bm a$, so the contribution
vanishes at every state if and only if the field is everywhere tangential.
(ii)$\Leftrightarrow$(iii): $\mathcal C$ is a polynomial on $\R^r$, and a
polynomial over an infinite field vanishes identically if and only if all of
its coefficients vanish.  Grouping the sum defining $\mathcal C$ by
monomials, the coefficient of $a_p^3$ is $B_{ppp}$; the coefficient of
$a_p^2a_q$ with $p\ne q$ is $B_{ppq}+B_{pqp}+B_{qpp}$; and the coefficient
of $a_pa_qa_s$ with $p,q,s$ distinct is the sum of $B_{ijk}$ over the six
permutations $(i,j,k)$ of $(p,q,s)$.  Setting each coefficient to zero gives
\cref{eq:cubic-relations}.
(iii)$\Leftrightarrow$(iv): substituting $B_{ijk}=0$ for $j>k$ into
\cref{eq:cubic-relations} removes the entries $B_{iji}$ for $i<j$, $B_{jji}$
for $i<j$, and, for $i<j<k$, the entries $B_{ikj}$, $B_{jki}$, and
$B_{kji}$, leaving exactly \cref{eq:reduced-relations}.  The count follows
from one equation per monomial: $r$ pure cubes, $r(r-1)$ ordered pairs for
the mixed-square monomials, and $\binom{r}{3}$ distinct triples.
\end{proof}

The relations \cref{eq:reduced-relations}, restricted to the 
coordinates of the regression library, are assembled row by row as
\begin{equation}
  {\bm G\bm\theta=\bm 0}.
  \label{eq:constraint}
\end{equation}

\begin{remark}[{Weighted energies}]
\nrev{\label{rem:weighted-energy}}%
{The choice $E(\bm a)=\tfrac12\|\bm a\|_2^2$ presumes a basis that is
orthonormal in the variable whose $L^2$ norm defines the conserved quantity.
If the conserved quantity is instead a general quadratic form
$E_{\bm M}(\bm a)=\tfrac12\bm a^\top\bm M\bm a$ with $\bm M$ symmetric
positive definite---for instance when POD is applied to the vorticity, so
that $\tfrac12\|\bm a\|_2^2$ is the resolved \emph{enstrophy} rather than the
energy---the change of variables $\tilde{\bm a}=\bm M^{1/2}\bm a$ reduces
$E_{\bm M}$ to the spherical case, and
\cref{prop:energy-equivalence} applies verbatim to the transformed tensor;
the resulting conditions remain linear equality constraints on $\bm\theta$.}
\end{remark}

\begin{remark}[Stochastic energy balance]
\label{rem:energy-balance}
For the full SDE \cref{eq:cpopinf-vector} with constant $\bm\Sigma$,
It\^o's formula gives
$\mathrm dE=\big[\bm a^\top\mathcal Q_{\mathcal B}(\bm a)
+\bm a^\top\bm C\bm a+\bm a^\top\bm D
+\tfrac12\operatorname{tr}(\bm\Sigma\bm\Sigma^\top)\big]\mathrm dt
+\bm a^\top\bm\Sigma\,\mathrm d\bm W_t$.
Under \cref{eq:energy} the first term vanishes identically, so
\begin{equation}
  \frac{\mathrm d}{\mathrm dt}\mathbb E[E]
  =\mathbb E\big[\bm a^\top\bm C\bm a+\bm a^\top\bm D\big]
  +\tfrac12\operatorname{tr}(\bm\Sigma\bm\Sigma^\top):
  \label{eq:mean-energy-balance}
\end{equation}
the quadratic term redistributes energy among components while the energy
budget is set entirely by the linear operator, the forcing, and the noise,
as in the stochastic climate-modeling framework
\citep{majda2001framework,mtv2003}.
\end{remark}

\begin{remark}[What the constraint removes, and what it does not]
\label{rem:blowup}
Decompose any quadratic field at $\bm a\ne\bm0$ into tangential and radial
parts, $\mathcal Q_{\mathcal B}(\bm a)
=\mathcal Q_{\rm tan}(\bm a)
+c(\widehat{\bm a})\,\|\bm a\|_2\,\bm a$ with
$\widehat{\bm a}=\bm a/\|\bm a\|_2$ and
$c(\widehat{\bm a})=\widehat{\bm a}^\top
\mathcal Q_{\mathcal B}(\widehat{\bm a})$.  The radial part injects energy
at rate $c(\widehat{\bm a})\,\|\bm a\|_2^3$, which is cubic in the state
norm; if regression noise makes $c$ positive along some direction, this
production eventually dominates any linear dissipation, whose removal rate
is only quadratic in the norm, and the identified model can blow up in
finite time \citep{schlegel2015}.  Constraint \cref{eq:constraint} sets
$c\equiv0$, deleting this mechanism identically.  The constraint alone does
not, however, guarantee boundedness: the linear part must additionally
dissipate, as in the trapping condition of \citet{kaptanoglu2021trapping}.
\Cref{lem:moment-bound} records the corresponding second-moment statement.
\end{remark}

\begin{lemma}[Second-moment bound under the constraint]
\label{lem:moment-bound}
Suppose \cref{eq:energy} holds and the symmetric part of the linear operator
is negative definite,
$\tfrac12\lambda_{\max}(\bm C+\bm C^\top)=-\mu<0$.  Then along
\cref{eq:cpopinf-vector},
\begin{equation}
  \frac{\mathrm d}{\mathrm dt}\mathbb E[E]
  \le-\mu\,\mathbb E[E]
  +\frac{\|\bm D\|_2^2}{2\mu}
  +\frac12\operatorname{tr}(\bm\Sigma\bm\Sigma^\top),
  \qquad\text{so}\qquad
  \limsup_{t\to\infty}\mathbb E[E]
  \le\frac{1}{\mu}\left(
  \frac{\|\bm D\|_2^2}{2\mu}
  +\frac12\operatorname{tr}(\bm\Sigma\bm\Sigma^\top)\right).
\end{equation}
\end{lemma}

\begin{proof}
By \cref{rem:energy-balance},
$\frac{\mathrm d}{\mathrm dt}\mathbb E[E]
=\mathbb E[\bm a^\top\bm C\bm a]+\mathbb E[\bm a^\top\bm D]
+\tfrac12\operatorname{tr}(\bm\Sigma\bm\Sigma^\top)$.
The quadratic form satisfies
$\bm a^\top\bm C\bm a\le-\mu\|\bm a\|_2^2=-2\mu E$, and Young's inequality
gives $\bm a^\top\bm D\le\|\bm a\|_2\|\bm D\|_2
\le\mu E+\|\bm D\|_2^2/(2\mu)$.  Combining the two bounds yields the
differential inequality, and Gr\"onwall's lemma yields the limit superior.
\end{proof}

\subsection{\texorpdfstring{An expectation--maximization \yrev{(EM)-based} iterative learning algorithm}{An expectation--maximization (EM)-based iterative learning algorithm}}
\rev{\label{sec:alternating}}

{\Cref{sec:joint-likelihood} stated the \spopinf{} criterion
\cref{eq:joint-deterministic-likelihood}; we now derive it from the increment
likelihood and solve it.  Recall that the unknowns are the coefficient vector
$\bm\theta=\operatorname{vec}(\mathcal B,\bm C,\bm D)\in\R^p$ of
\cref{eq:fixed-unknown-split}, which stacks the $r\cdot r(r+1)/2$ retained
quadratic coefficients, the $r^2$ linear coefficients, and, when the constant
forcing is inferred, the $r$ constant coefficients, so
$p=r^2(r+3)/2+r$ in the general case, and the diagonal stochastic amplitude
$\bm\Sigma=\operatorname{diag}(\sigma_1,\ldots,\sigma_r)$; the constraint
matrix $\bm G$ of \cref{eq:constraint} is known, and the
formulation applies to any such $\bm G$.}

Under the increment model \cref{eq:increment-model}, the
increments $\bm a_{n+1}-{\bm g_n}-\bm M_n\bm\theta$ are independent Gaussian
vectors, where we abbreviate
\begin{equation}
  {\bm g_n}=\bm a_n+\Delta t\,\bm f_{\rm known}(\bm a_n),
  \qquad
  \bm M_n=\Delta t\,\bm\Phi_n,
  \label{eq:increment-abbrev}
\end{equation}
that is, {$\bm g_n$} is the one-step prediction from the known part of the
drift and $\bm M_n$ maps coefficients to their one-step contribution.
The negative log-likelihood of the data, up to an additive constant, is
therefore
\begin{equation}
  \mathcal L(\bm\theta,\bm Q_\Delta)
  =\frac{\yrev{N-1}}{2}\log|\bm Q_\Delta|
  +\frac12\sum_{n=\yrev{1}}^{N-1}
  \yrev{\bm\epsilon_n^\top\bm Q_\Delta^{-1}\bm\epsilon_n},
  \qquad
  \yrev{\bm\epsilon_n}=\bm a_{n+1}-{\bm g_n}-\bm M_n\bm\theta,
  \qquad {\bm G\bm\theta=\bm0},
  \label{eq:joint-increment-likelihood}
\end{equation}
with $\bm Q_\Delta=\Delta t\,\bm\Sigma\bm\Sigma^\top$ the increment
covariance.  {Dividing each residual by $\Delta t$ converts
\cref{eq:joint-increment-likelihood} into the response form of
\cref{eq:forward-response}, with $\bm Q_\Delta=\Delta t^2\bm R$, and adding the
ridge term recovers exactly the working formulation
\cref{eq:joint-deterministic-likelihood} stated in
\cref{sec:joint-likelihood}.}

Following the joint calibration of \citet{mou2023multiscale}, we solve
\cref{eq:joint-deterministic-likelihood} by alternating minimization over
$\bm R$ and $\bm\theta$; {each subproblem is solved exactly and
inexpensively}.  Let
\yrev{$S=\{1,\ldots,p\}$ be the index set of library columns,} let $\bm\Phi_{n,S}$, $\bm\Gamma_{SS}$, and
{$\bm G_S$} denote the restrictions to those columns (discarding empty
constraint rows), and let $\bm e_n=\bm y_n-\bm\Phi_n\bm\theta$ be the
current residuals.  For fixed $\bm\theta$, minimizing over $\bm R$ gives
the residual-variance update
\begin{equation}
  \bm R \leftarrow
  \operatorname{diag}\!\left[\operatorname{diag}\!\left(
  \frac{1}{\yrev{N-1}}\sum_{n=\yrev{1}}^{N-1}\bm e_n\bm e_n^\top\right)\right],
  \label{eq:R-update}
\end{equation}
that is, $R_i$ is the mean squared residual of the $i$th equation.  For
fixed $\bm R$, the problem in $\bm\theta$ is an equality-constrained
weighted least-squares problem; with the Gram matrix and right-hand side
\begin{equation}
  \bm K_S=\sum_{n=\yrev{1}}^{N-1}\bm\Phi_{n,S}^\top\bm R^{-1}
  \bm\Phi_{n,S}+\bm\Gamma_{SS},
  \qquad
  \bm b_S=\sum_{n=\yrev{1}}^{N-1}\bm\Phi_{n,S}^\top\bm R^{-1}\bm y_n,
  \label{eq:normal-system}
\end{equation}
the stationarity conditions of its Lagrangian give the multiplier and the
constrained coefficients {explicitly} in closed form,
\begin{align}
  {\bm\lambda_G}&=({\bm G_S}\bm K_S^{-1}{\bm G_S^\top})^{-1}
  {\bm G_S}\bm K_S^{-1}\bm b_S,
  \label{eq:multiplier-update}\\
  \bm\theta_S&=\bm K_S^{-1}\bm b_S
  -\bm K_S^{-1}{\bm G_S^\top\bm\lambda_G},
  \label{eq:theta-update}
\end{align}
so every iterate satisfies {$\bm G_S\bm\theta_S=\bm0$} exactly.  This
alternation does not fit a deterministic drift first and then choose a
noise amplitude independently: through $\bm R^{-1}$ in
\cref{eq:normal-system}, the current noise estimate reweights the
coefficient update, and vice versa.  \yrev{Because all states are observed, the expectation step of the EM algorithm reduces to the residual-variance update \cref{eq:R-update}, and the maximization step to the constrained weighted least-squares solve.}

The inner alternation of
\cref{eq:R-update,eq:normal-system,eq:multiplier-update,eq:theta-update}
continues until the coefficient change falls below a tolerance
{$\varepsilon_{\rm tol}$}.

Finally, $\bm R$ is recomputed on the converged \yrev{coefficients},
and the stochastic amplitude follows from
$\bm\Sigma\bm\Sigma^\top=\Delta t\,\bm R$ of \cref{eq:forward-response}:
\begin{equation}
  \widehat{\sigma}_i=\sqrt{R_i\,\Delta t},
  \qquad i=1,\ldots,r.
  \label{eq:stochastic-return}
\end{equation}
The generalized least-squares equations
\cref{eq:normal-system,eq:multiplier-update,eq:theta-update} are an
\emph{inner update} for the joint likelihood, not the definition of
\spopinf{}.

{The complete procedure is summarized as
follows.\par\smallskip\noindent\textbf{Algorithm 1 (\spopinf{}).}
\textit{Input:} snapshots $\bm a_1,\ldots,\bm a_N$, step $\Delta t$, known
drift part $\bm f_{\rm known}$ (possibly zero), constraint matrix $\bm G$, ridge $\bm\Gamma$,
tolerance $\varepsilon_{\rm tol}$.
\begin{enumerate}
\item Form the responses $\bm y_n$ and features $\bm\Phi_n$ of
  \cref{eq:forward-response}; initialize $S=\{1,\ldots,p\}$ and
  $\bm R=\bm I$.
\item Repeat until $\|\bm\theta^{\rm new}-\bm\theta\|\le
  \varepsilon_{\rm tol}$: update $\bm\theta$ by the constrained weighted
  least squares \cref{eq:normal-system,eq:multiplier-update,eq:theta-update},
  then update $\bm R$ by \cref{eq:R-update}.
\item Recompute $\bm R$ on the converged \yrev{coefficients} and return
  $\widehat{\bm\theta}$ and
  $\widehat\sigma_i=\sqrt{R_i\,\Delta t}$ \cref{eq:stochastic-return}.
\end{enumerate}}

{\begin{remark}[Block structure and cost]
\label{rem:block-structure}
The claim that \spopinf{} costs nearly as much as OpInf rests on a structural
property of \cref{eq:normal-system} worth making explicit.  Because the same
feature vector $\bm\ell(\bm a)$ is applied to every equation, the feature
matrix of snapshot $n$ is $\bm\Phi_n=\bm I_r\otimes\bm\ell_n^\top$ with
$\bm\ell_n=\bm\ell(\bm a_n)\in\R^{p_{\rm eq}}$, $p_{\rm eq}=r(r+1)/2+r$
(plus one when $\bm D$ is inferred) and $p=r\,p_{\rm eq}$.  Hence
\begin{equation}
  \sum_{n=\yrev{1}}^{N-1}\bm\Phi_n^\top\bm R^{-1}\bm\Phi_n
  =\bm R^{-1}\otimes\bm M,
  \qquad
  \bm M=\sum_{n=\yrev{1}}^{N-1}\bm\ell_n\bm\ell_n^\top\in
  \R^{p_{\rm eq}\times p_{\rm eq}},
  \label{eq:gram-kronecker}
\end{equation}
so $\bm K$ is block diagonal with blocks $\bm M/R_i+\bm\Gamma_{ii}$ that share
the single matrix $\bm M$.  Three consequences follow.  (i) $\bm M$ and the
right-hand sides are accumulated once, in a single pass over the data costing
$\mathcal O(Np_{\rm eq}^2)$ --- the same pass OpInf makes --- after which no
iteration touches the snapshots again.  (ii) Without constraints the problem
\emph{decouples} equation by equation and is independent of $\bm R$, so the
unconstrained limit of \spopinf{} returns exactly $\bm\theta_{\rm LS}$: the
noise levels influence the estimate \emph{only} through the constrained
correction $-\bm K_S^{-1}\bm G_S^\top\bm\lambda_G$ of
\cref{eq:theta-update}, where they decide how the correction enforcing
$\bm G\bm\theta=\bm0$ is distributed across equations.  (iii) Each iteration
costs $\mathcal O(r\,p_{\rm eq}^3+m^3)$ with $m$ the number of retained
constraint rows, independent of $N$.  In practice $\bm K_S^{-1}$ is never
formed: each diagonal block is factorized by Cholesky and
\cref{eq:multiplier-update,eq:theta-update} are applied as triangular solves,
\yrev{and the inverses in those displays denote these solves.}  \Cref{fig:framework} summarizes the complete framework: the snapshot data of either setting, the joint calibration of the drift and the noise levels under the energy constraints, and the resulting stochastic model.
\begin{figure}[H]
  \centering
  \includegraphics[width=\textwidth]{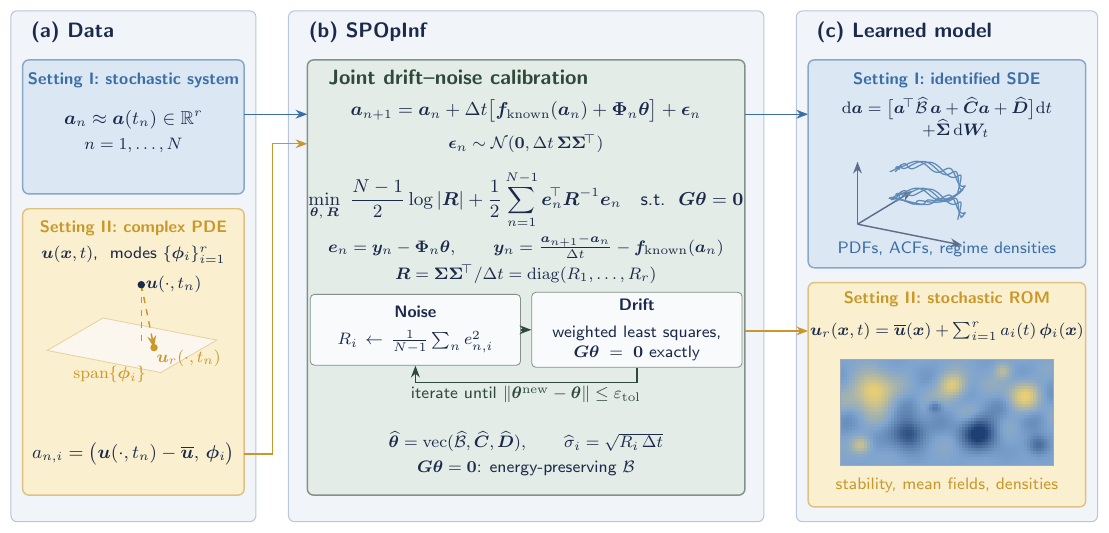}
  \caption{Overview of \spopinf{}.  (a) Training data: observed states of a
  stochastic system (Setting I) or projected coefficients of a PDE solution
  (Setting II).  (b) Joint calibration of the drift and the noise levels
  under the energy constraints $\bm G\bm\theta=\bm 0$.  (c) The learned
  stochastic model: an identified SDE (Setting I) or a stochastic ROM
  (Setting II).}
  \label{fig:framework}
\end{figure}
\end{remark}}

\subsection{Differences between \spopinf{} and OpInf}
\label{sec:opinf-contrast}
Since \spopinf{} and OpInf share the same quadratic--linear library and
the same training data, it is useful to state explicitly where the two
methods differ.  Throughout, we use drift and deterministic component
interchangeably, and likewise diffusion, or noise amplitude, and
stochastic component, as in \cref{sec:model-class}.

The first difference is what each method returns.  OpInf returns the
deterministic coefficients $\bm\theta_{\rm LS}$ and stops, so its output
is an ordinary differential equation: a drift alone.  \spopinf{} returns
the drift coefficients together with the diffusion amplitude,
$(\widehat{\bm\theta},\widehat{\bm\sigma})$, so its output is a complete
stochastic model of the form \cref{eq:cpopinf-vector}, which can be
simulated repeatedly and compared with the true system not only
trajectory by trajectory but also in its long-time statistics.

The second difference is how the coefficients are estimated.  OpInf
minimizes the plain sum of squared residuals \cref{eq:ls}, which treats
every equation equally.  \spopinf{} minimizes the likelihood criterion
\cref{eq:joint-deterministic-likelihood}, in which the noise level of
each equation is itself an unknown.  As a result, the deterministic and
stochastic components are estimated together: the current noise estimate
\cref{eq:R-update} decides how strongly each equation is weighted in the
coefficient update
\cref{eq:normal-system,eq:multiplier-update,eq:theta-update}, and the
returned diffusion amplitude \cref{eq:stochastic-return} belongs to the
final constrained model, rather than being read off afterwards from a
model fitted without it.

The third difference is physical consistency.  OpInf places no
restriction on the fitted coefficients, so noise in the data can produce
quadratic terms that create energy, which no physical member of the model
class does.  \spopinf{} enforces the energy-conservation relations
\cref{eq:constraint} as exact equality constraints inside the
optimization, so the learned quadratic tensor satisfies them by
construction rather than approximately.

In summary, these differences show that \spopinf{} is not simply
standard OpInf supplemented with a noise estimate.  Because the physical
constraints couple the coefficients of different equations, the
noise-weighted constrained regression yields a drift operator that differs
from the one obtained by unconstrained least squares.  Moreover, the
diffusion amplitude is not estimated a posteriori from the residual of a
previously fitted drift; instead, it is determined jointly with the drift
within the same optimization problem.

\section{Numerical tests}
\label{sec:numerical}
The numerical study is organized according to the two settings of \cref{sec:two-settings}. 
In \cref{sec:ident-sds} the method identifies stochastic dynamical systems from their observed state; in \cref{sec:pde-applications} it builds stochastic \yrev{reduced-order} models of complex PDEs from projected snapshots. 
In both cases the question is whether the learned model reproduces the key dynamical and statistical features of the underlying system
\citep{mou2023multiscale}, and the same evaluation methodology is used throughout.

\subsection{Stochastic dynamical systems}
\label{sec:ident-sds}
The first part of the numerical study concerns Setting I of
\cref{sec:two-settings}, in which the training data are directly observed
trajectories of finite-dimensional stochastic dynamical systems.  We
consider five test problems, illustrated in \cref{fig:test-models}, which
span intermittent, complex-valued, chaotic, cyclic, and metastable
dynamics.  All reference trajectories are generated by the
Euler--Maruyama scheme, and all stochastic regressions use the forward
response \cref{eq:forward-response}.  The constant forcing of each system
is treated as known and removed before fitting, so that the inferred
terms \yrev{contain} only the upper-triangular quadratic and the linear
terms.
{This instantiates the known-drift option of \cref{sec:data-model}: the
general \spopinf{} formulation infers $\bm D$ as part of $\bm\theta$, and
supplying the true forcing here simply removes the constant column from the
library.}

The definitions and experimental settings of the five systems are given
first; \cref{sec:sds-recovery} then shows the recovery results for the
operators and the stochastic amplitude, and \cref{sec:sds-statistics}
examines the long-time statistical quantities.

\begin{figure}[H]
  \centering
  \includegraphics[width=\textwidth]{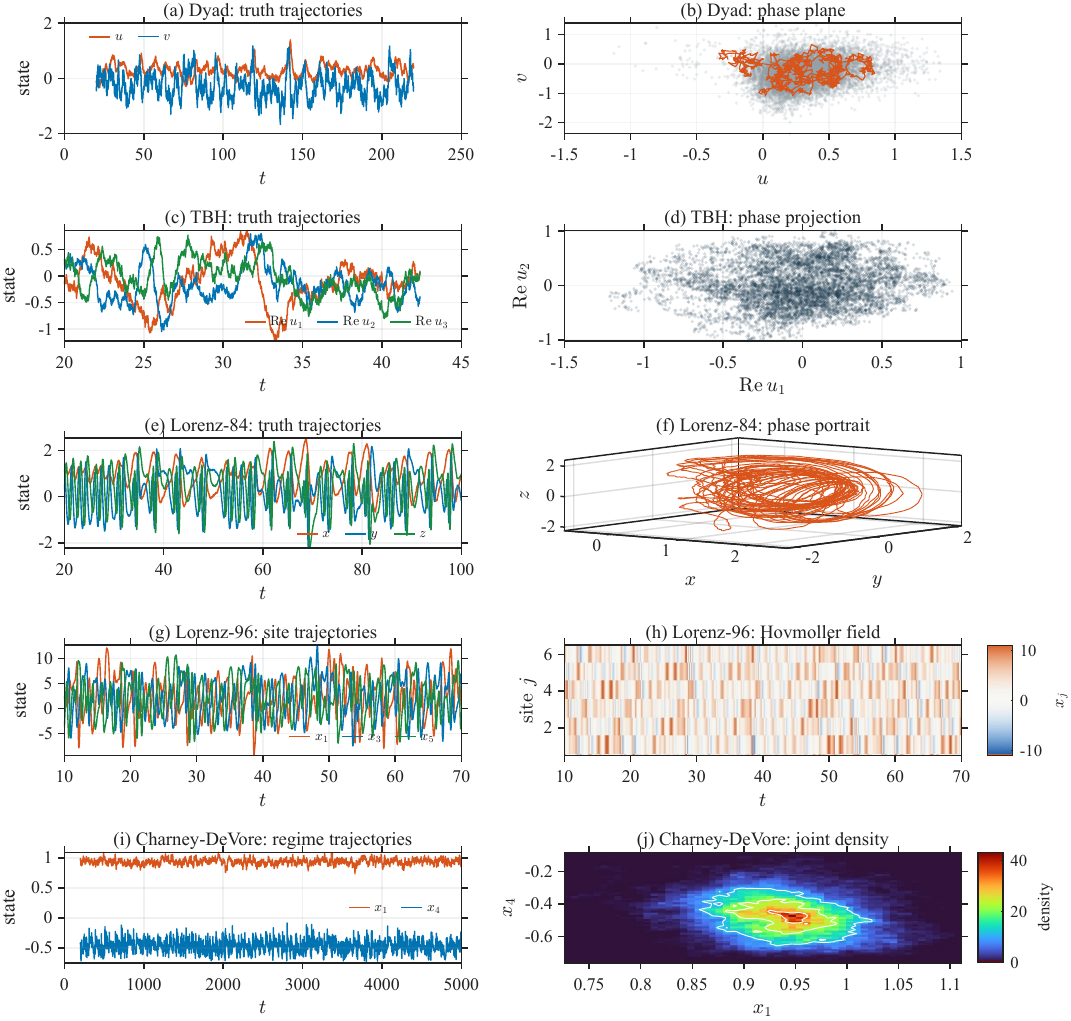}
  \caption{Test models. Truth visualizations show representative
  multicomponent trajectories and a complementary state-space view for the intermittent dyad
  (a,b), truncated Burgers--Hopf system (c,d), Lorenz--84 system (e,f),
  Lorenz--96 system (g,h), and Charney--DeVore system (i,j). The right-hand
  panels are, respectively, a phase plane, a two-mode phase projection, a
  three-dimensional phase portrait, a six-site Hovmoller field, and the joint
  density in the regime-discriminating $(x_1,x_4)$ coordinates. All panels use
  post-transient truth data from the corresponding identification runs.}
  \label{fig:test-models}
\end{figure}

\subsubsection{Intermittent stochastic dyad}
The two-state dyad isolates one nontrivial $a_i^2a_j$ energy constraint while
producing intermittent bursts.  Its governing SDE is \citep{majda2014conceptual,chen2016filtering,chen2018conditional}
\begin{align}
  \mathrm du&=(-d_u u+\gamma uv+F_u)\,\mathrm dt
  +\sigma_u\,\mathrm dW_u,
  \label{eq:dyad-u}\\
  \mathrm dv&=(-d_v v-\gamma u^2)\,\mathrm dt
  +\sigma_v\,\mathrm dW_v.
  \label{eq:dyad-v}
\end{align}
The quadratic contribution conserves $(u^2+v^2)/2$ because the two cubic
terms cancel.  Bursts occur when $v$ weakens the effective damping of $u$.

\paragraph{Settings.}
We set $d_u=0.1$, $d_v=0.5$, $\gamma=1$, $F_u=0.1$,
$(\sigma_u,\sigma_v)=(0.15,0.5)$, and $(u_0,v_0)=(0.1,-0.1)$.
The step is $\Delta t=0.005$, the training interval is $T=2000$, the first
$20$ time units are discarded, and the training seed is 10.  The known
forcing is $(F_u,0)^\top$.  For SPOpInf, we set $\yrev{\bm\Gamma}=10^{-8}\yrev{\bm I}$,
{$\varepsilon_{\rm tol}=10^{-8}$}\yrev{, and at most 100 alternating updates.  The pathwise comparison under identical noise realizations uses a window of 200 time units.}  Long-time PDFs and ACFs use three
independent trajectories of length 4000 per model.

\subsubsection{Three-mode truncated Burgers--Hopf system}
The truncated Burgers--Hopf (TBH) test is simulated in three complex modes
and inferred in the six-dimensional real coordinates
$\bm a=(\Re u_1,\Im u_1,\Re u_2,\Im u_2,\Re u_3,\Im u_3)^\top$
\citep{majda2006slow}.  With independent complex Wiener processes
$\mathrm dZ_k=(\mathrm dW_k^{(R)}+i\,\mathrm dW_k^{(I)})/\sqrt2$, the truth is
\begin{align}
  \mathrm du_1&=\left[-i\left(\overline u_1u_2+\overline u_2u_3\right)
  -\nu u_1\right]\mathrm dt+s\,\mathrm dZ_1,\\
  \mathrm du_2&=\left[-i\left(u_1^2+2\overline u_1u_3\right)
  -\nu u_2\right]\mathrm dt+s\,\mathrm dZ_2,\\
  \mathrm du_3&=\left[-3i u_1u_2-\nu u_3\right]\mathrm dt
  +s\,\mathrm dZ_3.
  \label{eq:tbh}
\end{align}
Thus the complex interaction coefficients are
$(-i,-i,-i,-2i,-3i)${, attached respectively to $\overline u_1u_2$ and
$\overline u_2u_3$ in the $u_1$ equation, to $u_1^2$ and $\overline u_1u_3$
in the $u_2$ equation, and to $u_1u_2$ in the $u_3$ equation,} and all three
linear damping coefficients are $-\nu$.
{In the real coordinates $\bm a$, the quadratic interactions conserve
$\tfrac12\sum_k|u_k|^2=\tfrac12\|\bm a\|_2^2$, so the energy-preserving
constraint of \cref{sec:constraints} applies.}
SPOpInf is not given this complex parametrization; it identifies a general
real quadratic tensor and the result is mapped back only for evaluation.

\paragraph{Settings.}
We set $\nu=0.5$, $s=0.5$, and hence real-component stochastic amplitude
$s/\sqrt2\approx0.3536$.  With random seed 10, each complex initial mode is
$0.1(\xi_k^{(R)}+i\xi_k^{(I)})$.  The step is $\Delta t=0.005$, the training
\yrev{interval} is $T=1000$, and the first 20 time units are discarded.  There is no
constant forcing.  For SPOpInf, we set $\yrev{\bm\Gamma}=10^{-8}\yrev{\bm I}$, {$\varepsilon_{\rm tol}=10^{-8}$}, and 100
maximum alternating updates.  \yrev{The pathwise comparison under identical noise realizations uses a window of 60 time units.}
Long-time PDFs and ACFs use three independent trajectories of length 1500 per
model after a 20-unit burn-in.

\subsubsection{Stochastic Lorenz--84 system}
The Lorenz--84 atmospheric circulation model \citep{lorenz1984} is
\begin{align}
  \mathrm dx&=\left[-y^2-z^2-a(x-F)\right]\mathrm dt
  +\sigma_x\,\mathrm dW_x,\\
  \mathrm dy&=\left[xy-bxz-y+G\right]\mathrm dt
  +\sigma_y\,\mathrm dW_y,\\
  \mathrm dz&=\left[bxy+xz-z\right]\mathrm dt
  +\sigma_z\,\mathrm dW_z.
  \label{eq:l84}
\end{align}
Its quadratic terms conserve $(x^2+y^2+z^2)/2$, whereas damping and forcing
control the energy balance.

\paragraph{Settings.}
The parameters are $a=0.25$, $b=4$, $F=8$, $G=1$, and
$\sigma_x=\sigma_y=\sigma_z=0.1$.  The initial state is $(1,1,1)^\top$,
$\Delta t=0.005$, $T=100$, the discarded transient is 20 time units, and the
training seed is 10.
For SPOpInf, we set $\yrev{\bm\Gamma}=10^{-8}\yrev{\bm I}$, {$\varepsilon_{\rm tol}=10^{-8}$} and 100 maximum alternating
updates.  \yrev{The pathwise comparison under identical noise realizations uses a window of 60 time units.}  Long-time PDFs and ACFs use three
independent trajectories of length 2000 per model after a 50-unit burn-in.

\subsubsection{Stochastic Lorenz--96 system}
For $J$ cyclic variables, with indices interpreted modulo $J$, the stochastic
Lorenz--96 model \citep{lorenz1996} is
\begin{equation}
  \mathrm dx_j=\left[(x_{j+1}-x_{j-2})x_{j-1}-x_j+F\right]\mathrm dt
  +\sigma\,\mathrm dW_j,
  \qquad j=1,\ldots,J.
  \label{eq:l96}
\end{equation}
The cyclic quadratic interaction is energy preserving and provides a test of
whether inference maintains both the cubic cancellation and site symmetry.

\paragraph{Settings.}
We set $J=6$, $F=8$, $\sigma=0.5$ at every site, and an initial state
$\bm x_0=F\bm1+0.01\bm\xi_0$.  The step is $\Delta t=0.005$, the training
\yrev{interval} is $T=1000$, the first 10 time units are discarded, and Euler--Maruyama
increments use seed 10.  The known forcing is $F\bm1$.  For SPOpInf, we set
$\yrev{\bm\Gamma}=10^{-8}\yrev{\bm I}$, {$\varepsilon_{\rm tol}=10^{-8}$}, and 100 maximum alternating updates.  \yrev{The pathwise comparison under identical noise realizations uses a window of 20 time units.}  Statistical PDFs are pooled over cyclic sites and ACFs are
site averaged using three independent trajectories of length 2000 per model
after a 20-unit burn-in.

\subsubsection{Six-mode Charney--DeVore system}
The six-mode Charney--DeVore (CDV) model tests metastable atmospheric regime
switching \citep{crommelin2004}.  To distinguish its scalar damping from the
inferred linear operator $\bm C$, denote the former by $c_d$.  The stochastic
governing equations are
\begin{align}
 \mathrm dx_1={}&[\widetilde\gamma_1x_3-c_d(x_1-x_1^*)]\,\mathrm dt
 +\sigma\,\mathrm dW_1,\\
 \mathrm dx_2={}&[-(\alpha_1x_1-\beta_1)x_3-c_dx_2
 -\delta_1x_4x_6]\,\mathrm dt+\sigma\,\mathrm dW_2,\\
 \mathrm dx_3={}&[(\alpha_1x_1-\beta_1)x_2-\gamma_1x_1-c_dx_3
 +\delta_1x_4x_5]\,\mathrm dt+\sigma\,\mathrm dW_3,\\
 \mathrm dx_4={}&[\widetilde\gamma_2x_6-c_d(x_4-x_4^*)
 +\epsilon(x_2x_6-x_3x_5)]\,\mathrm dt+\sigma\,\mathrm dW_4,\\
 \mathrm dx_5={}&[-(\alpha_2x_1-\beta_2)x_6-c_dx_5
 -\delta_2x_4x_3]\,\mathrm dt+\sigma\,\mathrm dW_5,\\
 \mathrm dx_6={}&[(\alpha_2x_1-\beta_2)x_5-\gamma_2x_4-c_dx_6
 +\delta_2x_4x_2]\,\mathrm dt+\sigma\,\mathrm dW_6.
 \label{eq:cdv}
\end{align}
For $m=1,2$, the coefficients are generated exactly as
\begin{align}
 \alpha_m&=\frac{8\sqrt2}{\pi}
 \frac{m^2(b^2+m^2-1)}{(4m^2-1)(b^2+m^2)},
 &\beta_m&=\frac{\beta b^2}{b^2+m^2},\\
 \delta_m&=\frac{64\sqrt2}{15\pi}
 \frac{b^2-m^2+1}{b^2+m^2},
 &\widetilde\gamma_m&=\frac{4\sqrt2\gamma_0}{\pi}
 \frac{mb^3}{(4m^2-1)(b^2+m^2)},\\
 \gamma_m&=\frac{4\sqrt2\gamma_0}{\pi}
 \frac{m^3b}{(4m^2-1)(b^2+m^2)},
 &\epsilon&=\frac{16\sqrt2}{5\pi}.
 \label{eq:cdv-coefficients}
\end{align}

\paragraph{Settings.}
The parameters are $x_1^*=0.95$, $x_4^*=-0.76095$, $c_d=0.1$,
$\beta=1.25$, $\gamma_0=0.2$, $b=0.5$, and the noise amplitude is
$\sigma=0.02$ in all six components.  The nominal initial state is
$(0.9,0,0,-0.6,0,0)^\top$, perturbed by Gaussian noise with standard
deviation $0.01$.  The system is integrated by the Euler--Maruyama scheme
with time step $\Delta t=0.005$ and random seed 10 over a training interval
of length $T=5000$, of which the first 200 time units are discarded as a
transient.  The known forcing is $(c_dx_1^*,0,0,c_dx_4^*,0,0)^\top$.  Because
the smallest physical couplings are of the order of $10^{-3}$, SPOpInf uses
a weak regularization $\yrev{\bm\Gamma}=10^{-8}\yrev{\bm I}$ and a tight tolerance
{$\varepsilon_{\rm tol}=10^{-8}$}, with at most 100
alternating updates.  Two constraint rows that are violated by the true
system, both associated with $\delta_2+\epsilon-\delta_1\approx-0.186$, are
removed. {These are the rows that enforce vanishing cubic coefficients
for the monomials $x_2x_4x_6$ and $x_3x_4x_5$, whose true values are
$\pm(\delta_2+\epsilon-\delta_1)$; consequently, the quadratic terms of CDV
are not exactly energy preserving}.  All remaining rows are enforced.
CDV is therefore the only test problem in this section whose quadratic
terms are \emph{not} energy preserving, and it is included for this reason:
it assesses the method when the standing assumption of
\cref{sec:constraints} is violated in a known and localized manner.  The
case in which the violation is \emph{not} known in advance is examined \yrev{below.}  Imposing the full constraint set, including the two
incompatible rows, yields $e_{\mathcal B}=0.155\pm0.022$, compared with
$0.249\pm0.014$ for unconstrained OpInf.  Thus, \spopinf{} retains most of
its advantage under \yrev{an incorrect} constraint set, and removing the two
incompatible rows reduces the error further to $0.140\pm0.024$.

We next describe how the two incompatible rows can be identified without
analytical knowledge of the system.  Ranking the constraint rows by the residual
$|\bm G_j\widehat{\bm\theta}|$ of a single unconstrained fit does
\emph{not} separate them: for one realization, the incompatible rows have
residuals $0.086$ and $0.013$ (normalized as in
\cref{eq:structure-residual}), whereas a compatible row attains $0.097$.
The sampling variability of the unconstrained fit thus exceeds the
magnitude of the violation.  Averaging the residuals over independent
trajectories separates the rows, because the violation is systematic
whereas the sampling error is not.  Over eight seeds, the two incompatible
rows rank first and second among all $56$ rows, with averaged residuals
$0.065$ and $0.063$, compared with a maximum of $0.060$ and a median of
$0.022$ over the compatible rows; their averaged magnitude agrees with the
true violation $0.063$ to within $3\%$.  This residual-based screening is
therefore effective, provided that an ensemble of trajectories, rather than
a single trajectory, is available.
The pathwise comparison under identical noise realizations uses a window of
400 time units.  The regime PDFs, ACFs, and the $(x_1,x_4)$ joint density
are computed from three independent trajectories of length 4000 per model.
\subsubsection{Operator and stochastic-component recovery}
\label{sec:sds-recovery}
Let $\widehat{\bm\theta}=\operatorname{vec}(\widehat{\mathcal B},\widehat{\bm C})$
denote the coefficient vector returned by \spopinf{}, that is, the
converged solution of the alternating updates
\cref{eq:R-update,eq:normal-system,eq:multiplier-update,eq:theta-update}\yrev{,}
 and let $\widehat{\bm\sigma}$ be the corresponding
stochastic amplitude \cref{eq:stochastic-return}.  Besides the relative
operator errors {$e_{\mathcal B}=\|\widehat{\mathcal B}-\mathcal
B\|_F/\|\mathcal B\|_F$ and $e_C=\|\widehat{\bm C}-\bm C\|_F/\|\bm C\|_F$}, we report the relative
{structure error}
\begin{equation}
  {e_G}(\widehat{\bm\theta})
  =\frac{\|{\bm G_{\rm kept}}\widehat{\bm\theta}\|_\infty}{\|\mathcal B\|_F},
  \label{eq:structure-residual}
\end{equation}
where {$\bm G_{\rm kept}=\bm G$} except for the CDV system, in which the two
constraint rows incompatible with the true dynamics are removed.  The
numerator is the largest violation among the constraints actually imposed;
since their target is zero, it is normalized by the Frobenius norm of the
true quadratic tensor, which provides the consistent scale for all
methods.  This measure follows the same energy-preservation principle as
the constrained quadratic identification methods of
\citet{kaptanoglu2021mhd,kaptanoglu2021trapping,koike2024epopinf} and the
stochastic ROMs of \citet{mou2023multiscale}, \yrev{here combined with the joint likelihood estimation} of \cref{sec:joint-likelihood}.
The operator-recovery results are collected in \cref{tab:summary}.

\begin{table}[H]
\centering
\caption{Relative deterministic-operator and structure errors.
Here $e_{\mathcal B}=\|\widehat{\mathcal B}-\mathcal B\|_F/
\|\mathcal B\|_F$ and \yrev{$e_C=\|\widehat{\bm C}-\bm C\|_F/\|\bm C\|_F$}.
{$e_G=\|\bm G_{\rm kept}\widehat{\bm\theta}\|_\infty/
\|\mathcal B\|_F$} is the relative structure error.}
\label{tab:summary}
\small
\begin{tabular}{llccc}
\toprule
System & Method & $e_{\mathcal B}$ & $e_C$ & {$e_G$} \\
\midrule
Dyad & \opinf{} & $1.0570\!\times\!10^{-1}$ & $1.3981\!\times\!10^{-1}$ & $9.1840\!\times\!10^{-2}$ \\
     & \spopinf{} & $3.8973\!\times\!10^{-3}$ & $5.4598\!\times\!10^{-2}$ & \yrev{$4.3\!\times\!10^{-18}$} \\
TBH  & \opinf{} & $1.1744\!\times\!10^{-1}$ & $1.8335\!\times\!10^{-1}$ & $4.5920\!\times\!10^{-2}$ \\
     & \spopinf{} & $8.6911\!\times\!10^{-2}$ & $1.8312\!\times\!10^{-1}$ & $5.4664\!\times\!10^{-17}$ \\
L84  & \opinf{} & $1.0992\!\times\!10^{-2}$ & $5.7567\!\times\!10^{-2}$ & $6.9352\!\times\!10^{-3}$ \\
     & \spopinf{} & $5.1161\!\times\!10^{-3}$ & $1.5536\!\times\!10^{-2}$ & $5.7824\!\times\!10^{-19}$ \\
L96  & \opinf{} & $5.7027\!\times\!10^{-3}$ & $3.0880\!\times\!10^{-2}$ & $1.8273\!\times\!10^{-3}$ \\
     & \spopinf{} & $2.7370\!\times\!10^{-3}$ & $3.3466\!\times\!10^{-3}$ & $2.3912\!\times\!10^{-17}$ \\
CDV  & \opinf{} & $2.4481\!\times\!10^{-1}$ & $1.1142$ & $8.9453\!\times\!10^{-2}$ \\
     & \spopinf{} & $1.1363\!\times\!10^{-1}$ & $3.3833\!\times\!10^{-1}$ & $9.8646\!\times\!10^{-17}$ \\
\bottomrule
\end{tabular}
\end{table}

For all five benchmarks, \spopinf{} reduces the
{quadratic-operator} error and
enforces the structure constraints to machine precision
(\cref{tab:summary}).
{The linear operator estimate is improved as well in four of the five systems, but
not uniformly: for TBH, $e_C=1.8312\times10^{-1}$ against
$1.8335\times10^{-1}$ for OpInf, an improvement of \yrev{only} about $0.1\%$.
This is what the structure of the estimator predicts.  By \cref{rem:block-structure}, the unconstrained problem decouples equation by equation and is independent of $\bm R$, so the noise weighting by itself
changes nothing; the constraints act only on the quadratic block, and they
reach $\bm C$ only indirectly, through the redistribution of the constrained
correction \yrev{across} equations.  Where the quadratic and linear blocks are
strongly correlated in the design,  as for the dyad, Lorenz--84, Lorenz--96 and
CDV, that indirect route is effective and $e_C$ falls by factors of $2$ to
$10$; for TBH the six real coordinates carry a nearly diagonal linear operator
$-\nu\bm I$ whose estimate is already at the noise floor of the data, and
there is nothing for the constraint to correct.}
{  All \spopinf{} entries are at
the level of the rounding error of the Cholesky solve, and none is
algebraically zero.}
{For CDV we report the \emph{total energy residual}
\begin{equation}
  \varrho(\mathcal B)
  =\max_{1\le i\le j\le k\le r}
  \Big|\!\!\sum_{(p,q,s)\in\mathrm{perm}(i,j,k)}\!\!B_{pqs}\Big|,
  \label{eq:total-energy-residual}
\end{equation}
the largest coefficient of the cubic form
$\mathcal C(\bm a)=\bm a^\top\mathcal Q_{\mathcal B}(\bm a)$; it vanishes
exactly when \cref{eq:energy} holds, and for the true CDV tensor it is
attained at the two monomials $x_2x_4x_6$ and $x_3x_4x_5$, whose coefficients
are $\pm(\delta_2+\epsilon-\delta_1)$, so that
$\varrho(\mathcal B_{\rm true})=|\delta_2+\epsilon-\delta_1|=0.1864$ exactly.
The constrained estimate gives $\varrho(\widehat{\mathcal B})=0.1825$, a
relative discrepancy of $0.021$: the method reproduces the physical energy
violation of CDV almost exactly, while satisfying every constraint
row to machine precision.}
Forcing the full residual to zero would nevertheless be physically incorrect;
the method instead satisfies the constraint set exactly.

For stochastic-component recovery we use a separate relative error,
\begin{equation}
  e_\sigma
  =\frac{\|\widehat{\bm\sigma}-\bm\sigma\|_2}
  {\|\bm\sigma\|_2}.
  \label{eq:sigma-error}
\end{equation}
This quantity is reported only for SPOpInf because ordinary OpInf does not
infer $\bm\sigma$.

\begin{table}[H]
\centering
\caption{SPOpInf stochastic-component estimates and their relative errors.
The $e_\sigma$ values {are computed from the full-precision
estimates\yrev{.}}}
\label{tab:stochastic-component}
\begin{tabular}{lccc}
\toprule
System & $\widehat{\bm\sigma}$ & $\bm\sigma$ (truth) & $e_\sigma$ \\
\midrule
Dyad & $(0.1501,0.4997)$ & $(0.15,0.50)$ & $6.0331\!\times\!10^{-4}$ \\
TBH & $(0.3534,0.3537,0.3534,0.3534,0.3541,0.3533)$
    & $0.3536\,\bm 1_6$ & $7.8853\!\times\!10^{-4}$ \\
L84 & $(0.0990,0.1005,0.0999)$ & $0.1\,\bm 1_3$ & $6.7355\!\times\!10^{-3}$ \\
L96 & $(0.4997,0.4999,0.5008,0.5000,0.4998,0.4997)$
    & $0.5\,\bm 1_6$ & $7.9237\!\times\!10^{-4}$ \\
CDV & $0.0200\,\bm 1_6$ & $0.02\,\bm 1_6$ & $7.7728\!\times\!10^{-4}$ \\
\bottomrule
\end{tabular}
\end{table}

The separate stochastic-component results in \cref{tab:stochastic-component} show that SPOpInf recovers $\bm\sigma$ with relative error below $10^{-3}$ for the dyad, TBH, Lorenz--96, and CDV cases and with error $6.7355\times10^{-3}$ for
Lorenz--84. {For TBH, in addition, the real-valued inference recovers the underlying
complex interaction coefficients with absolute errors between
$2.00\times10^{-2}$ and $4.50\times10^{-2}$, despite not being given the
complex parametrization.}
{\Cref{tab:tbh-complex} lists the five coefficients individually.  They
reconcile with $e_{\mathcal B}=8.69\times10^{-2}$ as follows: the absolute
errors are all of order $3\times10^{-2}$, while
$\|\mathcal B\|_F$ for the real six-dimensional tensor is dominated by the
single coefficient $-3i$, so the relative error of the tensor as a whole is a
few percent even though no individual coefficient is off by more than
$4.5\times10^{-2}$.  The spurious real parts, which should all vanish, are
bounded by $3.6\times10^{-2}$ and are the visible signature of the estimator
not being told that the underlying interaction is purely imaginary.}

{
\begin{table}[H]
\centering
\caption{{Truncated Burgers--Hopf interaction coefficients recovered by
\spopinf{} in real coordinates and mapped back to the complex
parametrization, which the estimator is not given.  The last column is
$|\widehat c-c|$.  The three linear damping coefficients are recovered as
$-0.5172+0.0056i$, $-0.5220+0.0329i$ and $-0.4851+0.0059i$ against the true
$-\nu=-0.5$.}}
\label{tab:tbh-complex}
\small
\begin{tabular}{llcc c}
\toprule
Equation & Term & True $c$ & Recovered $\widehat c$ & $|\widehat c-c|$ \\
\midrule
$u_1$ & $\overline u_1u_2$ & $-i$  & $\phantom{-}0.0263-1.0175\,i$ & $3.16\times10^{-2}$ \\
$u_1$ & $\overline u_2u_3$ & $-i$  & $-0.0127-1.0336\,i$           & $3.60\times10^{-2}$ \\
$u_2$ & $u_1^2$            & $-i$  & $-0.0356-1.0168\,i$           & $3.94\times10^{-2}$ \\
$u_2$ & $\overline u_1u_3$ & $-2i$ & $\phantom{-}0.0325-1.9689\,i$ & $4.50\times10^{-2}$ \\
$u_3$ & $u_1u_2$           & $-3i$ & $-0.0198-3.0025\,i$           & $2.00\times10^{-2}$ \\
\bottomrule
\end{tabular}
\end{table}
}
No stochastic-component error is reported for ordinary OpInf because its
least-squares deterministic fit does not identify $\bm\Sigma$.

\subsubsection{Long-time statistical fidelity}
\label{sec:sds-statistics}

\begin{figure}[H]
  \centering
  \includegraphics[width=.7\textwidth]{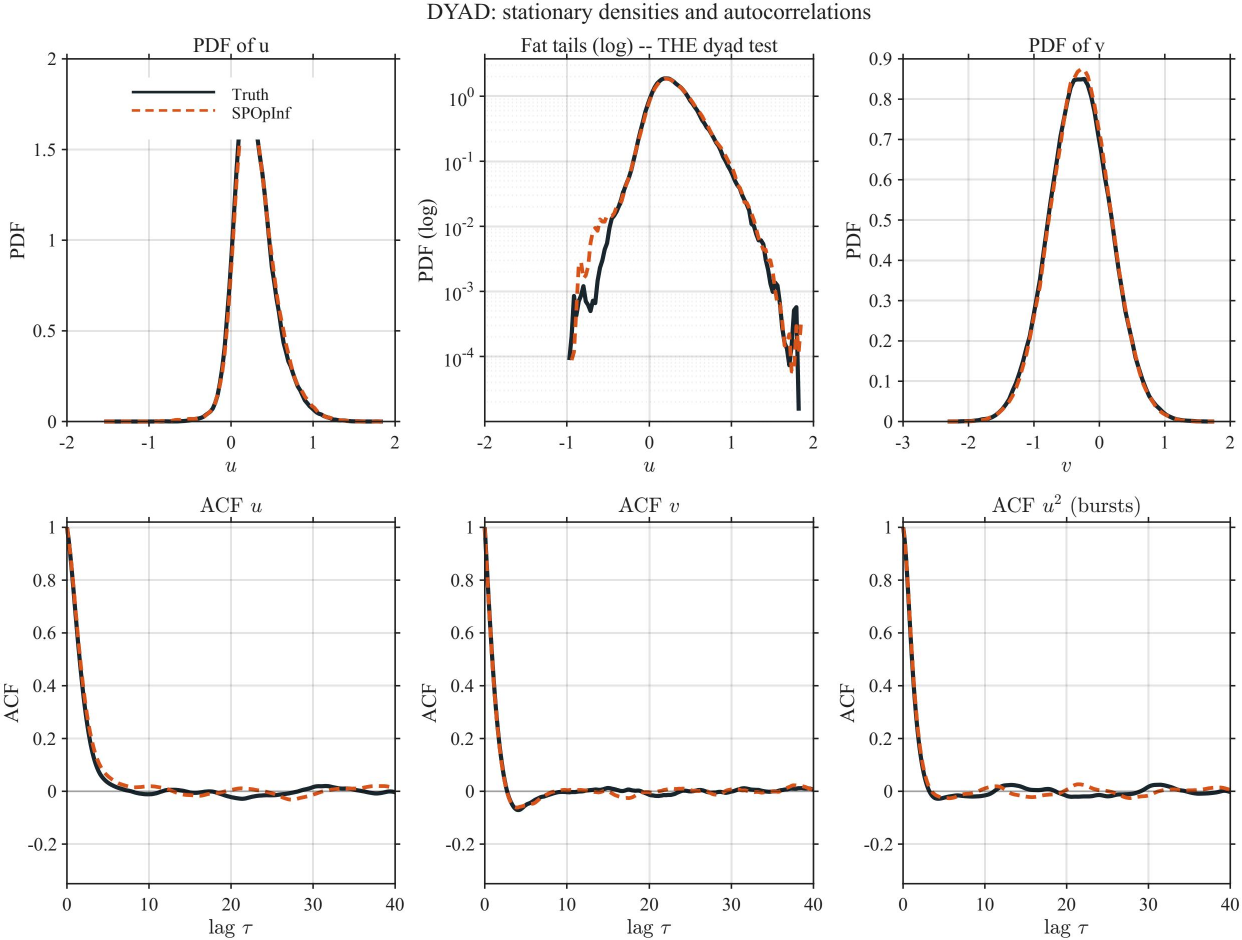}
  \caption{Intermittent dyad statistics from independent truth and SPOpInf
  ensembles.  The long-tailed marginal and autocorrelation diagnostics test
  burst frequency and duration, which are not determined by a small one-step
  residual alone.}
  \label{fig:dyadstats}
\end{figure}

\begin{figure}[H]
  \centering
  \includegraphics[width=.7\textwidth]{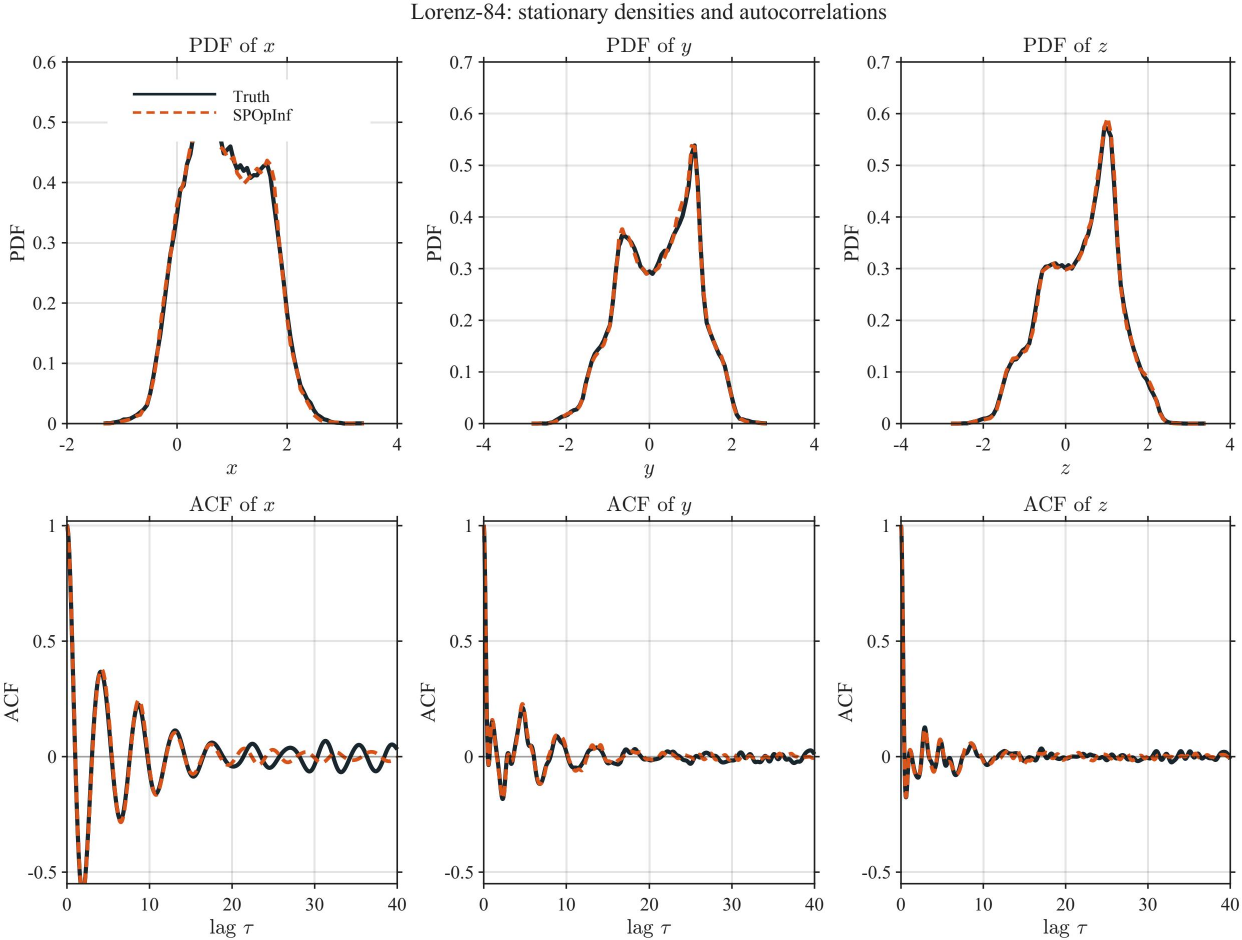}
  \caption{Lorenz--84 marginal densities (top) and autocorrelation functions
  (bottom).  The SPOpInf SDE reproduces the principal equilibrium features of
  the truth while simultaneously satisfying the quadratic energy relations.}
  \label{fig:l84stats}
\end{figure}

{\Cref{tab:statistics} shows the \yrev{quantitative} statistical comparison.  For each system it reports, over all resolved components, the largest discrepancy
between the \spopinf{} ensemble and the truth ensemble in the first four moments and in the integral correlation time.  Means and standard deviations agree to better than $4\%$ for four of the five systems, the
skewness to within $0.14$ and the kurtosis to within $0.37$ in absolute units,
and the decorrelation time to within $9\%$ everywhere.  The largest
discrepancies are those of CDV, whose mean is displaced by $0.20$ standard
deviations; this is the regime-occupancy error visible in
\cref{fig:cdvstats,fig:cdvjoint} and is the expected consequence of the weak
topographic couplings that set the switching rate.}

{
\begin{table}[H]
\centering
\caption{{Quantitative long-time statistical errors between independent
truth and \spopinf{} ensembles.  Each entry is the maximum over the resolved
components of: the mean discrepancy in units of the truth standard deviation,
$|\widehat\mu_i-\mu_i|/s_i$; the relative standard-deviation error
$|\widehat s_i-s_i|/s_i$; the absolute skewness and excess-kurtosis
discrepancies; and the relative error in the integral correlation time
$\tau_i$.  Computed from the same ensembles that produce
\cref{fig:dyadstats,fig:tbhstats,fig:l84stats,fig:l96stats,fig:cdvstats}.}}
\label{tab:statistics}
\small
\begin{tabular}{lccccc}
\toprule
System & $|\Delta\mu|/s$ & $|\Delta s|/s$ & $|\Delta\mathrm{skew}|$
       & $|\Delta\mathrm{kurt}|$ & $|\Delta\tau|/\tau$ \\
\midrule
Dyad & $0.016$ & $0.034$ & $0.068$ & $0.304$ & $0.045$ \\
TBH  & $0.070$ & $0.040$ & $0.122$ & $0.190$ & $0.054$ \\
L84  & $0.017$ & $0.011$ & $0.023$ & $0.043$ & $0.018$ \\
L96  & $0.009$ & $0.009$ & $0.021$ & $0.018$ & $0.000$ \\
CDV  & $0.203$ & $0.025$ & $0.141$ & $0.372$ & $0.083$ \\
\bottomrule
\end{tabular}
\end{table}
}
Each plot compares the reference
process with SPOpInf; ordinary OpInf is omitted because it does not infer a
stochastic component and therefore does not define the stochastic model needed
for a PDF or ACF comparison.  Lorenz--84 densities and ACFs are closely
matched (\cref{fig:l84stats}); the dyad retains intermittent tails
(\cref{fig:dyadstats}); {L96 retains cyclically pooled statistics
(\cref{fig:l96stats}); the TBH marginal PDFs and ACFs are compared in
\cref{fig:tbhstats}; and the CDV marginals and joint regime density are
examined in \cref{fig:cdvstats,fig:cdvjoint} below.}

\begin{figure}[H]
  \centering
  \includegraphics[width=.7\textwidth]{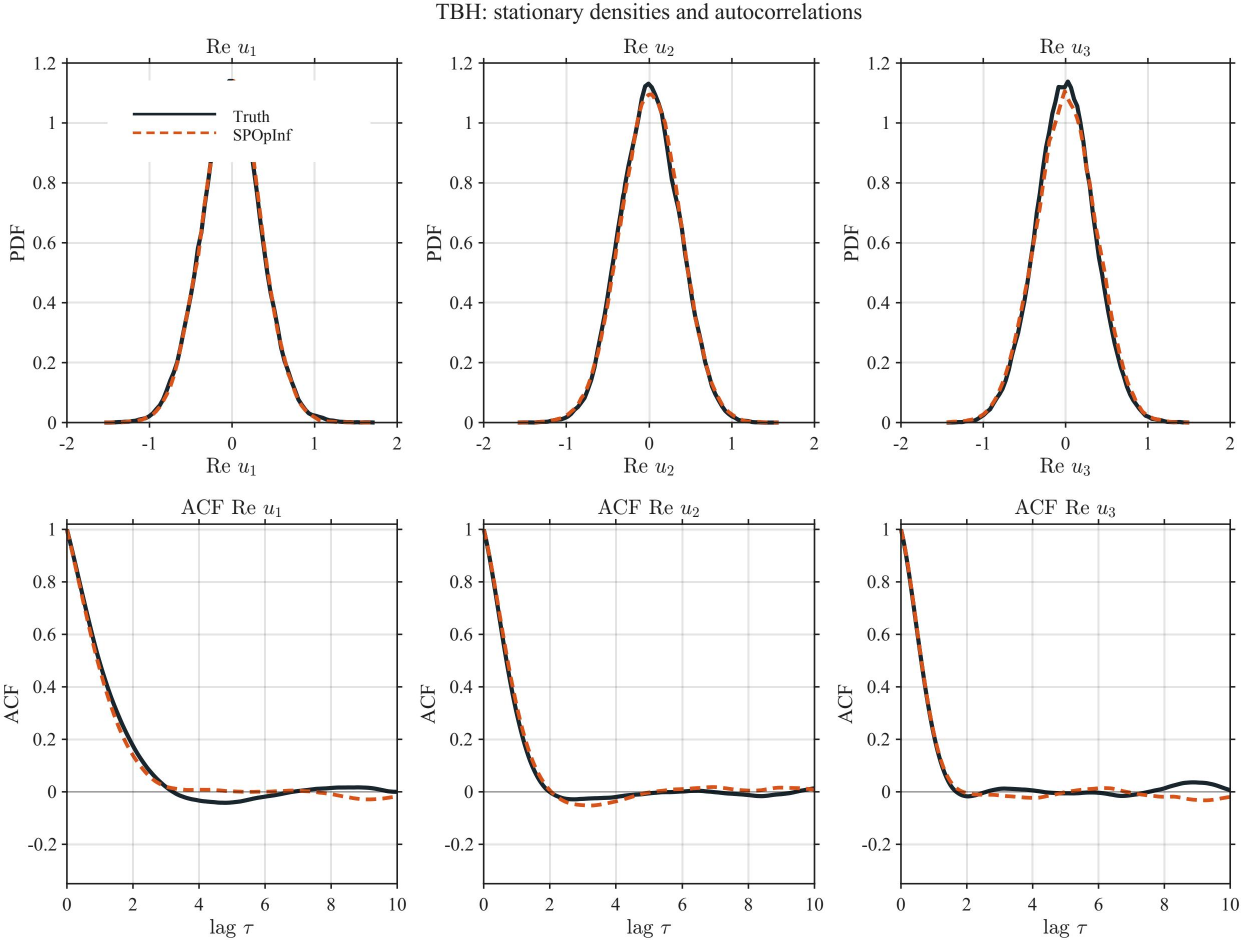}
  \caption{TBH marginal PDFs and autocorrelation functions from independent
  ensembles.  The complex truth and real-valued SPOpInf SDE test whether the
  inferred model reproduces the equilibrium statistics.}
  \label{fig:tbhstats}
\end{figure}

\begin{figure}[H]
  \centering
  \includegraphics[width=.7\textwidth]{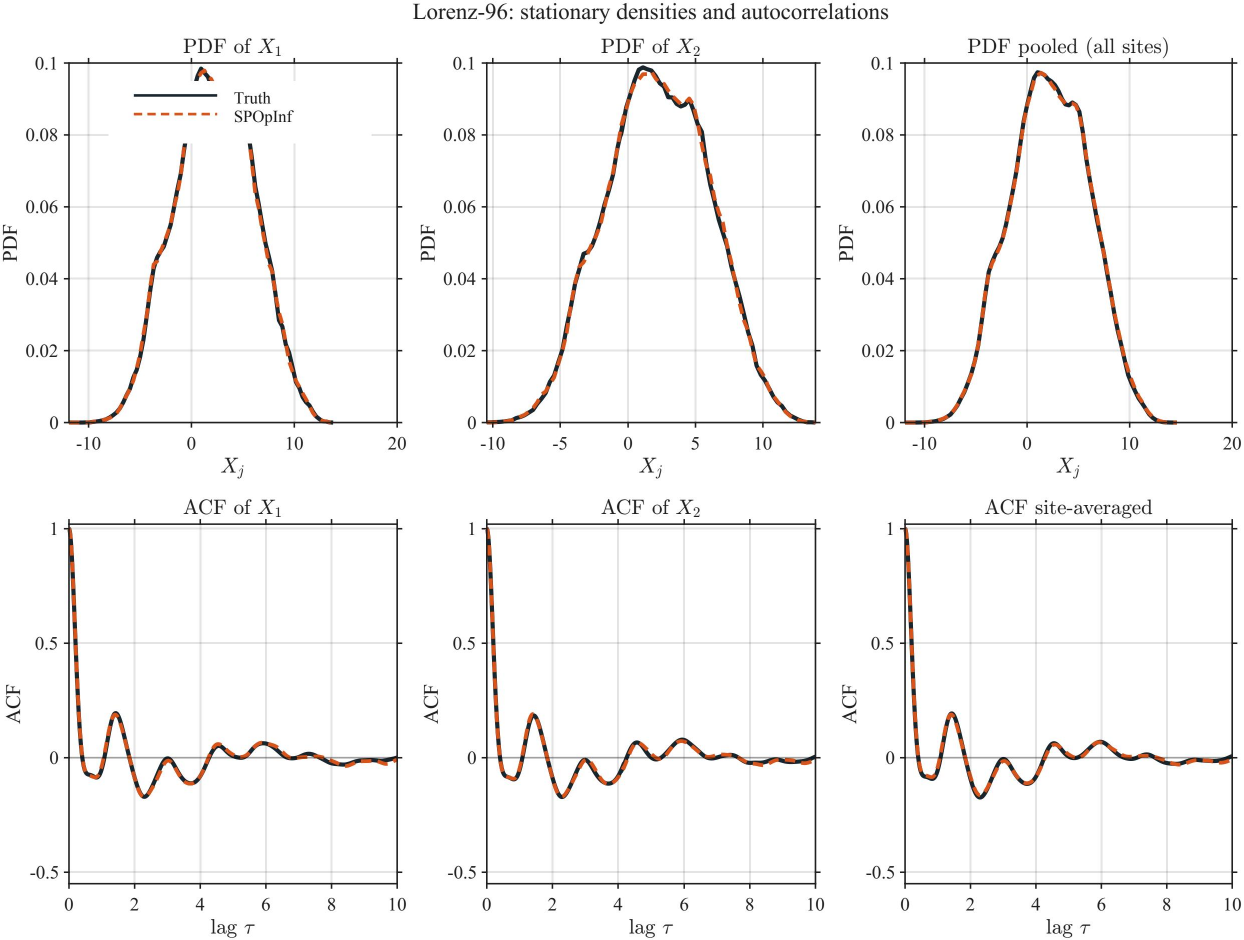}
  \caption{Lorenz--96 equilibrium PDF and ACF comparison between truth and
  SPOpInf.  PDFs are pooled over cyclic sites and ACFs are site averaged,
  making symmetry breaking or distributional bias visible without relying on
  individual chaotic paths.}
  \label{fig:l96stats}
\end{figure}

\begin{figure}[H]
  \centering
  \includegraphics[width=.7\textwidth]{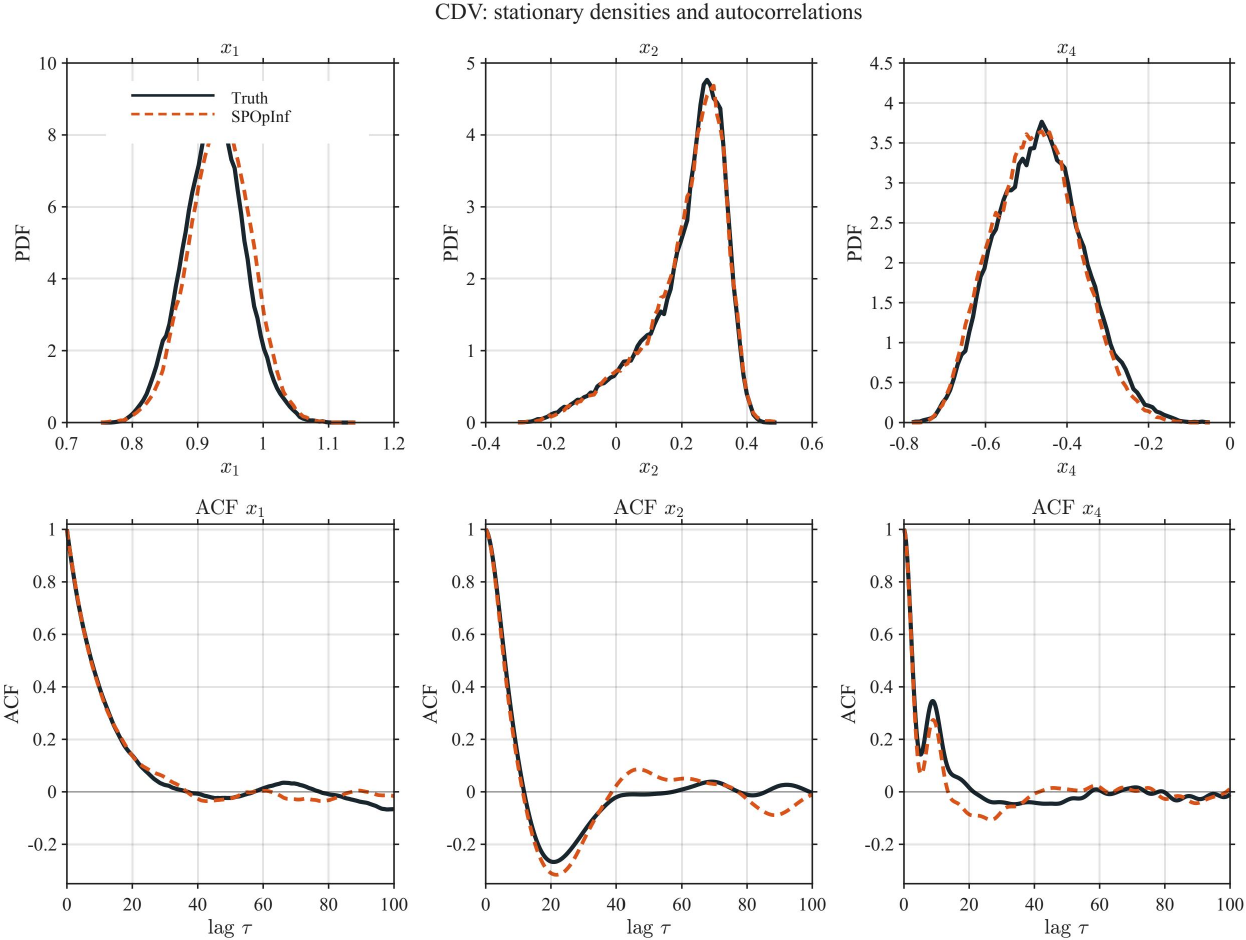}
  \caption{CDV marginal PDFs and autocorrelation functions from independent
  truth and SPOpInf ensembles.  These distributional diagnostics complement
  the joint regime density by testing each resolved component separately.}
  \label{fig:cdvstats}
\end{figure}

\begin{figure}[H]
  \centering
  \includegraphics[width=.7\textwidth]{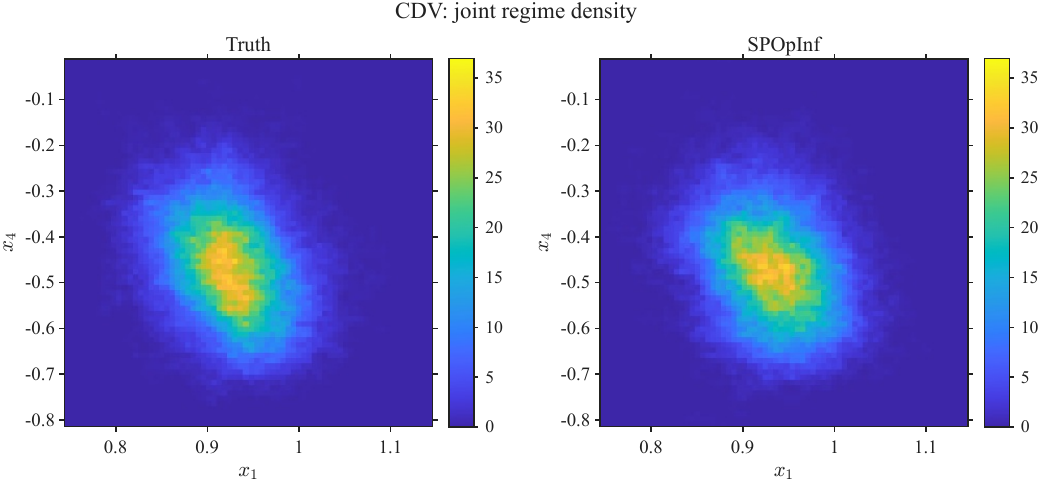}
  \caption{Two-panel comparison of the truth and SPOpInf joint CDV regime
  densities from independent simulations.  Retaining both metastable regions
  is a stronger test than matching componentwise means and variances.}
  \label{fig:cdvjoint}
\end{figure}

The CDV marginal comparisons in \cref{fig:cdvstats} and joint density in
\cref{fig:cdvjoint} are the most stringent equilibrium tests in the present
suite because weak topographic couplings control switching between regimes.
\yrev{Conclusions about rare transitions would, however, require uncertainty estimates across more independent training trajectories and confidence intervals for dwell-time distributions.}

\subsection{Complex partial differential equations}
\label{sec:pde-applications}
The second part of the numerical study is Setting II of
\cref{sec:two-settings}: the training data are ROM coefficients of
nonlinear PDEs, and the identified SDE {\cref{eq:cpopinf-vector}} serves as a
non-intrusive stochastic reduced-order model.
We test whether the \spopinf{} framework remains accurate after projection
of nonlinear PDEs.  The two applications are the one-dimensional stochastic
Burgers equation and the two-dimensional barotropic quasi-geostrophic (QG)
equation{; their reference solutions are shown in
\cref{fig:test-models-pde}}.

{In both test problems, \yrev{the training data are the responses $\bm y_n$ of \cref{eq:forward-response} with $\bm f_{\rm known}=\bm D_G$, the projected forcing \cref{eq:grom-D}; for the stochastic Burgers equation, which has no deterministic forcing, $\bm D_G=\bm 0$.}}
The estimator uses the upper-triangular quadratic-linear
formulation, \yrev{alternating updates of the coefficients and noise levels}, and constraint
$\bm a^\top\mathcal Q_{\mathcal B}(\bm a)=0$ described in \cref{sec:method}.

\begin{figure}[H]
  \centering
  {\includegraphics[width=\textwidth]{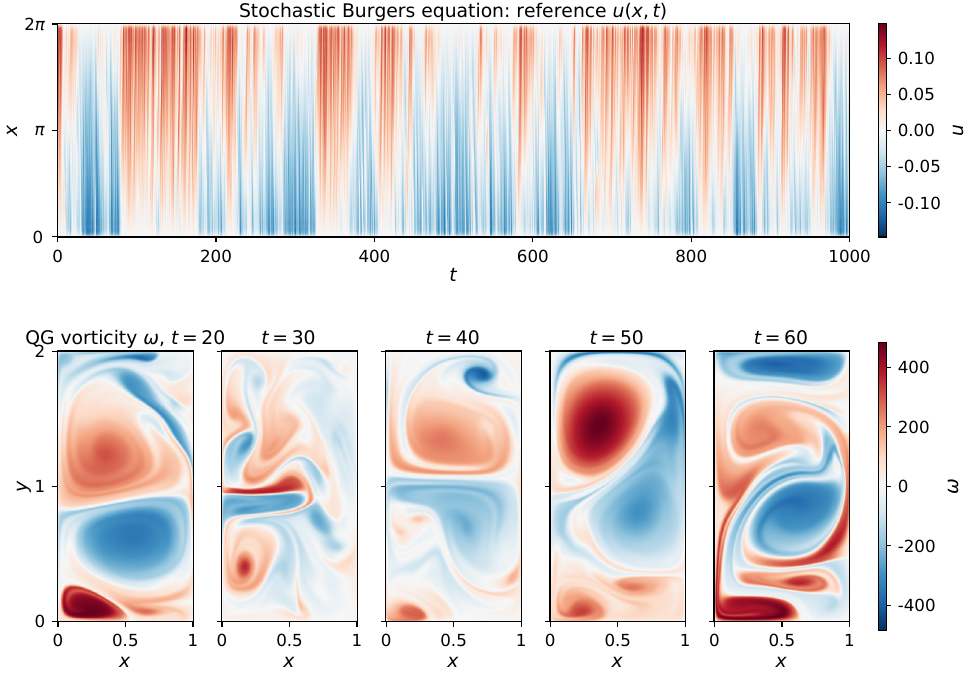}}
  \caption{Reference data for the complex PDEs (Setting II). Top:
  spatiotemporal field $u(x,t)$ of the stochastic Burgers reference
  solution{, reconstructed from the projected full-order snapshots}. Bottom: vorticity snapshots of the quasi-geostrophic reference
  solution at $t=20,30,40,50,60${; both rows use a symmetric diverging
  color scale}.}
  \label{fig:test-models-pde}
\end{figure}

\subsubsection{One-dimensional stochastic Burgers equation}
\label{sec:sbe-pde}
The first PDE is the viscous stochastic Burgers equation on
$x\in(0,L)$,
\begin{equation}
  \mathrm du=
  \left(\nu u_{xx}+{\kappa u}-\gamma u u_x\right)\mathrm dt
  +\sum_{k=1}^{m_\xi}\sigma_k e_k(x)\,\mathrm dW_k(t),
  \qquad u(0,t)=u(L,t)=0,
  \label{eq:stochastic-burgers-pde}
\end{equation}
with the orthonormal sine functions
$e_k(x)=\sqrt{2/L}\sin(k\pi x/L)$.  We use $L=2\pi$, $\nu=0.005$,
{$\kappa=0.01125$}, $\gamma=1$, $m_\xi=4$, and
$\sigma_k=0.003$ for $k=1,\ldots,4$.  The FOM step is $10^{-3}$ and data is saved every $50$ time steps, i.e., $\Delta t=0.05$.  The initial condition uses the sine modes one and four with amplitude $0.1$ as in \cite{mou2023multiscale}.

Projection onto the first $r$ sine modes yields the G-ROM:
\begin{equation}
  \mathrm d\bm a=
  \left[\mathcal Q_{\mathcal B_r}(\bm a)
  +\bm C_r\bm a\right]\mathrm dt
  +\bm\Sigma_r\,\mathrm d\bm W_t,
  \label{eq:sbe-rom}
\end{equation}
where G-ROM uses {$\bm C_r=\kappa\bm I-\nu\bm A$}\, with
$A_{ii}=(i\pi/L)^2$, \yrev{and the Galerkin-projected quadratic tensor as $\mathcal B_r$.}
OpInf and SPOpInf identify the quadratic-linear reduced drift from the
projected trajectory. Deterministic G-ROM and OpInf set $\bm\Sigma_r=0$.
G-ROM $(\sigma)$ uses the projected diffusion values, whereas OpInf
$(\sigma)$ estimates a diagonal diffusion from its regression residual.
{Because the sine basis is orthonormal and the forcing acts on the first
$m_\xi=4$ modes with known amplitudes, the projected diffusion is
simply $\bm\Sigma_r=\operatorname{diag}(\sigma_1,\ldots,\sigma_{\min(r,4)},
0,\ldots,0)$ with $\sigma_k=0.003$: at $r=2$ the G-ROM carries only the first
two forced modes, so $\bm\Sigma_r=0.003\,\bm I_2$ and the remaining two forced
modes are simply outside the reduced space; at $r=6$ the last two components
are unforced and $\bm\Sigma_r$ has two zeros.}
SPOpInf jointly estimates its drift and diagonal diffusion.  Every
stochastic ROM, i.e.,{ G-ROM\,$(\sigma)$, OpInf\,$(\sigma)$, and \spopinf{},} is driven by the exact Wiener-increment path saved with the FOM,
so pathwise differences cannot be attributed to different noise realizations.

\Cref{tab:sbe-sigma} compares the inferred amplitudes with this known
truth.  \spopinf{} recovers every forcing amplitude $\sigma_k$ to three
significant figures at all three reduced dimensions and, at $r=6$, correctly identifies the two unforced modes: the amplitudes it assigns them are two orders of magnitude below the forced ones.  The relative error $e_\sigma$
\cref{eq:sigma-error} grows with $r$, roughly doubling from $r=2$ to $r=4$ and tripling again from $r=4$ to $r=6$, because each added mode contributes its own small misfit and, at $r=6$, the two unforced modes enter the error as nonzero estimates of a zero truth.  

{\color{v9blue}
\begin{table}[H]
\centering
\caption{\nrev{Stochastic Burgers: \spopinf{} diffusion amplitudes against
the known projected truth $\sigma_k=0.003$ for the $m_\xi=4$ forced sine modes
and $0$ otherwise.  $e_\sigma$ is \cref{eq:sigma-error}.}}
\label{tab:sbe-sigma}
\small
\begin{tabular}{clll}
\toprule
$r$ & $\widehat{\bm\sigma}$ (\spopinf{}) & $\bm\sigma$ (truth) & $e_\sigma$ \\
\midrule
2 & $(0.00300,\,0.00300)$ & $0.003\,\bm 1_2$ & $1.60\times10^{-3}$ \\
4 & $(0.00299,\,0.00300,\,0.00299,\,0.00301)$ & $0.003\,\bm 1_4$
  & $2.71\times10^{-3}$ \\
6 & $(0.00299,\,0.00300,\,0.00299,\,0.00301,$
  & $(0.003\,\bm 1_4,\,0,\,0)$ & $8.06\times10^{-3}$ \\
  & $\ \ 2.35\times10^{-5},\,3.98\times10^{-5})$ & & \\
\bottomrule
\end{tabular}
\end{table}
}

For $r=2,4,6$, training ends at $T=500$ and propagation continues to
$T=1000$.  We report the \yrev{long-time} reduced-state error
\begin{equation}
  e_{L^2(0,T)}=
  \frac{\left(\int_0^T
  \|\bm a_{\rm ROM}(t)-\bm a_{\rm FOM}(t)\|_2^2\,\mathrm dt\right)^{1/2}}
  {\left(\int_0^T\|\bm a_{\rm FOM}(t)\|_2^2\,\mathrm dt\right)^{1/2}}.
  \label{eq:sbe-l2-error}
\end{equation}
A trajectory is classified as unstable if, at any time in the horizon
$0\le t\le1000$, its modal amplitude leaves the bounded region occupied by the projected full-order solution; its error is then undefined and is
marked by a dash in \cref{tab:sbe-pde-errors-full}.
Long-time densities use the complete interval $50\leq t\leq1000$ and therefore exclude any model classified as unstable, since it does not cover that entire window.

{\color{v9blue}
\begin{table}[!htbp]
\centering
\caption{\nrev{Stochastic Burgers pathwise comparison using the FOM Wiener
path: \yrev{long-time} relative $L^2$ error \cref{eq:sbe-l2-error} at each reduced
dimension $r$.  \yrev{A dash marks a trajectory that becomes unstable before $T=1000$}, so the \yrev{long-time} error is
undefined.  Best value in each column in bold.}}
\label{tab:sbe-pde-errors-full}
\small
\begin{tabular}{lccc}
\toprule
Model & $r=2$ & $r=4$ & $r=6$ \\
\midrule
G-ROM             & ---               & ---               & $1.4845$ \\
G-ROM $(\sigma)$  & ---               & ---               & $1.2875$ \\
OpInf             & $1.2043$          & $1.2352$          & --- \\
OpInf $(\sigma)$  & ---               & ---               & --- \\
SPOpInf           & $\mathbf{1.0288}$ & $\mathbf{0.4983}$ & $\mathbf{1.0966}$ \\
\bottomrule
\end{tabular}
\end{table}
}

\begin{figure}[H]
  \centering
  {\includegraphics[width=.8\textwidth]{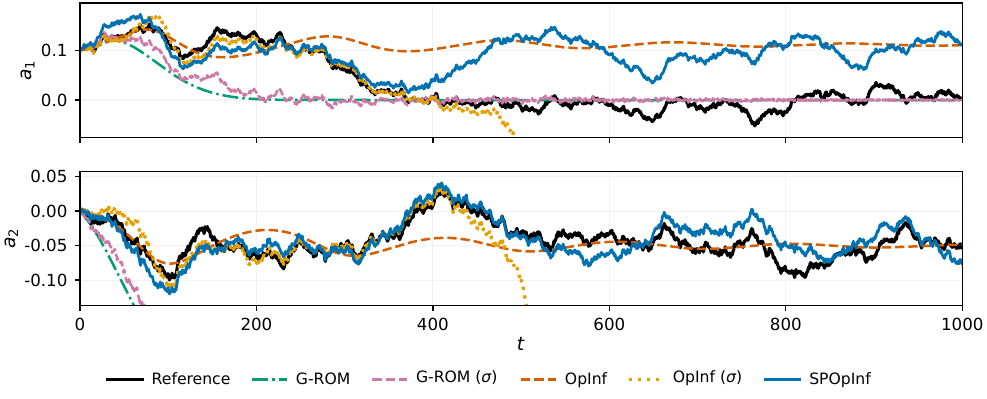}}

  \caption{Stochastic Burgers modal trajectories at $r=2$ over the complete
  interval $0\leq t\leq1000$.  {The vertical range of each panel is set
  by the reference and the ROMs that remain bounded through $T=1000$;
  divergent ROMs leave the panel\yrev{ and are marked by a dash in \cref{tab:sbe-pde-errors-full}}.}}
  \label{fig:sbe-r02-time-series}
\end{figure}

\begin{figure}[H]
  \centering
  {\includegraphics[width=.8\textwidth]{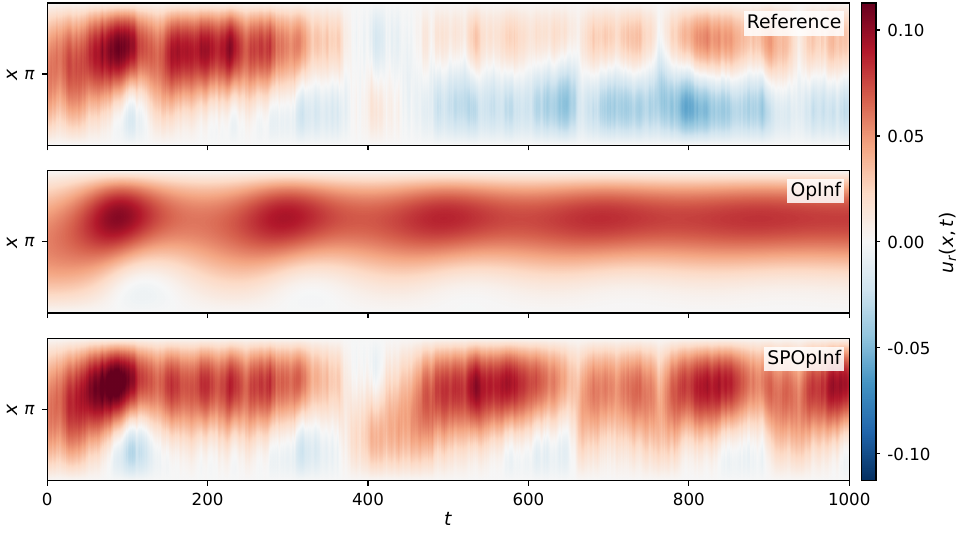}}
  \caption{Stochastic Burgers spatiotemporal reconstruction at $r=2$ over
  $0\leq t\leq1000$.  Only the projected FOM and ROMs that remain bounded over
  the full horizon are shown, using the symmetric color scale computed from
  the FOM reconstruction{ and a single shared colorbar}.}
  \label{fig:sbe-r02-spatiotemporal}
\end{figure}

\begin{figure}[H]
  \centering
  {\includegraphics[width=.8\textwidth]{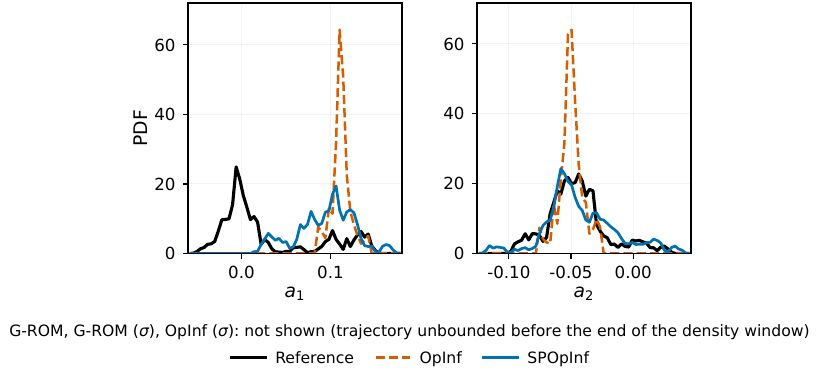}}
  \caption{Long-time marginal densities of the two retained Burgers modes at
  $r=2$, estimated over the complete interval $50\leq t\leq1000${.
  Only the models that remain bounded over this window are drawn; the legend
  names the excluded models}.}
  \label{fig:sbe-r02-pdfs}
\end{figure}

\begin{figure}[H]
  \centering
  {\includegraphics[width=.8\textwidth]{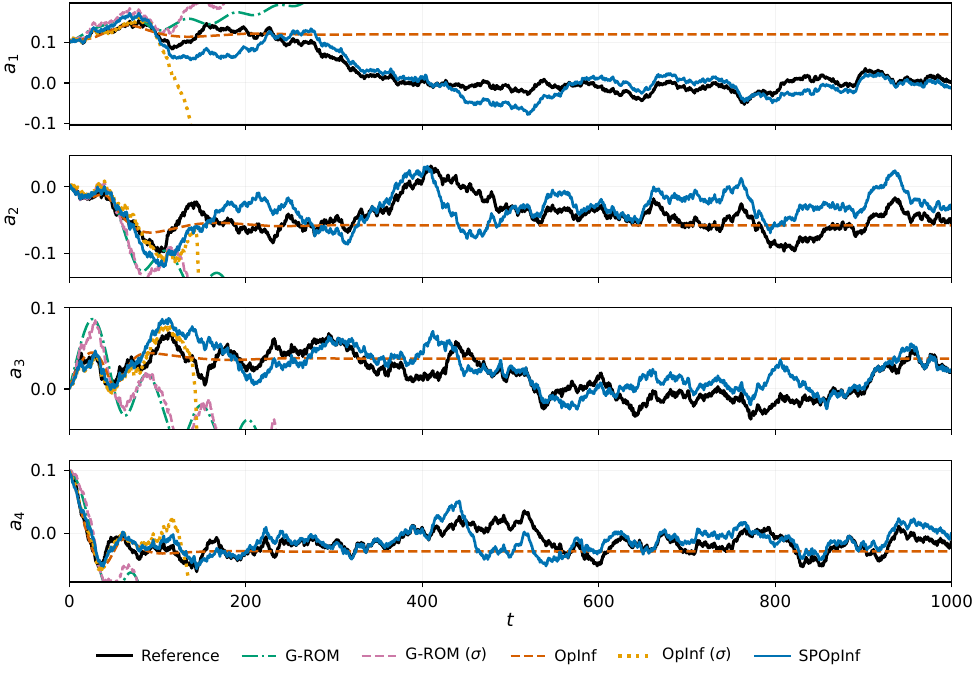}}
  \caption{Stochastic Burgers modal trajectories at $r=4$ over the complete
  interval $0\leq t\leq1000$.  {The vertical range of each panel is set
  by the reference and the ROMs that remain bounded through $T=1000$;
  divergent ROMs leave the panels\yrev{ and are marked by a dash in \cref{tab:sbe-pde-errors-full}}.}}
  \label{fig:sbe-r04-time-series}
\end{figure}

\begin{figure}[H]
  \centering
  {\includegraphics[width=.8\textwidth]{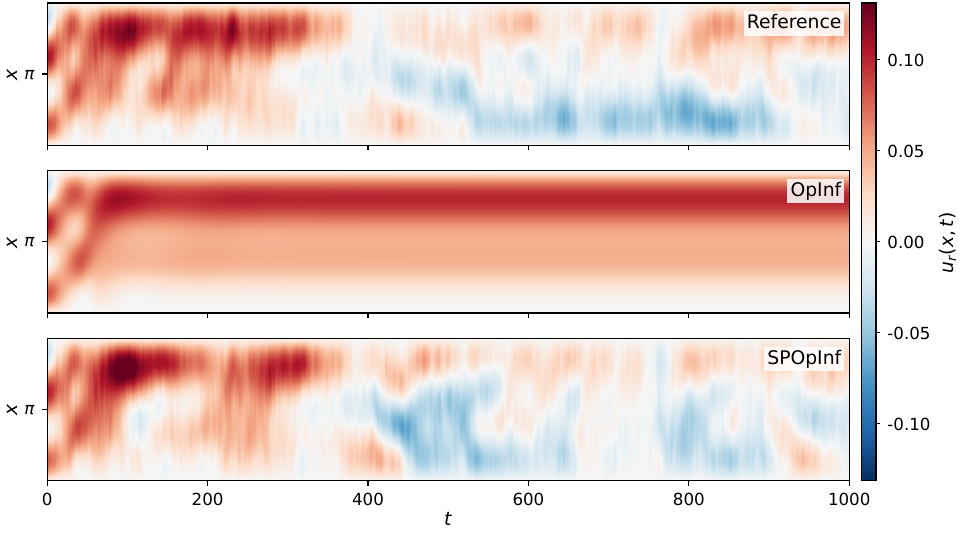}}
  \caption{Stochastic Burgers spatiotemporal reconstruction at $r=4$ over
  $0\leq t\leq1000$.  Only the projected FOM and ROMs that remain bounded over
  the full horizon are shown.  All three panels use the symmetric color scale
  computed from the FOM reconstruction{ and a single shared colorbar}.}
  \label{fig:sbe-r04-spatiotemporal}
\end{figure}

\begin{figure}[H]
  \centering
  {\includegraphics[width=.8\textwidth]{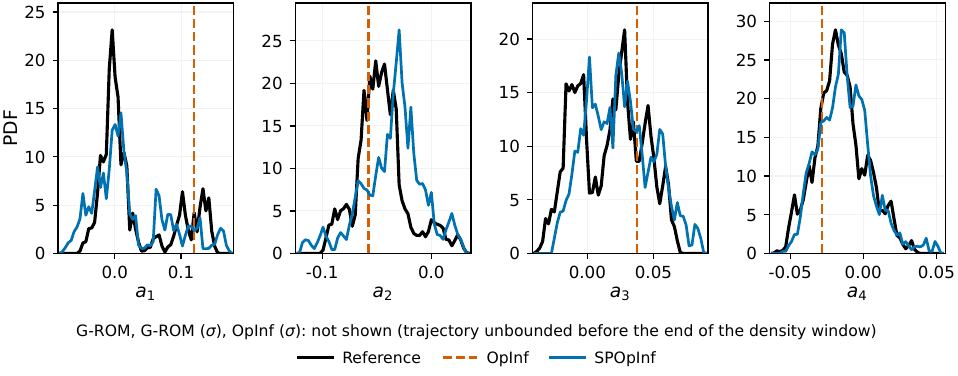}}
  \caption{Long-time marginal densities of the four retained Burgers modes at
  $r=4$, estimated from the complete interval $50\leq t\leq1000$.  Curves are
  drawn only for models that remain bounded over the full density window; the
  common legend {names} the excluded {models}.
  {Over this window the deterministic OpInf model has settled on a fixed
  point, so its density is a delta and is drawn as a vertical line at the
  steady value.}}
  \label{fig:sbe-r04-pdfs}
\end{figure}
{\Cref{tab:sbe-pde-errors-full} reports the long-time relative $L^2$
error of each model at $r=2,4,6$, with a dash for every model that becomes
unstable within the time interval.  SPOpInf is the only inferred model,
deterministic or stochastic, that remains stable over the whole interval at all three  ROM dimensions; at $r=6$ the only stable competitors are the two
Galerkin ROMs.  Among the stable models, SPOpInf also has the smallest long-time  error in every column.  
The improvement is largest at $r=4$, where SPOpInf reduces the error of deterministic OpInf, the only other stable model at that dimension, \yrev{by more than a factor of two}; at $r=2$ the reduction over
deterministic OpInf is about $15\%$, and at $r=6$ SPOpInf improves on the
better G-ROM $(\sigma)$ by a similar \yrev{margin}.  Adding
residual-calibrated noise to unconstrained OpInf does not stabilize the drift: OpInf $(\sigma)$ is unstable at every ROM dimension.  This
comparison isolates the distinction of \cref{sec:opinf-contrast} in the ROM
setting: OpInf $(\sigma)$ shares the library and the noise model with SPOpInf but lacks the constraint and the joint calibration.  The modal,
spatiotemporal, and density comparisons in
\cref{fig:sbe-r02-time-series,fig:sbe-r02-spatiotemporal,fig:sbe-r02-pdfs,fig:sbe-r04-time-series,fig:sbe-r04-spatiotemporal,fig:sbe-r04-pdfs}
show that SPOpInf remains stable and retains the broad switching pattern of
the reference, while also revealing remaining distributional bias.
}

\subsubsection{Two-dimensional quasi-geostrophic equations}
\label{sec:qg-pde}
The second PDE test is the quasi-geostrophic equations (QGE), also known as
the barotropic vorticity equations, a simplified model for geophysical flows
in which rotation plays a central role, such as wind-driven ocean
circulation in mid-latitude ocean basins
\citep{mou2020reduced,majdawang2006}.  In the streamfunction--vorticity
formulation, the QGE read
\begin{align}
  \frac{\partial\omega}{\partial t}+J(\omega,\psi)
  -\mathrm{Ro}^{-1}\,\frac{\partial\psi}{\partial x}
  &=\mathrm{Re}^{-1}\Delta\omega
  +\mathrm{Ro}^{-1}F_e,
  \label{eq:qg-pde}\\
  \omega&=-\Delta\psi,
\end{align}
where $\psi$ is the streamfunction, $\omega$ is the vorticity,
$J(\omega,\psi)=\frac{\partial\omega}{\partial x}
\frac{\partial\psi}{\partial y}
-\frac{\partial\omega}{\partial y}\frac{\partial\psi}{\partial x}$ is the
Jacobian, $F_e$ is the external forcing, $\mathrm{Re}$ is the Reynolds
number, and $\mathrm{Ro}$ is the Rossby number, which quantifies the
rotation effects.  Following \citet{mou2020reduced,mou2023multiscale}, we
consider the rectangular domain $\Omega=[0,1]\times[0,2]$ with homogeneous
Dirichlet boundary conditions for both $\psi$ and $\omega$, the symmetric
double-gyre wind forcing $F_e=\sin\!\big(\pi(y-1)\big)$, and
$\mathrm{Re}=450$, $\mathrm{Ro}=0.0036$.  Although a double-gyre wind
forcing is used, the long-term time average yields a \emph{four-gyre}
circulation pattern \citep{greatbatch2000}, and the relatively low Rossby
number produces a sharp western boundary layer, which makes this four-gyre
test problem challenging for both FOM and ROM simulations
\citep{mou2020reduced}.  The FOM uses a pseudospectral discretization with a
$257\times513$ spatial resolution; the ROM time step is $10^{-3}$, states
are saved every $0.01$ time units, and the comparison interval is
$10\leq t\leq80$, over which the flow is statistically steady but chaotic \citep{mou2020reduced,mou2021ddvms,mou2023multiscale}.
The POD basis is generated from 3500 equally spaced vorticity snapshots.

The ROM construction follows the Galerkin ROM (G-ROM) framework for the QGE
\citep{mou2020reduced}.  POD applied to the vorticity snapshots yields the
orthonormal vorticity basis $\{\varphi_i\}_{i=1}^{r}$, and the
streamfunction basis $\{\phi_i\}_{i=1}^{r}$ is defined by the Poisson
problems $-\Delta\phi_i=\varphi_i$ with homogeneous Dirichlet boundary
conditions, so that the ROM approximations
$\omega_r=\sum_{i=1}^{r}a_i(t)\,\varphi_i(\bm x)$ and
$\psi_r=\sum_{i=1}^{r}a_i(t)\,\phi_i(\bm x)$ share the same coefficients.
The componentwise G-ROM is: for $i=1,\ldots,r$,
\begin{equation}
  \dot a_i(t)={D_{G,i}}
  +\sum_{m=1}^{r}{(\bm C_G)_{im}}\,a_m(t)
  +\sum_{m=1}^{r}\sum_{n=1}^{r}{(\bm B_{G,i})_{mn}}\,a_m(t)\,a_n(t),
  \label{eq:qg-grom}
\end{equation}
where
\begin{equation}
  {D_{G,i}}=\mathrm{Ro}^{-1}\big(F_e,\varphi_i\big),
  \qquad
  {(\bm C_G)_{im}}=\mathrm{Ro}^{-1}
  \Big(\tfrac{\partial\phi_m}{\partial x},\varphi_i\Big)
  -\mathrm{Re}^{-1}\big(\nabla\varphi_m,\nabla\varphi_i\big),
  \qquad
  {(\bm B_{G,i})_{mn}}=-\big(J(\varphi_m,\phi_n),\varphi_i\big),
  \label{eq:qg-grom-operators}
\end{equation}

and $(\cdot,\cdot)$ denotes the $L^2$ inner product over $\Omega$.
\spopinf{} retains the projected forcing {$\bm D_G=(D_{G,i})$} as the known {part of the drift, $\bm f_{\rm known}=\bm D_G$ as in
\cref{sec:pde-applications},} and learns the complete quadratic-linear reduced drift and the diagonal stochastic amplitude from the projected vorticity trajectory; the training data and solver
settings are listed in the Settings paragraph below.
{For the ROM trajectories we report the \yrev{long-time} relative error
\begin{equation}
  e_{\rm traj}
  =\frac{\Big(\sum_{n}\|\bm a_{\rm ROM}(t_n)-\bm a_{\rm FOM}(t_n)\|_2^2\Big)^{1/2}}
        {\Big(\sum_{n}\|\bm a_{\rm FOM}(t_n)\|_2^2\Big)^{1/2}},
  \label{eq:qg-trajectory-error}
\end{equation}
the discrete analogue of \cref{eq:sbe-l2-error}, over $10\le t_n\le80$.}

\paragraph{Settings.}
The computational setting follows \citet{mou2020data,mou2023multiscale}: the
double-gyre forcing, the homogeneous Dirichlet boundary conditions, and the
parameters \yrev{$\mathrm{Re}=450$ and $\mathrm{Ro}=0.0036$} are those used there, the FOM is a spectral
solution on a $257\times513$ grid recorded every $0.01$ time units, and the
POD basis is built from the vorticity snapshots \yrev{described above}, with the
streamfunction basis obtained by solving the Poisson problem for each
vorticity mode.  The training data for \opinf{} and \spopinf{} are $400$
adjacent snapshot pairs $(\bm a_n,\bm a_{n+1})$ whose left members are spread
uniformly over the training interval $[10,50]$, so the step entering
\cref{eq:forward-response} is the snapshot spacing $\Delta t=0.01$; the number
of pairs is the same at every $r$.  For \spopinf{}, we set
$\yrev{\bm\Gamma}=10^{-8}\yrev{\bm I}$, $\varepsilon_{\rm tol}=10^{-6}$, and at most $20$
alternating updates, reduced to $10$ at $r=25$; at $r=30$ the normal
equations are numerically singular, with reciprocal condition numbers below
$10^{-16}$, so that case rescales the state by its global root-mean-square
value ($33.71$), uses $\yrev{\bm\Gamma}=10^{-1}\yrev{\bm I}$, and performs a single update.

All ROMs are integrated on $[10,80]$ from the projected FOM state at $t=10$
with the time step $10^{-3}$, and the reduced state is stored every ten steps
to match the FOM sampling rate.  The deterministic G-ROM and \opinf{} models
are advanced by the fourth-order Runge--Kutta scheme (RK4), as in
\citet{mou2020data,mou2023multiscale}.  \spopinf{}, whose model carries a
diffusion term, uses the same time step: the drift is advanced by RK4 and the
Gaussian increment \yrev{$\widehat{\bm\Sigma}\sqrt{\delta t}\,\bm\xi_n$} is added at each step\yrev{, where $\delta t=10^{-3}$ is the ROM time step and $\bm\xi_n$ is a standard normal vector}, with the increment path fixed at each $r$ by the seed $2026+r$.
A trajectory is \yrev{classified} as divergent when it becomes non-finite or $\|\bm a(t)\|_2>10^{16}$.

Because the QGE is chaotic, trajectories of different models diverge quickly
and pathwise agreement cannot be expected.  The trajectory error
\cref{eq:qg-trajectory-error} in \cref{tab:qg-pathwise-errors} therefore
serves only as a stability diagnostic: a model is useful if its error remains
at the level set by chaotic decorrelation rather than growing without bound.
The primary metric is statistical.  Following
\yrev{\citet{mou2020data,mou2023multiscale}}, we use the $L^2$ error of the
time-averaged ROM streamfunction over $[10,80]$, which tests the models on
the full interval although they are trained on $[10,50]$ only,
\begin{equation}
  \mathcal E(L^2)=
  \big\|\overline{\psi^{\rm FOM}}(\bm x)
  -\overline{\psi^{\rm ROM}}(\bm x)\big\|_{L^2},
  \qquad
  \overline{\psi}=\frac{1}{T_1-T_0}
  \int_{T_0}^{T_1}\psi(t)\,\mathrm dt,
  \label{eq:qg-mean-error}
\end{equation}
with $[T_0,T_1]=[10,80]$.  This metric is appropriate because successful
ROMs for this test problem are expected to reproduce the four-gyre
structure of the time-averaged streamfunction with the correct magnitude
\citep{greatbatch2000,mou2020reduced}.

\begin{table}[!htbp]
\centering
\caption{QG pathwise comparison over $10\leq t\leq80$: \yrev{long-time}
reduced-state relative  error {\cref{eq:qg-trajectory-error}}
for each
reduced dimension $r$.  \xrev{A dash marks a trajectory that diverges
(becomes non-finite or exceeds $\|\bm a\|_2=10^{16}$)} before the final time $T=80$, so the \yrev{long-time} error is
undefined.  Best value in each row in bold.
\xrev{G-ROM\,$(\sigma)$ is not included, since the full-order model
\cref{eq:qg-pde} is deterministic and there is no physical noise to
project.}}
\label{tab:qg-pathwise-errors}
\small
\begin{tabular}{rcccc}
\toprule
$r$ & G-ROM &  OpInf & OpInf $(\sigma)$ & SPOpInf \\
\midrule
 5 & $32.3530$ & $1.6952$ & \xrev{---} & $\mathbf{1.1870}$ \\
10 & $15.6463$ & \xrev{---} & \xrev{---} & $\mathbf{1.1056}$ \\
15 & $8.7633$  & \xrev{---} & \xrev{---} & $\mathbf{1.0596}$ \\
20 & $2.0168$  & \xrev{---} & \xrev{---} & $\mathbf{1.0648}$ \\
25 & $1.7169$  & \xrev{---} & \xrev{---} & $\mathbf{1.0390}$ \\
30 & $1.5223$  & \xrev{---} & \xrev{---} & $\mathbf{1.0844}$ \\
\bottomrule
\end{tabular}
\end{table}
\Cref{tab:qg-pathwise-errors} shows that \yrev{\opinf{} is unstable beyond $r=5$ and \opinf{}\,$(\sigma)$ at every $r$}, whereas \spopinf{} remains stable at every $r$,
with $e_{\rm traj}$ close to one, the level set by chaotic decorrelation of
two trajectories with matched statistics.  The G-ROM is also stable at every
$r$, but at small $r$ its trajectory error lies far above that level ($32.4$
at $r=5$); this reflects a drift of the mean state rather than decorrelation,
the same defect that \cref{tab:qg-pde-errors} measures directly.

\Cref{tab:qg-pde-errors} therefore compares \spopinf{} with the G-ROM, the
only other model that remains stable over $[10,80]$ at every $r$: \opinf{} is
stable only at $r=5$ and \opinf{}\,$(\sigma)$ at no $r$
(\cref{tab:qg-pathwise-errors}), and G-ROM\,$(\sigma)$ is not defined for
this deterministic full-order model.
\begin{table}[!htbp]
\centering
\caption{QG time-averaged streamfunction errors \cref{eq:qg-mean-error}
over $10\leq t\leq80$.}
\label{tab:qg-pde-errors}
\small
\begin{tabular}{rcc}
\toprule
$r$ & G-ROM & SPOpInf \\
\midrule
 5 & $99.6506$ & $\mathbf{0.5183}$ \\
10 & $50.4060$ & $\mathbf{0.9725}$ \\
15 & $39.2524$ & $\mathbf{0.8656}$ \\
20 & $7.5923$  & $\mathbf{0.8647}$ \\
25 & $3.6814$  & $\mathbf{1.1333}$ \\
30 & $1.7693$  & $\mathbf{0.9602}$ \\
\bottomrule
\end{tabular}
\end{table}

\begin{figure}[H]
  \centering
  {\includegraphics[width=\textwidth]{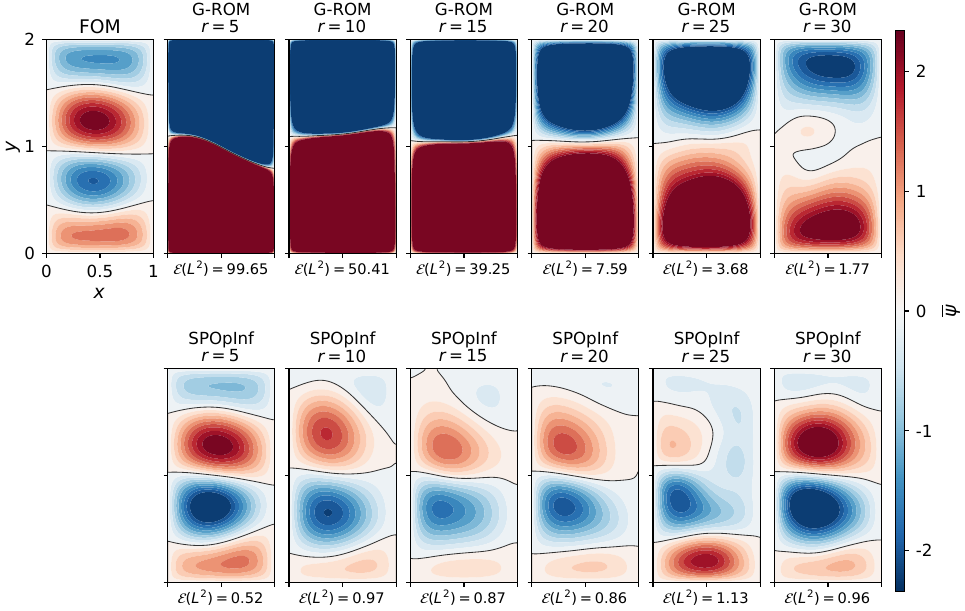}}

  \caption{Time-averaged QG streamfunction, $\overline\psi$, for different $r$ values for
  the FOM, G-ROM, and SPOpInf.  {All panels share the symmetric color
  scale of the FOM field and a single colorbar, so magnitude errors are
  visible directly; the thin contour marks $\overline\psi=0$, and the
  mean-streamfunction error \cref{eq:qg-mean-error} of
  \cref{tab:qg-pde-errors} is printed under each ROM panel.}}
  \label{fig:qg-mean-streamfunction}
\end{figure}

\begin{remark}[Conserved quadratic invariant]
\label{rem:qg-invariant}
As noted in \cref{sec:constraints}, the constraint \cref{eq:constraint},
imposed on the vorticity coefficients, conserves the resolved enstrophy
$\tfrac12\|\bm a\|_2^2=\tfrac12\|\omega_r\|_{L^2}^2$ rather than the kinetic
energy $\tfrac12(\psi_r,\omega_r)$.  This is the appropriate invariant for
the QGE: the Jacobian $J(\omega,\psi)$ conserves the enstrophy in the
continuum, and the projected Galerkin tensor assembled from
\cref{eq:qg-grom-operators} inherits this property to quadrature accuracy,
which we verified by evaluating the cubic form
$\bm a^\top\mathcal Q_{\mathcal B_G}(\bm a)$ on the assembled operators.  The
same quadrature confirms that the vorticity POD basis is orthonormal, so no
weighting of the kind described in \cref{rem:weighted-energy} is needed and
\cref{eq:reduced-relations} applies verbatim.
\end{remark}

The low-dimensional G-ROM has a severe mean-field bias: at $r=5$ its error
exceeds that of \spopinf{} by roughly two orders of magnitude.  The
\spopinf{} error stays within a narrow band of order one across the entire
sweep, a level the G-ROM approaches only at $r=30$, where \spopinf{} still
reduces the G-ROM error by nearly a factor of two.  The \spopinf{} error is
not monotone in $r$ and is in fact smallest at $r=5$.  Two mechanisms act in
opposite directions.  The learned stochastic closure absorbs the truncation
error directly at each dimension, so enlarging the basis does not by itself
improve the mean field; and with the training data fixed at $400$ snapshot
pairs while the size of the quadratic library grows with $r$, the regression
becomes progressively less well conditioned, until at $r=30$ the normal
equations are numerically singular (see the Settings paragraph) and the mean
field is shaped partly by the ridge term rather than by the data alone.
Neither mechanism dominates, which is why the error varies within the band
instead of decreasing.  Whether the non-monotonicity persists when the
training data grow with $r$ has not been tested.  The time-averaged
streamfunction fields in \cref{fig:qg-mean-streamfunction} show that the
\spopinf{} gain concerns the long-time four-gyre spatial structure rather
than synchronization of individual chaotic trajectories.

\subsection{Sensitivity tests}
\label{sec:robustness}
Except for the CDV constraint study, the results of \cref{sec:ident-sds} are each obtained from a single training trajectory, generated with one
realization of the driving noise, sampled at one time interval, and observed without noise.  This subsection examines each of these choices in turn for the five Setting I benchmarks, using eight independent random seeds \yrev{per system}.  
\spopinf{} and \opinf{} use the same library of terms and differ only in the energy constraints and the joint calibration of drift and noise.
The seed-to-seed comparison in \yrev{\cref{tab:robustness}} uses the same training
intervals as \cref{sec:ident-sds}.  For the other three tests (\cref{tab:robustness-dt,tab:robustness-T,tab:robustness-obs}), 
we simulate one reference trajectory per system and seed with time step $\Delta t=0.005$ and then derive every training set from it: coarser sampling intervals by keeping every $k$-th snapshot, shorter training intervals by keeping only the
initial part, and noisy observations by adding random perturbations.  Each
test therefore compares estimates obtained from the same underlying trajectory.  The results for the sampling and training intervals are also
plotted in \cref{fig:sensitivity-dt,fig:sensitivity-T}.

\paragraph{Seed-to-seed variability.}
\Cref{tab:robustness} repeats the experiments of \cref{sec:ident-sds} with
eight independent realizations of the driving noise.  For each realization,
\opinf{} and \spopinf{} are trained on the same trajectory, so any
difference between them comes from the method and not from the data.  The
\spopinf{} errors vary noticeably from seed to seed, most for the dyad,
whose bursty trajectory contains the fewest informative events, less for
Lorenz--84, and least for the other three systems.  This variation does not
change the comparison with \opinf{}: \spopinf{} has the smaller quadratic
error $e_{\mathcal B}$ for every seed of every system, and the smaller linear
error $e_C$ for every seed of Lorenz--84, Lorenz--96, and CDV and for most
seeds of the dyad and TBH.  The TBH exceptions are expected from
\cref{sec:sds-recovery}: the constraints act directly on $\mathcal B$ and
only indirectly on $\bm C$, and for TBH there is little for them to correct.
For every seed, \spopinf{} recovers the noise amplitude accurately and
satisfies the constraints to machine precision, whereas the \opinf{}
operators violate them clearly.  The single-trajectory results of
\cref{sec:sds-recovery} therefore give the correct ranking of the two
methods, although not always the typical error size: for the dyad, the
reported run has a quadratic error about an order of magnitude smaller than
the mean over seeds and should be read with that in mind.

\paragraph{Sampling interval.}
We keep every first, second, fourth, tenth, twentieth, and fortieth snapshot
of the same trajectory, which gives sampling intervals from $\Delta t=0.005$
to $0.2$ (\cref{tab:robustness-dt,fig:sensitivity-dt}).  The sampling
interval affects the estimate in two ways: a longer interval makes the
finite-difference approximation in \cref{eq:forward-response} less accurate,
and it changes how much of the one-step change is attributed to noise.  For
the dyad, TBH, and CDV, the quadratic error $e_{\mathcal B}$ of both methods
is nearly unchanged up to moderate sampling intervals, because it is limited
by the finite length of the record rather than by the finite difference.
For the two chaotic Lorenz systems, whose states change fastest along a
trajectory, the finite-difference error dominates already at the second
smallest interval, and $e_{\mathcal B}$ grows steadily with $\Delta t$.  The
noise-amplitude error $e_\sigma$ grows with $\Delta t$ in every system:
slowly for CDV, roughly in proportion to $\Delta t$ for the dyad and TBH, and
by several orders of magnitude for the two Lorenz systems, because the
increment likelihood assigns the entire one-step discretization error to the
noise.  The noise amplitude is therefore more sensitive to the sampling
interval than the drift, and for chaotic systems both require the finest
available sampling; the interval used in \cref{sec:ident-sds} is the
smallest tested and gives the smallest errors throughout.

\opinf{} depends on $\Delta t$ in the same way, and its quadratic error is
larger than that of \spopinf{} for almost every combination of system, seed,
and sampling interval; the only exceptions are one dyad seed at the two
coarsest intervals.  The gap between the two methods narrows as $\Delta t$
grows, since the discretization error common to both then dominates, except
for Lorenz--84, where the gap widens.  The linear operator behaves
differently.  For Lorenz--84 at every interval above the finest, and for
Lorenz--96 at the coarsest, \spopinf{} has the larger linear error $e_C$ for
every seed.  This is consistent with the energy constraints keeping the
finite-difference error out of $\mathcal B$, so that part of it is absorbed
by $\bm C$ instead; at the finest interval the effect disappears and
\spopinf{} has the smaller $e_C$ for every seed of both systems.

\paragraph{Training interval.}
We train on the first $T$ time units of each trajectory, after the initial
transient of \cref{sec:ident-sds} is discarded, for $T$ between $100$ and
$1000$ and a sampling interval of $\Delta t=0.005$
(\cref{tab:robustness-T,fig:sensitivity-T}).  The quadratic error
$e_{\mathcal B}$ of both methods decreases steadily as $T$ grows, in every
system and without \yrev{leveling} off, at a rate comparable to the reference line
$T^{-1/2}$ in \cref{fig:sensitivity-T}; the dyad shows the largest
fluctuations because of its intermittent bursts.  The systems differ mainly
in how much data they need.  The dyad and TBH reach small errors on the
shorter records, whereas CDV, whose smallest coefficients are very small,
still has a sizeable error on the longest record tested, which is why its
training interval in \cref{sec:ident-sds} is much longer than for the other
systems.  Lorenz--84, trained on the shortest interval in
\cref{sec:ident-sds}, benefits from more data as well, but its error is
already small on that interval.  The advantage over \opinf{} does not depend
on $T$: \opinf{} decreases at nearly the same rate, so the gap between the
two methods stays roughly constant, largest for the dyad and moderate for the
other four systems.  \spopinf{} has the smaller quadratic error for all but
one combination of system, seed, and $T$, and the smaller linear error $e_C$
for most of them.  The noise-amplitude error $e_\sigma$ is nearly the same
for all five systems and follows the standard error expected when a noise
amplitude is estimated from $T/\Delta t$ independent Gaussian increments
(last row of \cref{tab:robustness-T}).  With the sampling interval used here,
the accuracy of the noise amplitude is therefore determined mainly by the
length of the record.

\paragraph{Observation noise.}
We replace each snapshot by $a_{n,i}^{\rm obs}=a_{n,i}+\eta_{n,i}$, where the
errors $\eta_{n,i}\sim\mathcal N(0,\sigma_{{\rm obs},i}^2)$ are independent
across snapshots and $\sigma_{{\rm obs},i}=\varepsilon\,\mathrm{std}(a_i)$ is
proportional to the spread of each coefficient, to examine the effect of the
noise-free assumption discussed in \cref{sec:data-model}.  The quadratic
error $e_{\mathcal B}$ is essentially unchanged at the two smaller noise
levels (\cref{tab:robustness-obs}, top).  At the largest level the effect
differs between systems: the error is unchanged for TBH and Lorenz--84,
increases moderately for Lorenz--96 and the dyad, and more than doubles for
CDV, whose driving noise is by far the weakest and which is therefore the
most sensitive to added observation noise.  The noise amplitude is much more
sensitive, and in a predictable way.  Since the errors at two consecutive
snapshots are independent, each finite difference carries the variance of
both, so \cref{eq:R-update} overestimates the squared noise amplitude by
$2\sigma_{{\rm obs},i}^2/\Delta t$.  The measured estimates agree closely
with this prediction in all five systems (\cref{tab:robustness-obs},
bottom).  The prediction also explains why the sensitivity differs between
systems: it is largest where the driving noise is weak relative to the state
amplitude, as for Lorenz--84 and Lorenz--96.  The effect of observation
noise on the noise amplitude is thus a systematic and predictable bias
rather than random error, and it could be removed if $\sigma_{{\rm obs},i}$
were known or estimated together with the model, which we leave to future
work.

\begin{remark}[Incorrect constraints]
\label{rem:cdv-constraint-sets}
\Cref{tab:cdv-constraint-sets} completes the CDV constraint comparison
discussed in \cref{sec:ident-sds}.  Imposing every constraint, including the
two that the true dynamics does not satisfy, increases the quadratic error
$e_{\mathcal B}$ only slightly, and the result remains clearly more accurate
than the unconstrained fit.  In every case, the constraints that the true
dynamics does satisfy are met to rounding error.  \spopinf{} therefore
degrades gracefully when part of the constraint set is incorrect.
\end{remark}

\begin{table}[H]
\centering
\caption{Seed-to-seed variability at the settings of
\cref{sec:ident-sds}: mean $\pm$ standard deviation over eight independent
driving-noise realizations, with \opinf{} and \spopinf{} fitted to the
\emph{same} eight trajectories.
\yrev{The smaller quadratic and linear errors are in bold.}  $e_\sigma$ is
reported for \spopinf{} only, since \opinf{} does not infer $\bm\Sigma$.}
\label{tab:robustness}
\small
\setlength{\tabcolsep}{4pt}
\begin{tabular}{llccc}
\toprule
System & Method & $e_{\mathcal B}$ & $e_C$ & $e_\sigma$ \\
\midrule
Dyad & \opinf{}  & $9.10\!\times\!10^{-2}\!\pm\!3.8\!\times\!10^{-2}$ & $1.37\!\times\!10^{-1}\!\pm\!5.8\!\times\!10^{-2}$ & --- \\
     & \spopinf{} & $\mathbf{3.43\!\times\!10^{-2}}\!\pm\!1.9\!\times\!10^{-2}$ & $\mathbf{8.70\!\times\!10^{-2}}\!\pm\!5.1\!\times\!10^{-2}$ & $1.23\!\times\!10^{-3}\!\pm\!7.8\!\times\!10^{-4}$ \\
\addlinespace
TBH  & \opinf{}  & $1.09\!\times\!10^{-1}\!\pm\!7.1\!\times\!10^{-3}$ & $1.65\!\times\!10^{-1}\!\pm\!2.2\!\times\!10^{-2}$ & --- \\
     & \spopinf{} & $\mathbf{8.07\!\times\!10^{-2}}\!\pm\!6.8\!\times\!10^{-3}$ & $\mathbf{1.62\!\times\!10^{-1}}\!\pm\!2.3\!\times\!10^{-2}$ & $1.71\!\times\!10^{-3}\!\pm\!6.5\!\times\!10^{-4}$ \\
\addlinespace
L84  & \opinf{}  & $1.15\!\times\!10^{-2}\!\pm\!2.8\!\times\!10^{-3}$ & $5.16\!\times\!10^{-2}\!\pm\!1.6\!\times\!10^{-2}$ & --- \\
     & \spopinf{} & $\mathbf{7.66\!\times\!10^{-3}}\!\pm\!2.5\!\times\!10^{-3}$ & $\mathbf{3.41\!\times\!10^{-2}}\!\pm\!1.4\!\times\!10^{-2}$ & $4.75\!\times\!10^{-3}\!\pm\!1.1\!\times\!10^{-3}$ \\
\addlinespace
L96  & \opinf{}  & $5.67\!\times\!10^{-3}\!\pm\!4.7\!\times\!10^{-4}$ & $3.23\!\times\!10^{-2}\!\pm\!2.4\!\times\!10^{-3}$ & --- \\
     & \spopinf{} & $\mathbf{3.92\!\times\!10^{-3}}\!\pm\!2.0\!\times\!10^{-4}$ & $\mathbf{2.19\!\times\!10^{-2}}\!\pm\!1.5\!\times\!10^{-3}$ & $1.65\!\times\!10^{-3}\!\pm\!5.0\!\times\!10^{-4}$ \\
\addlinespace
CDV  & \opinf{}  & $2.41\!\times\!10^{-1}\!\pm\!1.3\!\times\!10^{-2}$ & $1.100\pm0.138$ & --- \\
     & \spopinf{} & $\mathbf{1.38\!\times\!10^{-1}}\!\pm\!1.9\!\times\!10^{-2}$ & $\mathbf{5.70\!\times\!10^{-1}}\!\pm\!1.4\!\times\!10^{-1}$ & $6.88\!\times\!10^{-4}\!\pm\!1.9\!\times\!10^{-4}$ \\
\bottomrule
\end{tabular}
\end{table}

\begin{table}[H]
\centering
\caption{Sampling interval.  Mean over eight seeds of the \yrev{quadratic error}
$e_{\mathcal B}$ of \opinf{} and \spopinf{} and of the \yrev{noise-amplitude error}
$e_\sigma$ of \spopinf{}, obtained by sub-sampling the same trajectories, of
length $T=2000$ for the dyad and Lorenz--84, $1000$ for TBH and Lorenz--96,
and $5000$ for CDV.  The value used in \cref{sec:ident-sds} is
$\Delta t=0.005$.  The smaller \yrev{quadratic error} of each pair is in bold.}
\label{tab:robustness-dt}
\small
\setlength{\tabcolsep}{4pt}
\begin{tabular}{llcccccc}
\toprule
& $\Delta t$ & $0.005$ & $0.01$ & $0.02$ & $0.05$ & $0.1$ & $0.2$ \\
\midrule
\multicolumn{8}{l}{\textit{\yrev{quadratic error}} $e_{\mathcal B}$}\\
Dyad & \opinf{}  & $1.17\!\times\!10^{-1}$ & $1.16\!\times\!10^{-1}$ & $1.12\!\times\!10^{-1}$ & $1.07\!\times\!10^{-1}$ & $1.16\!\times\!10^{-1}$ & $1.64\!\times\!10^{-1}$ \\
     & \spopinf{} & $\mathbf{2.64\!\times\!10^{-2}}$ & $\mathbf{2.56\!\times\!10^{-2}}$ & $\mathbf{2.39\!\times\!10^{-2}}$ & $\mathbf{2.85\!\times\!10^{-2}}$ & $\mathbf{4.71\!\times\!10^{-2}}$ & $\mathbf{1.08\!\times\!10^{-1}}$ \\
TBH  & \opinf{}  & $1.21\!\times\!10^{-1}$ & $1.21\!\times\!10^{-1}$ & $1.20\!\times\!10^{-1}$ & $1.23\!\times\!10^{-1}$ & $1.37\!\times\!10^{-1}$ & $1.97\!\times\!10^{-1}$ \\
     & \spopinf{} & $\mathbf{8.89\!\times\!10^{-2}}$ & $\mathbf{8.87\!\times\!10^{-2}}$ & $\mathbf{8.90\!\times\!10^{-2}}$ & $\mathbf{9.39\!\times\!10^{-2}}$ & $\mathbf{1.14\!\times\!10^{-1}}$ & $\mathbf{1.83\!\times\!10^{-1}}$ \\
L84  & \opinf{}  & $2.20\!\times\!10^{-3}$ & $1.75\!\times\!10^{-2}$ & $5.15\!\times\!10^{-2}$ & $1.53\!\times\!10^{-1}$ & $3.20\!\times\!10^{-1}$ & $6.25\!\times\!10^{-1}$ \\
     & \spopinf{} & $\mathbf{1.38\!\times\!10^{-3}}$ & $\mathbf{8.78\!\times\!10^{-3}}$ & $\mathbf{2.55\!\times\!10^{-2}}$ & $\mathbf{6.32\!\times\!10^{-2}}$ & $\mathbf{6.12\!\times\!10^{-2}}$ & $\mathbf{2.40\!\times\!10^{-1}}$ \\
L96  & \opinf{}  & $5.52\!\times\!10^{-3}$ & $1.48\!\times\!10^{-2}$ & $4.22\!\times\!10^{-2}$ & $1.32\!\times\!10^{-1}$ & $3.09\!\times\!10^{-1}$ & $6.89\!\times\!10^{-1}$ \\
     & \spopinf{} & $\mathbf{3.65\!\times\!10^{-3}}$ & $\mathbf{1.19\!\times\!10^{-2}}$ & $\mathbf{3.49\!\times\!10^{-2}}$ & $\mathbf{1.13\!\times\!10^{-1}}$ & $\mathbf{2.78\!\times\!10^{-1}}$ & $\mathbf{6.49\!\times\!10^{-1}}$ \\
CDV  & \opinf{}  & $2.60\!\times\!10^{-1}$ & $2.59\!\times\!10^{-1}$ & $2.59\!\times\!10^{-1}$ & $2.58\!\times\!10^{-1}$ & $2.59\!\times\!10^{-1}$ & $2.73\!\times\!10^{-1}$ \\
     & \spopinf{} & $\mathbf{1.55\!\times\!10^{-1}}$ & $\mathbf{1.54\!\times\!10^{-1}}$ & $\mathbf{1.54\!\times\!10^{-1}}$ & $\mathbf{1.55\!\times\!10^{-1}}$ & $\mathbf{1.59\!\times\!10^{-1}}$ & $\mathbf{1.79\!\times\!10^{-1}}$ \\
\addlinespace
\multicolumn{8}{l}{\textit{\yrev{noise-amplitude error}} $e_\sigma$, \spopinf{}}\\
Dyad & & $1.23\!\times\!10^{-3}$ & $1.92\!\times\!10^{-3}$ & $4.19\!\times\!10^{-3}$ & $1.07\!\times\!10^{-2}$ & $2.21\!\times\!10^{-2}$ & $4.47\!\times\!10^{-2}$ \\
TBH  & & $1.63\!\times\!10^{-3}$ & $2.60\!\times\!10^{-3}$ & $4.39\!\times\!10^{-3}$ & $1.14\!\times\!10^{-2}$ & $2.19\!\times\!10^{-2}$ & $3.83\!\times\!10^{-2}$ \\
L84  & & $1.27\!\times\!10^{-3}$ & $1.76\!\times\!10^{-3}$ & $1.15\!\times\!10^{-2}$ & $2.46\!\times\!10^{-1}$ & $1.62\!\times\!10^{0}$ & $5.45\!\times\!10^{0}$ \\
L96  & & $1.65\!\times\!10^{-3}$ & $2.32\!\times\!10^{-3}$ & $4.49\!\times\!10^{-2}$ & $7.66\!\times\!10^{-1}$ & $3.05\!\times\!10^{0}$ & $7.14\!\times\!10^{0}$ \\
CDV  & & $6.57\!\times\!10^{-4}$ & $9.81\!\times\!10^{-4}$ & $1.70\!\times\!10^{-3}$ & $2.80\!\times\!10^{-3}$ & $5.06\!\times\!10^{-3}$ & $9.37\!\times\!10^{-3}$ \\
\bottomrule
\end{tabular}
\end{table}

\begin{figure}[H]
  \centering
  \includegraphics[width=\textwidth]{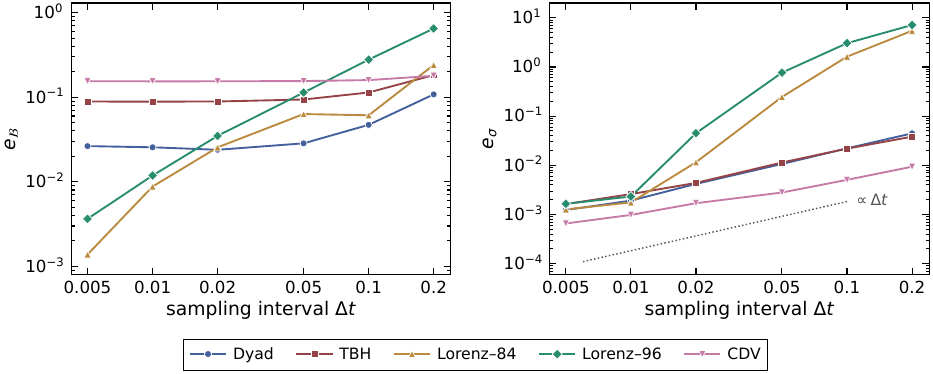}
  \caption{Sampling-interval sweep of \cref{tab:robustness-dt}: mean over
  eight seeds of the \yrev{quadratic error} $e_{\mathcal B}$ (left) and the \yrev{noise-amplitude error} $e_\sigma$ (right) of \spopinf{} against $\Delta t$, obtained by
  sub-sampling the same trajectories; the \opinf{} \yrev{quadratic error}s are listed in
  \cref{tab:robustness-dt}.  The dotted line has slope one.}
  \label{fig:sensitivity-dt}
\end{figure}

\begin{table}[H]
\centering
\caption{Training interval on a common grid.  Mean over eight seeds of the
\yrev{quadratic error} $e_{\mathcal B}$ of \opinf{} and \spopinf{} and of the \yrev{noise-amplitude error} $e_\sigma$ of \spopinf{} at $\Delta t=0.005$; each estimate uses the
first $T$ time units of one trajectory per system and seed, the same data for
both methods.  The smaller \yrev{quadratic error} of each pair is in bold.  The last
row is $(2T/\Delta t)^{-1/2}$, the relative standard error of a noise
amplitude estimated from $T/\Delta t$ Gaussian increments.}
\label{tab:robustness-T}
\small
\setlength{\tabcolsep}{3.5pt}
\begin{tabular}{llccccccc}
\toprule
& $T$ & $100$ & $150$ & $200$ & $300$ & $500$ & $700$ & $1000$ \\
\midrule
\multicolumn{9}{l}{\textit{\yrev{quadratic error}} $e_{\mathcal B}$}\\
Dyad & \opinf{}  & $4.67\!\times\!10^{-1}$ & $4.37\!\times\!10^{-1}$ & $3.50\!\times\!10^{-1}$ & $2.64\!\times\!10^{-1}$ & $2.12\!\times\!10^{-1}$ & $1.58\!\times\!10^{-1}$ & $1.26\!\times\!10^{-1}$ \\
     & \spopinf{} & $\mathbf{1.63\!\times\!10^{-1}}$ & $\mathbf{1.19\!\times\!10^{-1}}$ & $\mathbf{8.42\!\times\!10^{-2}}$ & $\mathbf{6.59\!\times\!10^{-2}}$ & $\mathbf{4.82\!\times\!10^{-2}}$ & $\mathbf{2.86\!\times\!10^{-2}}$ & $\mathbf{2.86\!\times\!10^{-2}}$ \\
TBH  & \opinf{}  & $4.62\!\times\!10^{-1}$ & $3.39\!\times\!10^{-1}$ & $2.83\!\times\!10^{-1}$ & $2.20\!\times\!10^{-1}$ & $1.63\!\times\!10^{-1}$ & $1.31\!\times\!10^{-1}$ & $1.09\!\times\!10^{-1}$ \\
     & \spopinf{} & $\mathbf{3.25\!\times\!10^{-1}}$ & $\mathbf{2.42\!\times\!10^{-1}}$ & $\mathbf{2.01\!\times\!10^{-1}}$ & $\mathbf{1.59\!\times\!10^{-1}}$ & $\mathbf{1.15\!\times\!10^{-1}}$ & $\mathbf{9.67\!\times\!10^{-2}}$ & $\mathbf{8.07\!\times\!10^{-2}}$ \\
L84  & \opinf{}  & $1.02\!\times\!10^{-2}$ & $8.71\!\times\!10^{-3}$ & $8.05\!\times\!10^{-3}$ & $6.80\!\times\!10^{-3}$ & $4.76\!\times\!10^{-3}$ & $3.94\!\times\!10^{-3}$ & $2.93\!\times\!10^{-3}$ \\
     & \spopinf{} & $\mathbf{5.56\!\times\!10^{-3}}$ & $\mathbf{4.40\!\times\!10^{-3}}$ & $\mathbf{4.88\!\times\!10^{-3}}$ & $\mathbf{4.08\!\times\!10^{-3}}$ & $\mathbf{2.61\!\times\!10^{-3}}$ & $\mathbf{2.42\!\times\!10^{-3}}$ & $\mathbf{1.96\!\times\!10^{-3}}$ \\
L96  & \opinf{}  & $1.86\!\times\!10^{-2}$ & $1.46\!\times\!10^{-2}$ & $1.25\!\times\!10^{-2}$ & $1.03\!\times\!10^{-2}$ & $7.79\!\times\!10^{-3}$ & $6.58\!\times\!10^{-3}$ & $5.67\!\times\!10^{-3}$ \\
     & \spopinf{} & $\mathbf{1.17\!\times\!10^{-2}}$ & $\mathbf{9.20\!\times\!10^{-3}}$ & $\mathbf{8.31\!\times\!10^{-3}}$ & $\mathbf{6.76\!\times\!10^{-3}}$ & $\mathbf{5.30\!\times\!10^{-3}}$ & $\mathbf{4.54\!\times\!10^{-3}}$ & $\mathbf{3.92\!\times\!10^{-3}}$ \\
CDV  & \opinf{}  & $4.08\!\times\!10^{0}$ & $2.80\!\times\!10^{0}$ & $2.17\!\times\!10^{0}$ & $1.47\!\times\!10^{0}$ & $9.71\!\times\!10^{-1}$ & $7.70\!\times\!10^{-1}$ & $5.92\!\times\!10^{-1}$ \\
     & \spopinf{} & $\mathbf{2.13\!\times\!10^{0}}$ & $\mathbf{1.42\!\times\!10^{0}}$ & $\mathbf{1.12\!\times\!10^{0}}$ & $\mathbf{7.60\!\times\!10^{-1}}$ & $\mathbf{5.41\!\times\!10^{-1}}$ & $\mathbf{4.25\!\times\!10^{-1}}$ & $\mathbf{3.40\!\times\!10^{-1}}$ \\
\addlinespace
\multicolumn{9}{l}{\textit{\yrev{noise-amplitude error}} $e_\sigma$, \spopinf{}}\\
Dyad & & $4.72\!\times\!10^{-3}$ & $5.19\!\times\!10^{-3}$ & $4.51\!\times\!10^{-3}$ & $2.68\!\times\!10^{-3}$ & $2.19\!\times\!10^{-3}$ & $1.58\!\times\!10^{-3}$ & $1.57\!\times\!10^{-3}$ \\
TBH  & & $5.22\!\times\!10^{-3}$ & $4.40\!\times\!10^{-3}$ & $3.48\!\times\!10^{-3}$ & $2.39\!\times\!10^{-3}$ & $1.89\!\times\!10^{-3}$ & $1.64\!\times\!10^{-3}$ & $1.71\!\times\!10^{-3}$ \\
L84  & & $5.07\!\times\!10^{-3}$ & $4.38\!\times\!10^{-3}$ & $3.59\!\times\!10^{-3}$ & $2.38\!\times\!10^{-3}$ & $2.06\!\times\!10^{-3}$ & $1.44\!\times\!10^{-3}$ & $1.50\!\times\!10^{-3}$ \\
L96  & & $4.31\!\times\!10^{-3}$ & $4.14\!\times\!10^{-3}$ & $3.64\!\times\!10^{-3}$ & $2.35\!\times\!10^{-3}$ & $1.83\!\times\!10^{-3}$ & $1.50\!\times\!10^{-3}$ & $1.65\!\times\!10^{-3}$ \\
CDV  & & $4.70\!\times\!10^{-3}$ & $3.81\!\times\!10^{-3}$ & $3.65\!\times\!10^{-3}$ & $3.03\!\times\!10^{-3}$ & $2.37\!\times\!10^{-3}$ & $1.91\!\times\!10^{-3}$ & $1.70\!\times\!10^{-3}$ \\
\multicolumn{2}{l}{$(2T/\Delta t)^{-1/2}$} & $5.00\!\times\!10^{-3}$ & $4.08\!\times\!10^{-3}$ & $3.54\!\times\!10^{-3}$ & $2.89\!\times\!10^{-3}$ & $2.24\!\times\!10^{-3}$ & $1.89\!\times\!10^{-3}$ & $1.58\!\times\!10^{-3}$ \\
\bottomrule
\end{tabular}
\end{table}

\begin{figure}[H]
  \centering
  \includegraphics[width=\textwidth]{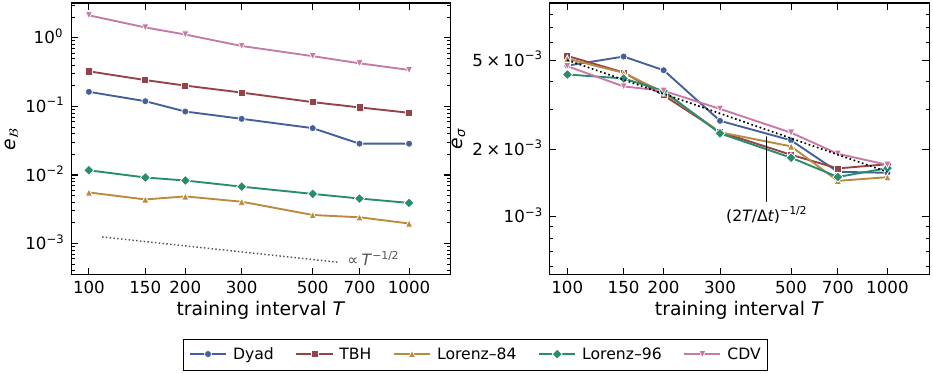}
  \caption{Training-interval sweep of \cref{tab:robustness-T}: mean over
  eight seeds of the \yrev{quadratic error} $e_{\mathcal B}$ (left) and the \yrev{noise-amplitude error} $e_\sigma$ (right) of \spopinf{} against $T$, at $\Delta t=0.005$; the
  \opinf{} \yrev{quadratic error}s are listed in \cref{tab:robustness-T}.  The dotted
  line on the left has slope $-1/2$; on the right it is
  $(2T/\Delta t)^{-1/2}$, the relative standard error of a noise amplitude
  estimated from $T/\Delta t$ Gaussian increments, which the five systems
  follow closely.}
  \label{fig:sensitivity-T}
\end{figure}

\begin{table}[H]
\centering
\caption{Observation noise of relative amplitude
$\varepsilon=\sigma_{\rm obs}/\mathrm{std}(a_i)$ (\spopinf{}).  Top: \yrev{quadratic error} as a multiple of its noise-free value.  Bottom: \yrev{measured ratio $\widehat\sigma/\sigma$, with the predicted value in parentheses.}}
\label{tab:robustness-obs}
\small
\begin{tabular}{lccccc}
\toprule
$\varepsilon$ & Dyad & TBH & L84 & L96 & CDV \\
\midrule
\multicolumn{6}{l}{\textit{\yrev{quadratic error}} $e_{\mathcal B}(\varepsilon)/e_{\mathcal B}(0)$}\\
$10^{-4}$ & $1.000$ & $1.000$ & $1.000$ & $1.000$ & $1.000$ \\
$10^{-3}$ & $1.002$ & $1.000$ & $0.999$ & $1.001$ & $1.003$ \\
$10^{-2}$ & $1.306$ & $0.998$ & $0.990$ & $1.156$ & $2.663$ \\
\addlinespace
\multicolumn{6}{l}{\textit{\yrev{noise-amplitude ratio}} $\widehat\sigma/\sigma$, measured (predicted)}\\
$10^{-4}$ & $1.000$ & $1.000$ & $1.000$ & $1.000$ & $1.000$ \\
$10^{-3}$ & $1.000$ & $1.000$ & $1.013$ & $1.011$ & $1.009$ \\
$10^{-2}$ & $1.034$ & $1.020$ & $1.922$ & $1.811$ & $1.662$ \\
          & $(1.034)$ & $(1.020)$ & $(1.920)$ & $(1.811)$ & $(1.659)$ \\
\bottomrule
\end{tabular}
\end{table}

\begin{table}[H]
\centering
\caption{Charney--DeVore under \yrev{an incorrect} constraint set: mean
$\pm$ standard deviation over four seeds.  ``pruned'' removes the two rows the
true dynamics violates (the setting used in \cref{sec:ident-sds}), ``full''
imposes all rows including those two, and ``none'' is the unconstrained
least-squares fit.  The last column is the largest residual of the
\emph{compatible} rows, normalized as in \cref{eq:structure-residual}.}
\label{tab:cdv-constraint-sets}
\small
\begin{tabular}{lccc}
\toprule
Constraint set & $e_{\mathcal B}$ & $e_C$ & $e_G$ (compatible rows) \\
\midrule
pruned (\spopinf{}) & $0.140\pm0.024$ & $0.605\pm0.168$ & $7.8\times10^{-17}$ \\
full (\yrev{incorrect}) & $0.155\pm0.022$ & $0.639\pm0.150$ & $1.8\times10^{-16}$ \\
none (\opinf{})      & $0.249\pm0.014$ & $1.172\pm0.049$ & $1.3\times10^{-1}$ \\
\bottomrule
\end{tabular}
\end{table}

\section{Conclusions}
\label{sec:conclusion}

We proposed stochastic physics-constrained operator inference (\spopinf{}),
a non-intrusive method that identifies a complete stochastic quadratic
model, drift and diffusion amplitude together, from snapshot data.
\spopinf{} differs from classical operator inference in three ways: the
deterministic and stochastic components are calibrated jointly by maximizing
an Euler--Maruyama increment likelihood in which the residual covariance is
itself an unknown; the energy-preserving structure of the quadratic tensor is
enforced exactly, as linear equality constraints inside the optimization; and
the model is parametrized in upper-triangular form, which removes the
structural non-identifiability of the full tensor.  The optimization is
solved by an expectation--maximization \yrev{(EM)-based} iterative algorithm that
alternates between updating the drift coefficients and the noise levels,
each step an inexpensive weighted least-squares solve, so the method retains
nearly the cost of ordinary least-squares regression.

We tested \spopinf{} on five stochastic systems that show intermittent
bursts, chaos, and switching between regimes.  Compared with \opinf{}, it
gave more accurate quadratic operators, satisfied the energy constraints to
machine precision, recovered the noise amplitudes accurately, and reproduced
the long-time probability densities and autocorrelations of each system.
For the stochastic Burgers equation, it was the only learned model that
remained stable over the whole time interval at every reduced dimension
tested, and it recovered the known noise amplitude of each mode, correctly
assigning near-zero values to the modes that receive no forcing.  For the
quasi-geostrophic double-gyre flow, it reproduced the four-gyre
time-averaged circulation at reduced dimensions where the Galerkin ROM
fails.  The sensitivity tests of \cref{sec:robustness} varied the noise
realization, the sampling interval, the length of the training record, and
the level of observation noise.  They showed that \spopinf{} learns the
quadratic operator more accurately than \opinf{} for almost every
combination tested; that the quadratic error decreases steadily as the
training record grows, while the error in the noise amplitude matches what
is expected from a finite number of independent increments; and that the
noise amplitude is more sensitive than the drift to both the sampling
interval and observation noise, in the latter case through a simple bias
that we predict and confirm numerically.

Several limitations point to future work.  The noise is assumed to be
additive, with a constant amplitude for each mode and no correlation between
modes; noise that is correlated across modes, depends on the state, or
changes with the flow regime would require a more general noise model.  The
likelihood is built on the Euler--Maruyama scheme and first-order finite
differences, so the learned drift carries a discretization error and the
noise amplitude depends on the sampling interval, strongly so for chaotic
systems; likelihoods based on higher-order schemes would reduce this
dependence.  The snapshots are assumed to be free of observation noise; the
sensitivity tests of \cref{sec:robustness} show that such noise biases the
noise amplitude in a predictable way, and estimating the observation-noise
level together with the model is a natural extension.  When the true
dynamics does not satisfy part of the constraint set, the constraints that
it does satisfy must be identified beforehand; the method remains accurate
when they are not, but a systematic way of selecting them is still missing.
Finally, for the quasi-geostrophic flow at the largest reduced dimension,
the fit is poorly conditioned with the amount of training data used here, so
the \yrev{behavior} of the method when the training data grow with the reduced
dimension remains open.

Beyond these limitations, three directions are of particular interest.
First, the method should be applied to more challenging problems, such as
three-dimensional turbulent flows, systems with several interacting scales,
and systems whose reduced dynamics are not well described by a quadratic
model.  Second, the method calls for numerical analysis: showing that the
constrained maximum likelihood estimator converges to the true operators and
noise amplitudes as the sampling interval decreases and the training record
grows, and deriving error bounds for the resulting reduced-order model that
account for both the discretization error and the finite-data error observed
in \cref{sec:robustness}.  Third, the learned drift and noise amplitudes
could be combined with the statistical closure framework of
\citet{qi2025data,qi2023random}, in which the reduced-order model is
calibrated to reproduce the leading statistics of the full system rather
than individual trajectories; extending this combination to parametric and
memory-dependent (non-Markovian) closures would further broaden the range of
problems the method can address.

\section*{Code availability}
The code that implements \spopinf{} and reproduces the numerical results of
this paper is publicly available at
\url{https://github.com/chhmou/SPOpinf}.  The full-order snapshot data for
the stochastic Burgers and quasi-geostrophic tests of
\cref{sec:pde-applications} are not included in the repository and are
available from the author upon reasonable request.

\section*{Declaration of AI usage}
During the preparation of this work, the author used Claude (Anthropic) to
improve the language and consistency of the manuscript.  After using this
tool, the author reviewed and edited the content as needed and takes full
responsibility for the content of the publication.
\bibliographystyle{plainnat}
\bibliography{references}

\end{document}